%% file: main_and_si.tex
\def\CombinedDocument{1}
\documentclass[
  aps,pra,onecolumn,10pt,letterpaper,showkeys,
  secnumarabic,nobibnotes,nofootinbib,superscriptaddress,longbibliography
]{revtex4-2}
\ifdefined\StandaloneSI
  \usepackage{xr-hyper}
\fi
\input{preamble}
\ifdefined\CombinedDocument
  \usepackage{chapterbib}
\fi
\hypersetup{
  pdftitle={Minimal building blocks for molecular quantum circuits with exact spin symmetry},
  pdfauthor={Mengwei Liu and Zhenyu Li},
  pdfsubject={Dynamical Lie algebras and symmetry-preserving molecular variational circuits},
  pdfkeywords={variational quantum eigensolver, dynamical Lie algebra, spin adaptation, quantum chemistry, controllability},
  hidelinks,
  hypertexnames=false
}

\ifdefined\StandaloneSI
  \hypersetup{pdftitle={Supplementary Information: Minimal building blocks for molecular quantum circuits with exact spin symmetry}}
\fi

\makeatletter
\AtBeginDocument{%
  \let\auto@bib\@empty
  \let\auto@bib@innerbib\@empty
  \ifdefined\CombinedDocument
    \let\write@bibliographystyle\relax
    \let\ArticleTableOfContents\tableofcontents
    \renewcommand{\tableofcontents}{%
      \let\SavedCitationContext\the@ipfilectr
      \ArticleTableOfContents
      \global\let\the@ipfilectr\SavedCitationContext
    }%
  \fi
}
\makeatother

\newcommand{\ArticleBibliography}{%
  \ifdefined\CombinedDocument
\input{main_and_si.bbl}  \else
    \bibliography{refs}%
  \fi
}
\makeatletter
\newcommand{\SIBibliography}{%
  \begingroup
  \let\SIOriginalLabel\label
  \renewcommand{\label}[1]{\SIOriginalLabel{si-bib-##1}}%
  \def\NAT@bibsetnum##1{\setlength{\topsep}{0pt}\NATx@bibsetnum{##1}}%
  \let\@bibsetup\NAT@bibsetnum
  \renewcommand{\bibsection}{\section*{References}}%
  \ifdefined\CombinedDocument
\input{si_references.bbl}  \else
    \bibliography{refs}%
  \fi
  \endgroup
}
\makeatother

\begin{document}
\begin{cbunit}
\input{article_body}

\clearpage
\end{cbunit}
\begin{cbunit}
\input{supplementary/si_body}
\clearpage
\end{cbunit}
\end{document}

%% file: preamble.tex
\usepackage[letterpaper,scale=0.7,vmarginratio=2:3]{geometry}
\usepackage{amsthm,amsmath,amssymb,amsfonts,bm}
\usepackage{graphicx,xcolor}
\usepackage{enumitem,booktabs}
\usepackage{mathtools,array,placeins,longtable,needspace,microtype}
\usepackage{algorithm}
\usepackage{algpseudocode}
\usepackage[caption=false]{subfig}
\usepackage[english]{babel}

\usepackage[hidelinks,pdfusetitle]{hyperref}
\usepackage[capitalize]{cleveref}

\AtBeginDocument{%
  \renewcommand{\thesection}{\Roman{section}}%
  \renewcommand{\thesubsection}{\thesection.\arabic{subsection}}%
}

\theoremstyle{plain}
\newtheorem{theorem}{Theorem}
\newtheorem{lemma}[theorem]{Lemma}
\newtheorem{proposition}[theorem]{Proposition}
\newtheorem{corollary}[theorem]{Corollary}

\newtheorem{remark}[theorem]{Remark}
\AddToHook{env/theorem/begin}{\crefalias{theorem}{theorem}}
\AddToHook{env/lemma/begin}{\crefalias{theorem}{lemma}}
\AddToHook{env/proposition/begin}{\crefalias{theorem}{proposition}}
\AddToHook{env/corollary/begin}{\crefalias{theorem}{corollary}}
\AddToHook{env/assumption/begin}{\crefalias{theorem}{assumption}}
\AddToHook{env/definition/begin}{\crefalias{theorem}{definition}}
\AddToHook{env/remark/begin}{\crefalias{theorem}{remark}}

\crefname{theorem}{Theorem}{Theorems}
\Crefname{theorem}{Theorem}{Theorems}
\crefname{lemma}{Lemma}{Lemmas}
\Crefname{lemma}{Lemma}{Lemmas}
\crefname{proposition}{Proposition}{Propositions}
\Crefname{proposition}{Proposition}{Propositions}
\crefname{corollary}{Corollary}{Corollaries}
\Crefname{corollary}{Corollary}{Corollaries}
\crefname{assumption}{Assumption}{Assumptions}
\Crefname{assumption}{Assumption}{Assumptions}
\crefname{definition}{Definition}{Definitions}
\Crefname{definition}{Definition}{Definitions}
\crefname{remark}{Remark}{Remarks}
\Crefname{remark}{Remark}{Remarks}
\crefname{notation}{Notation}{Notations}
\Crefname{notation}{Notation}{Notations}
\crefname{fact}{Fact}{Facts}
\Crefname{fact}{Fact}{Facts}
\crefname{algorithm}{Algorithm}{Algorithms}
\Crefname{algorithm}{Algorithm}{Algorithms}



\input{manuscript_macros}

\makeatletter
\AtBeginDocument{%
  \def\pdfstartlink@attr{attr{/Border[0 0 0]/H/I}}%
}
\makeatother

%% file: manuscript_macros.tex
\newcommand{\Hsp}{\mathcal H}
\newcommand{\Fsp}{\mathcal F}
\newcommand{\Lie}{\operatorname{Lie}}
\newcommand{\End}{\operatorname{End}}
\newcommand{\Hom}{\operatorname{Hom}}
\newcommand{\im}{\operatorname{im}}
\newcommand{\ad}{\operatorname{ad}}
\newcommand{\so}{\mathfrak{so}}
\newcommand{\Span}{\operatorname{span}}
\newcommand{\Lpool}{\mathcal L}
\newcommand{\Xpool}{\mathcal X}
\newcommand{\Ypool}{\mathcal Y}
\newcommand{\Ypooldisj}{\mathcal Y^{\mathrm{disj}}}
\newcommand{\R}{\mathbb R}
\newcommand{\C}{\mathbb C}
\newcommand{\e}{\mathrm e}
\newcommand{\ket}[1]{\lvert #1\rangle}
\newcommand{\bra}[1]{\langle #1\rvert}
\newcommand{\tblcell}[2]{%
  \begin{minipage}[t]{#1}\raggedright\strut #2\strut\end{minipage}}

\allowdisplaybreaks[3]
\setlist{itemsep=0.2em,topsep=0.35em}
\crefname{equation}{Eq.}{Eqs.}
\Crefname{equation}{Equation}{Equations}
\crefname{figure}{Fig.}{Figs.}
\Crefname{figure}{Fig.}{Figs.}

%% file: si_references.bbl
\begin{thebibliography}{14}%
\makeatletter
\providecommand \@ifxundefined [1]{%
 \@ifx{#1\undefined}
}%
\providecommand \@ifnum [1]{%
 \ifnum #1\expandafter \@firstoftwo
 \else \expandafter \@secondoftwo
 \fi
}%
\providecommand \@ifx [1]{%
 \ifx #1\expandafter \@firstoftwo
 \else \expandafter \@secondoftwo
 \fi
}%
\providecommand \natexlab [1]{#1}%
\providecommand \enquote  [1]{``#1''}%
\providecommand \bibnamefont  [1]{#1}%
\providecommand \bibfnamefont [1]{#1}%
\providecommand \citenamefont [1]{#1}%
\providecommand \href@noop [0]{\@secondoftwo}%
\providecommand \href [0]{\begingroup \@sanitize@url \@href}%
\providecommand \@href[1]{\@@startlink{#1}\@@href}%
\providecommand \@@href[1]{\endgroup#1\@@endlink}%
\providecommand \@sanitize@url [0]{\catcode `\\12\catcode `\$12\catcode
  `\&12\catcode `\#12\catcode `\^12\catcode `\_12\catcode `\%12\relax}%
\providecommand \@@startlink[1]{}%
\providecommand \@@endlink[0]{}%
\providecommand \url  [0]{\begingroup\@sanitize@url \@url }%
\providecommand \@url [1]{\endgroup\@href {#1}{\urlprefix }}%
\providecommand \urlprefix  [0]{URL }%
\providecommand \Eprint [0]{\href }%
\providecommand \doibase [0]{https://doi.org/}%
\providecommand \selectlanguage [0]{\@gobble}%
\providecommand \bibinfo  [0]{\@secondoftwo}%
\providecommand \bibfield  [0]{\@secondoftwo}%
\providecommand \translation [1]{[#1]}%
\providecommand \BibitemOpen [0]{}%
\providecommand \bibitemStop [0]{}%
\providecommand \bibitemNoStop [0]{.\EOS\space}%
\providecommand \EOS [0]{\spacefactor3000\relax}%
\providecommand \BibitemShut  [1]{\csname bibitem#1\endcsname}%
\let\auto@bib@innerbib\@empty
\bibitem [{\citenamefont {Magoulas}\ and\ \citenamefont
  {Evangelista}(2025)}]{Magoulas2025SpinAdapted}%
  \BibitemOpen
  \bibfield  {author} {\bibinfo {author} {\bibfnamefont {I.}~\bibnamefont
  {Magoulas}}\ and\ \bibinfo {author} {\bibfnamefont {F.~A.}\ \bibnamefont
  {Evangelista}},\ }\href {https://doi.org/10.48550/arXiv.2511.13485} {\bibinfo
  {title} {Spin-adapted fermionic unitaries: From {Lie} algebras to compact
  quantum circuits}} (\bibinfo {year} {2025}),\ \Eprint
  {https://arxiv.org/abs/2511.13485v2} {arXiv:2511.13485v2 [quant-ph]}
  \BibitemShut {NoStop}%
\bibitem [{\citenamefont {Evangelista}\ \emph {et~al.}(2019)\citenamefont
  {Evangelista}, \citenamefont {Chan},\ and\ \citenamefont
  {Scuseria}}]{Evangelista2019UCC}%
  \BibitemOpen
  \bibfield  {author} {\bibinfo {author} {\bibfnamefont {F.~A.}\ \bibnamefont
  {Evangelista}}, \bibinfo {author} {\bibfnamefont {G.~K.-L.}\ \bibnamefont
  {Chan}},\ and\ \bibinfo {author} {\bibfnamefont {G.~E.}\ \bibnamefont
  {Scuseria}},\ }\bibfield  {title} {\bibinfo {title} {Exact parameterization
  of fermionic wave functions via unitary coupled cluster theory},\ }\href
  {https://doi.org/10.1063/1.5133059} {\bibfield  {journal} {\bibinfo
  {journal} {The Journal of Chemical Physics}\ }\textbf {\bibinfo {volume}
  {151}},\ \bibinfo {pages} {244112} (\bibinfo {year} {2019})}\BibitemShut
  {NoStop}%
\bibitem [{\citenamefont {Burton}\ \emph {et~al.}(2023)\citenamefont {Burton},
  \citenamefont {Marti-Dafcik}, \citenamefont {Tew},\ and\ \citenamefont
  {Wales}}]{Burton2023Exact}%
  \BibitemOpen
  \bibfield  {author} {\bibinfo {author} {\bibfnamefont {H.~G.~A.}\
  \bibnamefont {Burton}}, \bibinfo {author} {\bibfnamefont {D.}~\bibnamefont
  {Marti-Dafcik}}, \bibinfo {author} {\bibfnamefont {D.~P.}\ \bibnamefont
  {Tew}},\ and\ \bibinfo {author} {\bibfnamefont {D.~J.}\ \bibnamefont
  {Wales}},\ }\bibfield  {title} {\bibinfo {title} {Exact electronic states
  with shallow quantum circuits from global optimisation},\ }\href
  {https://doi.org/10.1038/s41534-023-00744-2} {\bibfield  {journal} {\bibinfo
  {journal} {npj Quantum Information}\ }\textbf {\bibinfo {volume} {9}},\
  \bibinfo {pages} {75} (\bibinfo {year} {2023})}\BibitemShut {NoStop}%
\bibitem [{\citenamefont {Burton}(2024)}]{Burton2024TUPS}%
  \BibitemOpen
  \bibfield  {author} {\bibinfo {author} {\bibfnamefont {H.~G.~A.}\
  \bibnamefont {Burton}},\ }\bibfield  {title} {\bibinfo {title} {Accurate and
  gate-efficient quantum ans{\"a}tze for electronic states without adaptive
  optimization},\ }\href {https://doi.org/10.1103/PhysRevResearch.6.023300}
  {\bibfield  {journal} {\bibinfo  {journal} {Physical Review Research}\
  }\textbf {\bibinfo {volume} {6}},\ \bibinfo {pages} {023300} (\bibinfo {year}
  {2024})}\BibitemShut {NoStop}%
\bibitem [{\citenamefont {Stergiou}\ and\ \citenamefont
  {Sawaya}(2026)}]{Stergiou2026}%
  \BibitemOpen
  \bibfield  {author} {\bibinfo {author} {\bibfnamefont {A.}~\bibnamefont
  {Stergiou}}\ and\ \bibinfo {author} {\bibfnamefont {N.~P.~D.}\ \bibnamefont
  {Sawaya}},\ }\href {https://doi.org/10.48550/arXiv.2605.00979} {\bibinfo
  {title} {Universality of quantum gates in particle and symmetry constrained
  subspaces}} (\bibinfo {year} {2026}),\ \Eprint
  {https://arxiv.org/abs/2605.00979v1} {arXiv:2605.00979v1 [quant-ph]}
  \BibitemShut {NoStop}%
\bibitem [{\citenamefont {Oszmaniec}\ and\ \citenamefont
  {Zimbor{\'a}s}(2017)}]{Oszmaniec2017}%
  \BibitemOpen
  \bibfield  {author} {\bibinfo {author} {\bibfnamefont {M.}~\bibnamefont
  {Oszmaniec}}\ and\ \bibinfo {author} {\bibfnamefont {Z.}~\bibnamefont
  {Zimbor{\'a}s}},\ }\bibfield  {title} {\bibinfo {title} {Universal extensions
  of restricted classes of quantum operations},\ }\href
  {https://doi.org/10.1103/PhysRevLett.119.220502} {\bibfield  {journal}
  {\bibinfo  {journal} {Physical Review Letters}\ }\textbf {\bibinfo {volume}
  {119}},\ \bibinfo {pages} {220502} (\bibinfo {year} {2017})}\BibitemShut
  {NoStop}%
\bibitem [{\citenamefont {Gard}\ \emph {et~al.}(2020)\citenamefont {Gard},
  \citenamefont {Zhu}, \citenamefont {Barron}, \citenamefont {Mayhall},
  \citenamefont {Economou},\ and\ \citenamefont {Barnes}}]{Gard2020}%
  \BibitemOpen
  \bibfield  {author} {\bibinfo {author} {\bibfnamefont {B.~T.}\ \bibnamefont
  {Gard}}, \bibinfo {author} {\bibfnamefont {L.}~\bibnamefont {Zhu}}, \bibinfo
  {author} {\bibfnamefont {G.~S.}\ \bibnamefont {Barron}}, \bibinfo {author}
  {\bibfnamefont {N.~J.}\ \bibnamefont {Mayhall}}, \bibinfo {author}
  {\bibfnamefont {S.~E.}\ \bibnamefont {Economou}},\ and\ \bibinfo {author}
  {\bibfnamefont {E.}~\bibnamefont {Barnes}},\ }\bibfield  {title} {\bibinfo
  {title} {Efficient symmetry-preserving state preparation circuits for the
  variational quantum eigensolver algorithm},\ }\href
  {https://doi.org/10.1038/s41534-019-0240-1} {\bibfield  {journal} {\bibinfo
  {journal} {npj Quantum Information}\ }\textbf {\bibinfo {volume} {6}},\
  \bibinfo {pages} {10} (\bibinfo {year} {2020})}\BibitemShut {NoStop}%
\bibitem [{\citenamefont {Arrazola}\ \emph {et~al.}(2022)\citenamefont
  {Arrazola}, \citenamefont {{Di Matteo}}, \citenamefont {Quesada},
  \citenamefont {Jahangiri}, \citenamefont {Delgado},\ and\ \citenamefont
  {Killoran}}]{Arrazola2022}%
  \BibitemOpen
  \bibfield  {author} {\bibinfo {author} {\bibfnamefont {J.~M.}\ \bibnamefont
  {Arrazola}}, \bibinfo {author} {\bibfnamefont {O.}~\bibnamefont {{Di
  Matteo}}}, \bibinfo {author} {\bibfnamefont {N.}~\bibnamefont {Quesada}},
  \bibinfo {author} {\bibfnamefont {S.}~\bibnamefont {Jahangiri}}, \bibinfo
  {author} {\bibfnamefont {A.}~\bibnamefont {Delgado}},\ and\ \bibinfo {author}
  {\bibfnamefont {N.}~\bibnamefont {Killoran}},\ }\bibfield  {title} {\bibinfo
  {title} {Universal quantum circuits for quantum chemistry},\ }\href
  {https://doi.org/10.22331/q-2022-06-20-742} {\bibfield  {journal} {\bibinfo
  {journal} {Quantum}\ }\textbf {\bibinfo {volume} {6}},\ \bibinfo {pages}
  {742} (\bibinfo {year} {2022})}\BibitemShut {NoStop}%
\bibitem [{\citenamefont {Tang}\ \emph {et~al.}(2021)\citenamefont {Tang},
  \citenamefont {Shkolnikov}, \citenamefont {Barron}, \citenamefont {Grimsley},
  \citenamefont {Mayhall}, \citenamefont {Barnes},\ and\ \citenamefont
  {Economou}}]{Tang2021}%
  \BibitemOpen
  \bibfield  {author} {\bibinfo {author} {\bibfnamefont {H.~L.}\ \bibnamefont
  {Tang}}, \bibinfo {author} {\bibfnamefont {V.~O.}\ \bibnamefont
  {Shkolnikov}}, \bibinfo {author} {\bibfnamefont {G.~S.}\ \bibnamefont
  {Barron}}, \bibinfo {author} {\bibfnamefont {H.~R.}\ \bibnamefont
  {Grimsley}}, \bibinfo {author} {\bibfnamefont {N.~J.}\ \bibnamefont
  {Mayhall}}, \bibinfo {author} {\bibfnamefont {E.}~\bibnamefont {Barnes}},\
  and\ \bibinfo {author} {\bibfnamefont {S.~E.}\ \bibnamefont {Economou}},\
  }\bibfield  {title} {\bibinfo {title} {{Qubit-ADAPT-VQE}: An adaptive
  algorithm for constructing hardware-efficient ans{\"a}tze on a quantum
  processor},\ }\href {https://doi.org/10.1103/PRXQuantum.2.020310} {\bibfield
  {journal} {\bibinfo  {journal} {PRX Quantum}\ }\textbf {\bibinfo {volume}
  {2}},\ \bibinfo {pages} {020310} (\bibinfo {year} {2021})}\BibitemShut
  {NoStop}%
\bibitem [{\citenamefont {Shkolnikov}\ \emph {et~al.}(2023)\citenamefont
  {Shkolnikov}, \citenamefont {Mayhall}, \citenamefont {Economou},\ and\
  \citenamefont {Barnes}}]{Shkolnikov2023}%
  \BibitemOpen
  \bibfield  {author} {\bibinfo {author} {\bibfnamefont {V.~O.}\ \bibnamefont
  {Shkolnikov}}, \bibinfo {author} {\bibfnamefont {N.~J.}\ \bibnamefont
  {Mayhall}}, \bibinfo {author} {\bibfnamefont {S.~E.}\ \bibnamefont
  {Economou}},\ and\ \bibinfo {author} {\bibfnamefont {E.}~\bibnamefont
  {Barnes}},\ }\bibfield  {title} {\bibinfo {title} {Avoiding symmetry
  roadblocks and minimizing the measurement overhead of adaptive variational
  quantum eigensolvers},\ }\href {https://doi.org/10.22331/q-2023-06-12-1040}
  {\bibfield  {journal} {\bibinfo  {journal} {Quantum}\ }\textbf {\bibinfo
  {volume} {7}},\ \bibinfo {pages} {1040} (\bibinfo {year} {2023})}\BibitemShut
  {NoStop}%
\bibitem [{\citenamefont {Haidar}\ \emph {et~al.}(2025)\citenamefont {Haidar},
  \citenamefont {Adjoua}, \citenamefont {Badreddine}, \citenamefont {Peruzzo},\
  and\ \citenamefont {Piquemal}}]{Haidar2025}%
  \BibitemOpen
  \bibfield  {author} {\bibinfo {author} {\bibfnamefont {M.}~\bibnamefont
  {Haidar}}, \bibinfo {author} {\bibfnamefont {O.}~\bibnamefont {Adjoua}},
  \bibinfo {author} {\bibfnamefont {S.}~\bibnamefont {Badreddine}}, \bibinfo
  {author} {\bibfnamefont {A.}~\bibnamefont {Peruzzo}},\ and\ \bibinfo {author}
  {\bibfnamefont {J.-P.}\ \bibnamefont {Piquemal}},\ }\bibfield  {title}
  {\bibinfo {title} {Non-iterative disentangled unitary coupled-cluster based
  on {Lie}-algebraic structure},\ }\href
  {https://doi.org/10.1088/2058-9565/adb3c5} {\bibfield  {journal} {\bibinfo
  {journal} {Quantum Science and Technology}\ }\textbf {\bibinfo {volume}
  {10}},\ \bibinfo {pages} {025031} (\bibinfo {year} {2025})}\BibitemShut
  {NoStop}%
\bibitem [{\citenamefont {Viswanathan}\ \emph {et~al.}(2026)\citenamefont
  {Viswanathan}, \citenamefont {Adjoua}, \citenamefont {Feniou}, \citenamefont
  {Badreddine},\ and\ \citenamefont {Piquemal}}]{Viswanathan2026}%
  \BibitemOpen
  \bibfield  {author} {\bibinfo {author} {\bibfnamefont {Y.}~\bibnamefont
  {Viswanathan}}, \bibinfo {author} {\bibfnamefont {O.}~\bibnamefont {Adjoua}},
  \bibinfo {author} {\bibfnamefont {C.}~\bibnamefont {Feniou}}, \bibinfo
  {author} {\bibfnamefont {S.}~\bibnamefont {Badreddine}},\ and\ \bibinfo
  {author} {\bibfnamefont {J.-P.}\ \bibnamefont {Piquemal}},\ }\href
  {https://doi.org/10.48550/arXiv.2511.22593} {\bibinfo {title} {An optimized
  construction of {Lie} algebra generator pools for variational quantum
  eigensolvers in chemistry}} (\bibinfo {year} {2026}),\ \Eprint
  {https://arxiv.org/abs/2511.22593v3} {arXiv:2511.22593v3 [quant-ph]}
  \BibitemShut {NoStop}%
\bibitem [{\citenamefont {Anselmetti}\ \emph {et~al.}(2021)\citenamefont
  {Anselmetti}, \citenamefont {Wierichs}, \citenamefont {Gogolin},\ and\
  \citenamefont {Parrish}}]{Anselmetti2021}%
  \BibitemOpen
  \bibfield  {author} {\bibinfo {author} {\bibfnamefont {G.-L.~R.}\
  \bibnamefont {Anselmetti}}, \bibinfo {author} {\bibfnamefont
  {D.}~\bibnamefont {Wierichs}}, \bibinfo {author} {\bibfnamefont
  {C.}~\bibnamefont {Gogolin}},\ and\ \bibinfo {author} {\bibfnamefont {R.~M.}\
  \bibnamefont {Parrish}},\ }\bibfield  {title} {\bibinfo {title} {Local,
  expressive, quantum-number-preserving {VQE} ans{\"a}tze for fermionic
  systems},\ }\href {https://doi.org/10.1088/1367-2630/ac2cb3} {\bibfield
  {journal} {\bibinfo  {journal} {New Journal of Physics}\ }\textbf {\bibinfo
  {volume} {23}},\ \bibinfo {pages} {113010} (\bibinfo {year}
  {2021})}\BibitemShut {NoStop}%
\bibitem [{\citenamefont {Sapova}\ and\ \citenamefont
  {Fedorov}(2022)}]{Sapova2022}%
  \BibitemOpen
  \bibfield  {author} {\bibinfo {author} {\bibfnamefont {M.~D.}\ \bibnamefont
  {Sapova}}\ and\ \bibinfo {author} {\bibfnamefont {A.~K.}\ \bibnamefont
  {Fedorov}},\ }\bibfield  {title} {\bibinfo {title} {Variational quantum
  eigensolver techniques for simulating carbon monoxide oxidation},\ }\href
  {https://doi.org/10.1038/s42005-022-00982-4} {\bibfield  {journal} {\bibinfo
  {journal} {Communications Physics}\ }\textbf {\bibinfo {volume} {5}},\
  \bibinfo {pages} {199} (\bibinfo {year} {2022})}\BibitemShut {NoStop}%
\end{thebibliography}%

%% file: article_body.tex
\title{Minimal building blocks for molecular quantum circuits\texorpdfstring{\\}{ }
with exact spin symmetry}

\input{author_information.tex}

\begin{abstract}
Preserving particle number and spin helps quantum circuits target molecular electronic states, but does not guarantee access to every state with the required quantum numbers. We determine which additional operations, combined with spin-independent orbital rotations connecting all spatial orbitals, generate every real state-space rotation within each complete subspace of fixed particle number \(N\), total spin \(S\), and spin projection \(M_S\). We consider spin-free molecular calculations in real orbitals, without an additional spatial-symmetry restriction, and continuously tunable operations that preserve particle number, full spin symmetry, and real amplitudes. Below the maximal-spin limits set by the electron and hole numbers, repeated singlet-pair transfer between any two fixed spatial orbitals is sufficient. On nontrivial maximal-spin boundaries, pair transfer vanishes, and an additional generator built from one- and two-electron terms is sufficient exactly when its action in the target subspace is not a linear combination of orbital-rotation generators. The minimum number of spatial orbitals needed by one additional generator is two in the interior and three on nontrivial boundaries, even when terms involving more than two electrons are allowed. Both minima are attained using only one- and two-electron terms. The three-orbital optimum rotates two orbitals according to the occupation of a third, and its unitary factors exactly into eight commuting Pauli rotations. Molecular benchmarks show that this operation removes the observed boundary energy-error plateaus, while sparse interior constructions attain the prescribed energy accuracy with fewer compiled CNOT gates in selected fixed-orbital comparisons.
\end{abstract}

\keywords{dynamical Lie algebra, spin-preserving quantum circuits,
real-state controllability, variational quantum deflation, quantum chemistry}

\maketitle
\tableofcontents
\vspace{2em}

\input{sections/main_results_working.tex}

\begin{acknowledgments}
This work was supported by the National Natural Science Foundation of China
(Grant No. 22393913).
\end{acknowledgments}

\section*{Data Availability}
The formal-verification development, numerical data, and reproducibility
code for this work are available through Zenodo at
\href{https://doi.org/10.5281/zenodo.22230014}{10.5281/zenodo.22230014}~\cite{Liu2026Zenodo}.
The archive contains the Lean formalization, figure and table data,
accepted circuit parameters and states, numerical problem definitions, and an
independent Qiskit~\cite{Qiskit2024} workflow that reconstructs and replays
the 650 circuits underlying main Figs.~\ref{fig:benchmark-boundary-repair}
and~\ref{fig:benchmark-sota} and Supplementary Fig.~2.
Production VQD
optimization and the Q\textsuperscript{2}Chemistry compilation used for the
manuscript resource counts were performed with
Q\textsuperscript{2}Chemistry~\cite{Fan2022Q2Chemistry}, which is not
distributed in the archive. The Qiskit resource results follow a separate
compilation protocol and provide an independent circuit-replay and validation
route rather than replacements for the manuscript resource counts.

\appendix
\renewcommand{\thesubsection}{\thesection.\arabic{subsection}}
\input{appendices/appendix_a_preliminaries.tex}

\input{appendices/appendix_b_boundary.tex}
\input{appendices/appendix_c_interior.tex}
\input{appendices/appendix_d_support_factorization.tex}
\input{appendices/appendix_f_methods.tex}

\nocite{Tang2021,Shkolnikov2023,Haidar2025,Viswanathan2026,Sapova2022}
  \ifdefined\CombinedDocument
\input{main_and_si.bbl}  \else
    \bibliography{refs}%
  \fi

%% file: author_information.tex
\author{Mengwei Liu}
\affiliation{State Key Laboratory of Precision and Intelligent Chemistry,
University of Science and Technology of China, Hefei 230026, China}
\affiliation{Hefei National Laboratory,
University of Science and Technology of China, Hefei 230088, China}
\author{Zhenyu Li}
\email{zyli@ustc.edu.cn}
\affiliation{State Key Laboratory of Precision and Intelligent Chemistry,
University of Science and Technology of China, Hefei 230026, China}
\affiliation{Hefei National Laboratory,
University of Science and Technology of China, Hefei 230088, China}

%% file: sections/main_results_working.tex
\input{sections/01_introduction.tex}
\input{sections/02_classification.tex}
\input{sections/03_mechanism.tex}
\input{sections/04_implementation.tex}
\input{sections/05_benchmarks.tex}
\input{sections/06_discussion.tex}

%% file: sections/01_introduction.tex
\section{Introduction}

Molecular electronic states are distinguished not only by their energies but also by their particle number, total spin, and spin projection. When the quantum numbers of a target state are known, building them into a quantum circuit restricts the search to the subspace containing states with those quantum numbers, called a symmetry sector \cite{Gard2020}. This is particularly useful for excited-state calculations: a state that is excited in the full molecular spectrum can be the lowest-energy state in its own sector. Excluding states with incompatible quantum numbers simplifies state targeting without a spin-targeting penalty; methods such as variational quantum deflation (VQD) then need only distinguish lower-energy states within the chosen sector \cite{Higgott2019VQD}.

We consider the complete subspace specified simultaneously by particle number \(N\), total spin \(S\), and spin projection \(M_S\), with no additional restriction by molecular point-group symmetry. Our physical setting is a nonrelativistic, spin-free molecular Hamiltonian, represented in real spatial orbitals and without magnetic or spin--orbit terms. In this setting the Hamiltonian has a real matrix representation in a real many-electron basis with definite particle number and spin, and its eigenvectors may be chosen with real amplitudes \cite{Helgaker2000}. We therefore ask whether every real state in the complete fixed-\((N,S,M_S)\) subspace can be reached from a chosen real reference in that subspace, using only operations that preserve the specified quantum numbers.

We start from orbital rotations, which mix spatial orbitals in the same way for both spin components. Additional operations are built from terms acting on one electron or an electron pair at a time, conventionally called one- and two-body terms. This counts the electrons involved in each term, not the number of spatial orbitals it uses. Each operation has an adjustable angle; the operator specifying how the state changes as this angle is varied is called its generator. Every allowed operation must preserve particle number, commute with all three components of total spin, and preserve real amplitudes.

Applying operations in different orders, together with their inverses, produces transformations described at small angles by commutators. Repeated commutators of one- and two-body generators can contain higher-body terms, a mechanism used in general constructions of arbitrarily accurate fermionic wavefunctions \cite{Evangelista2019UCC}.

Neither symmetry preservation nor the appearance of higher-body terms guarantees complete access within a sector: operations may vanish or become redundant after restriction, and their combinations may still leave allowed states inaccessible. Such a restriction cannot be removed by a longer sequence of the same operations or a better optimizer.

Previous work has developed spin-preserving products of simple operations and shown how some operations can be removed while their effects are recovered through combinations of those retained \cite{Burton2023Exact,Magoulas2025SpinAdapted}. Here we address a structural question in symmetry-preserving quantum-circuit design: with orbital rotations already available, how simple can the elementary operations remain while generating every real state-space rotation within a complete fixed-\((N,S,M_S)\) sector? The resulting classification is sector-dependent but independent of the particular molecular Hamiltonian. We therefore ask which single additional generator is sufficient, how few spatial orbitals it must involve, and whether the resulting construction remains useful in a finite circuit.

\subsection{Main results and physical picture}

The relevant distinction is set by the largest total spin compatible with the orbital space. For \(M\) spatial orbitals and \(N\) electrons, the spin is bounded both by \(N/2\) and by \((2M-N)/2\). The second bound counts holes---the missing electrons relative to a fully occupied orbital space. We call sectors saturating either bound the electron or hole maximal-spin boundary, respectively; sectors below both bounds form the interior.

Throughout the interior, a particularly simple operation is enough: it transfers a spin-singlet electron pair from one spatial orbital to another. We denote its generator on distinct orbitals \(p\) and \(q\) by \(X_{pq}\). Any one fixed \(X_{pq}\) with \(p\ne q\), together with orbital rotations along a connected graph, generates every real rotation of the many-electron amplitude vector.

Commutators of pair transfers make a transfer depend on the occupations of other orbitals. This allows us to separate changes that a single pair transfer would apply to several configurations at once.

At either maximal-spin boundary, the same pair transfer has no effect. Its failure is transparent in the component with \(M_S=S\): on the electron boundary there is no opposite-spin pair to remove, whereas on the hole boundary every spatial orbital already contains a spin-up electron, blocking transfer of an additional pair. Because the operations commute with the full spin action, this conclusion applies to every allowed \(M_S\), not only to the component used to explain it. Some limiting sectors already admit all required transformations using orbital rotations alone, or contain only one physical state. We call the remaining boundary sectors nontrivial. In these sectors, the disappearance of pair transfer leaves a genuine restriction on the accessible states.

The classification extends beyond this example. Within the allowed class of one- and two-body generators, we establish a necessary-and-sufficient boundary criterion: an additional generator restores all real state-space rotations precisely when its action within the target sector is not a linear combination of the orbital-rotation generators on all orbital pairs. The comparison includes rotations obtainable indirectly from the available connected set. This criterion applies to arbitrary generators in that class, rather than only to selected excitation types.

The classification determines how many spatial orbitals an additional generator must involve. Every allowed generator acting on only two spatial orbitals reduces to an orbital rotation on a maximal-spin boundary. Changing the two-orbital operation, or adding more of them, therefore cannot resolve the obstruction on a nontrivial boundary. We prove that the minimum number of spatial orbitals needed by an additional generator to produce all real state-space rotations is two in the interior and three on nontrivial boundaries. The remaining endpoint sectors require no addition. This count is taken for the generator before restriction to the target sector and is called its orbital support. Both minima can be attained using only one- and two-body terms; involving three spatial orbitals does not require a three-body elementary generator.

For a sector of dimension \(d\), the sufficient operation sets identified here generate the full dynamical Lie algebra \(\mathfrak{so}(d)\), a stronger result than state reachability alone \cite{DAlessandro2007}.

The boundary optimum is attained by an operation that lets the occupation of a third orbital determine a rotation between two others. The rotation proceeds in opposite directions when the third orbital is empty or doubly occupied, and is absent when that orbital is singly occupied. If \(L_{pq}\) generates the spin-independent rotation between orbitals \(p\) and \(q\), and \(n_c\) counts the electrons in a distinct orbital \(c\), the additional generator is
\[
C_{c;pq}=(n_c-1)L_{pq}.
\]
The factor \(n_cL_{pq}\) is two-body, and subtracting \(L_{pq}\) adds only an already available one-body rotation. Together with orbital rotations, any one fixed choice of three distinct orbitals is sufficient on every nontrivial boundary. Under the Jordan--Wigner mapping, its unitary factors exactly into eight mutually commuting Pauli rotations, with no product-formula error.

For \(M\ge3\), the classification also yields a fixed three-orbital generator that suffices with orbital rotations in every sector separately, although two orbitals are enough in the interior.

Finally, the ability to reach every state does not specify the depth or optimization effort needed to reach a chemically relevant one. We examine the boundary mechanism in high-spin \(\mathrm{H}_4\) and an active-space model of \(\mathrm{O}_2\), and test sparse interior constructions on linear \(\mathrm{H}_6\), tetramethyleneethane, and \(\mathrm{N}_2\) dissociation. The conditioned rotation removes the observed boundary energy-error plateaus, while sparse constructions based on the classification can reduce compiled two-qubit-gate counts.

\subsection{Relation to previous work}

The most direct circuit starting point for this work is the quantum-number-preserving (QNP) gate fabric of Anselmetti et al.\ \cite{Anselmetti2021}. Its QNP-Q blocks combine a spin-independent orbital rotation with a singlet-pair exchange, corresponding, up to parameter conventions, to the generators \(L_{pq}\) and \(X_{pq}\) used here. The published five-parameter QNP-F block has additional local rotation axes; on strict maximal-spin boundaries it reduces to orbital rotations at finite angles, so its all-pair fabric is incomplete on nontrivial boundaries (Supplemental Material~\cite{SupplementalMaterial}, Sec.~2). The authors numerically tested reachability between randomly chosen states in prescribed quantum-number sectors and identified high-spin exceptions in which pair exchange becomes inactive. Both the basic operations and the existence of exceptional sectors therefore precede the present analysis. The broader combination of generalized singles and paired doubles also appears in the earlier \(k\)-UpCCGSD construction of Lee et al.\ \cite{Lee2019UpCCGSD}.

Subsequent work developed this operator family into both variable-sequence and fixed-structure circuits. Burton et al.\ introduced symmetry-preserving unitary product states and DISCO-VQE, which jointly optimizes the choice and ordering of operations and their continuous parameters \cite{Burton2023Exact}. They also supplied a commutator-based universality argument for spin-adapted singles and paired doubles, with an explicit requirement for a doubly occupied and an empty spatial orbital in the reference. Burton's tiled unitary product state (tUPS) construction places a pair transfer between two orbital rotations in a predetermined local arrangement and uses orbital optimization and a perfect-pairing reference to improve finite-depth performance. It also establishes that adjacent-orbital singles and paired doubles generate their arbitrary-pair counterparts \cite{Burton2024TUPS}. These works provide analytical foundations as well as practical circuit designs. The Supplemental Material~\cite{SupplementalMaterial}, Sec.~1, details the relation to these commutator constructions, including the treatment of spin-resolved intermediate operations and occupation-dependent factors.

A complementary route starts with larger collections of spin-adapted excitations and asks which operations can be removed. Magoulas and Evangelista develop exact circuit factorizations for these operations and show how retained excitation classes can recover the action of omitted classes in the non-fully-spin-polarized target spaces they consider \cite{Magoulas2025SpinAdapted}. Their work therefore addresses both exact implementation and operator-pool reduction, including settings with point-group constraints. Those constraints matter for the comparison: Magoulas, Zhang, and Evangelista show that enforcing spatial symmetry can leave the singles-plus-pairing collection with additional conserved occupation parities, obstructing otherwise allowed transformations \cite{Magoulas2026Dilemmas}. Our target instead includes all states with the prescribed \(N\), \(S\), and \(M_S\), imposing no point-group restriction on the orbital rotations used for state preparation. This freedom is essential to the orbital-transport arguments used here.

Stergiou and Sawaya~\cite{Stergiou2026} establish the full real orthogonal Lie algebra on fixed-occupation qubit subspaces using occupation-conditioned exchanges without nonlocal Jordan--Wigner strings. We apply their fixed-occupation result to hard-core-pair slices with fixed singly occupied orbitals, where the remaining orbitals are empty or doubly occupied. Completing the molecular spin sector additionally requires strict localization and extension in the interior, and derivation of conditioned exchanges from admissible spin-preserving generators on the boundaries. The detailed comparison is given in Supplemental Material, Sec.~1.3.

\begingroup
\widowpenalty=10000
Building on these developments, we fix connected orbital rotations and ask which single additional generator suffices in a prescribed sector, and how few spatial orbitals it can involve. We determine the generated algebra directly within each complete fixed-\((N,S,M_S)\) sector: the interior result proves sufficiency of every fixed pair transfer \(X_{pq}\) with \(p\ne q\), while the boundary result classifies all admissible one- and two-body additions. The same two-orbital limitation determines the strict maximal-spin boundary action of the five-parameter \(F\) block introduced with QNP (Supplemental Material~\cite{SupplementalMaterial}, Sec.~2). The accompanying minimum-orbital theorem identifies when increasing the support of an additional operation is unavoidable, connecting the established circuit ingredients to a sector-dependent design principle.
\par
\endgroup

\subsection{Organization}

Section~\ref{sec:results} gives the sector classification, the minimum orbital
requirements, and explicit sufficient operations. Section~\ref{sec:mechanism}
explains the mechanisms and proof strategy. The exact implementation in
Section~\ref{sec:implementation} and the finite-depth tests in
Section~\ref{sec:numerical-protocol} can each be read using the results of
Section~\ref{sec:results}, without following the proofs. Section~\ref{sec:discussion}
discusses their implications and limits. Proofs of the main classification and
minimum-support results, together with numerical protocols, appear in the
appendices; detailed comparisons, extension results, and additional diagnostics
are given in the Supplemental Material~\cite{SupplementalMaterial}.

%% file: sections/02_classification.tex
\section{Sector-dependent requirements for complete state access}
\label{sec:results}

\subsection{Target sectors and elementary operations}
\label{sec:target-operations}

For \(M\) spatial orbitals, we consider the complete simultaneous eigenspace
\[
\Hsp_{\C}=\Hsp^{\C}_{M,N,S,M_S}
=\{\psi:\widehat N\psi=N\psi,\ 
\widehat S^2\psi=S(S+1)\psi,\ \widehat S_z\psi=M_S\psi\}.
\]
There is no further restriction to a molecular point-group representation.
In the spin-free, real-orbital setting specified in the Introduction, choose
a real orthonormal configuration-state-function (CSF) basis and let
\(\Hsp\) be its real span, with \(d=\dim\Hsp\). The complex sector is
\(\Hsp\otimes_{\R}\C\), and its molecular eigenvectors may be chosen real
\cite{Helgaker2000}. The dimension formula and equivalence of magnetic
components are given in Appendix~\ref{app:magnetic-equivalence}; the results
below apply to every allowed \(M_S\).

An allowed continuously tunable operation has the form \(e^{\theta K}\).
Its generator \(K\) is anti-Hermitian, conserves particle number, commutes
with all three components of total spin, and has real matrix elements in
the occupation basis. It is therefore real skew-symmetric in the compatible
real CSF bases. We call such a generator admissible. Our sufficient
constructions use only one- and two-body terms, which act on one electron
or an electron pair at a time. The minimum-orbital question below allows
admissible generators of arbitrary body rank.

Orbital rotations mix the spatial orbitals identically for both spin
components. With \(a_{p\sigma}^{\dagger}\) creating an electron of spin
\(\sigma\) in orbital \(p\), their generators are
\begin{equation}
\label{eq:s1-L}
E_{pq}=\sum_{\sigma=\alpha,\beta}a_{p\sigma}^{\dagger}a_{q\sigma},
\qquad n_p=E_{pp},\qquad L_{pq}=E_{pq}-E_{qp}.
\end{equation}
The second building block transfers a singlet electron pair from one
spatial orbital to another:
\begin{equation}
\label{eq:pd-X}
X_{pq}=P_p^{\dagger}P_q-P_q^{\dagger}P_p,
\qquad P_p^{\dagger}=a_{p\alpha}^{\dagger}a_{p\beta}^{\dagger},
\quad P_p=a_{p\beta}a_{p\alpha}.
\end{equation}
Both \(L_{qp}=-L_{pq}\) and \(X_{qp}=-X_{pq}\); the orientation changes
only the sign of the adjustable angle. These operations are the orbital
rotation and pair exchange used in QNP, up to parameter conventions
\cite{Anselmetti2021}, and also belong to established singles-and-pairing
product constructions \cite{Lee2019UpCCGSD,Burton2024TUPS}.

The available orbital pairs form a graph \(G\) on \([M]\), with edge set
\(\operatorname{Ed}(G)\). We require this
graph to be connected, rather than requiring a direct rotation between
every pair. Indeed, \([L_{pq},L_{qr}]=L_{pr}\), so commutators along paths
supply the missing orbital rotations. Their action on the many-electron
sector generates
\begin{equation}
\label{eq:connected-backbone}
\mathfrak g_{\mathrm{orb}}
=\Lie\{L_{pq}|_{\Hsp}:\{p,q\}\in\operatorname{Ed}(G)\}
=\rho_{M,N,S}(\so(M)).
\end{equation}
Here \(\rho_{M,N,S}\) denotes the orbital representation on the chosen
sector, with its magnetic label suppressed.

For any generator collection, \(\Lie\) denotes the real linear span of
its restricted generators and all their nested commutators. This is the
dynamical Lie algebra: it describes the infinitesimal transformations
obtainable from arbitrary sequences of the allowed operations. Our target is
\[
\Lie\bigl(\mathfrak g_{\mathrm{orb}},K|_{\Hsp}\bigr)=\so(\Hsp).
\]
The right-hand side contains all real skew-symmetric \(d\times d\)
matrices. These rotate the amplitudes of the many-electron basis states,
not just the underlying spatial orbitals. For \(d\ge2\), they generate
\(SO(d)\), which can carry any normalized real state to any other
\cite{DAlessandro2007}. Generating the full algebra is a stronger target
than state reachability alone. For \(d=1\), there is only one real
projective state. A fixed additional generator may be used repeatedly,
with independently chosen angles at each occurrence.

\subsection{Classification across the sector family}
\label{sec:classification}

The largest possible spin is bounded by both the electron number and the
hole number: \(2S\le N\) and \(2S\le2M-N\). It is useful to write
\begin{equation}
\label{eq:bhr}
b=\frac{N-2S}{2},\qquad r=2S,\qquad
h=\frac{2M-N-2S}{2}.
\end{equation}
These nonnegative integers obey \(b+r+h=M\). In a highest-weight
occupation pattern with the fewest singly occupied orbitals, they count
doubly occupied, singly occupied, and empty orbitals. They do not assign
the same occupations to every state in the sector. The interior has
\(b,h\ge1\), including partially filled singlet sectors. The electron
maximal-spin boundary has \(b=0\), and the hole boundary has \(h=0\).
On a boundary let \(k=N\) for \(b=0\) and \(k=2M-N\) for \(h=0\); in
either case \(k=2S\). We call the
boundary nontrivial when \(2\le k\le M-2\); the remaining boundary
sectors need no additional operation.

\begin{theorem}[Sector-dependent classification]
\label{thm:main}
Consider a complete fixed-\((N,S,M_S)\) sector with no additional spatial
symmetry restriction and a connected graph of orbital rotations.
\begin{enumerate}[label=\textup{(\roman*)}]
\item In every interior sector, each fixed \(p\ne q\) satisfies
\[
\Lie\bigl(\mathfrak g_{\mathrm{orb}},X_{pq}|_{\Hsp}\bigr)
=\so(\Hsp).
\]
\item On either maximal-spin boundary, every \(X_{pq}|_{\Hsp}\) vanishes.
On a nontrivial boundary, every admissible generator \(K\) of body rank
at most two obeys the necessary-and-sufficient criterion
\begin{equation}
\label{eq:boundary-action-criterion}
\boxed{\Lie\bigl(\mathfrak g_{\mathrm{orb}},K|_{\Hsp}\bigr)
=\so(\Hsp)
\quad\Longleftrightarrow\quad
K|_{\Hsp}\notin\mathfrak g_{\mathrm{orb}}.}
\end{equation}
\item In the reduced boundary representation, \(k=1\) and \(M-1\)
correspond to a single particle and a single hole, respectively. Thus
\(d=M\), and orbital rotations already generate \(\so(d)\). At \(k=0\)
or \(M\), the sector has \(d=1\). Neither case requires an additional generator.
\end{enumerate}
\end{theorem}

Thus a useful one- or two-body addition on a nontrivial boundary is exactly
one whose restricted action supplies a direction outside the linear space
of orbital rotations. Within this allowed class, any such direction is
sufficient. Section~\ref{sec:pair-space} gives a directly computable version
of this complete boundary classification in terms of the pair coefficients
of \(K\). The interior statement has a different quantifier: it proves
sufficiency of every individual \(X_{pq}\) with \(p\ne q\), without classifying arbitrary
two-body linear combinations. The proofs are developed in
Section~\ref{sec:mechanism} and Appendices~\ref{app:boundary-completion}
and~\ref{app:interior-full-proof}.

A concrete addition that completes every nontrivial boundary is
\begin{equation}
\label{eq:conditioned-single}
C_{c;pq}=(n_c-1)L_{pq},\qquad c,p,q\ \text{pairwise distinct}.
\end{equation}
It rotates \(p,q\) in opposite directions when orbital \(c\) is empty
or doubly occupied, and has no effect when \(c\) is singly occupied.
Any one fixed pairwise distinct choice of \(c,p,q\) suffices:
\begin{equation}
\label{eq:C-coordinatewise-universality}
\Lie\bigl(\mathfrak g_{\mathrm{orb}},C_{c;pq}|_{\Hsp}\bigr)
=\so(\Hsp)\qquad\text{on every nontrivial boundary}.
\end{equation}
The third orbital supplies an occupation condition on an otherwise ordinary
orbital rotation. The next result shows why that extra orbital is necessary.

\subsection{Minimum orbital requirements and an optimal construction}
\label{sec:minimum-orbitals}

The orbital support \(\operatorname{supp}_{\rm orb}(K)\) is the minimum
number of spatial orbitals whose fermionic modes support \(K\). It is
counted before restriction to the target sector. A global sector projector,
or a factor depending on other orbital occupations, contributes every
orbital on which it acts. When orbital rotations do not already generate
the full algebra, define
\begin{equation}
\label{eq:minimum-support-definition}
\begin{split}
s_{\min}(M,N,S)=\min\{&\operatorname{supp}_{\rm orb}(K):
K\ \text{admissible},\\
&\Lie(\mathfrak g_{\mathrm{orb}},K|_{\Hsp})=\so(\Hsp)\}.
\end{split}
\end{equation}
Orbital rotations already connect all \(M\) spatial orbitals;
\(s_{\min}\) counts only the orbitals involved in the additional generator.
This minimum allows arbitrary body rank. In particular, its lower bound
does not only exclude two-body operators on two orbitals.

\begin{theorem}[Minimum orbital support]
\label{thm:min-support}
Under the assumptions of \cref{thm:main},
\[
s_{\min}(M,N,S)=
\begin{cases}
2,&b,h\ge1,\\
3,&bh=0,\quad2\le k\le M-2.
\end{cases}
\]
Every \(X_{pq}\) with \(p\ne q\) attains the interior minimum. Every \(C_{c;pq}\) with
three distinct orbitals attains the nontrivial-boundary minimum.
\end{theorem}

Both optima are achieved using terms of body rank at most two:
\(X_{pq}\) is two-body, while \(C_{c;pq}=n_cL_{pq}-L_{pq}\) combines
two-body and one-body terms. The change from two to three counts spatial
orbitals, not electrons in an elementary interaction. On a nontrivial
boundary, every admissible generator genuinely confined to two orbitals
restricts to a multiple of an orbital-rotation generator;
adding more such generators cannot supply
a missing direction. Section~\ref{sec:boundary-mechanism} explains this
restriction, and Appendix~\ref{app:min-support-proof} gives the full
minimum-support proof, including the interior lower bound.

Figure~\ref{fig:boundary-classification} and Table~\ref{tab:classification}
summarize the sector dependence. The endpoint sectors are stated separately
because they already need no addition. These results minimize the orbital
support of one completing generator; they do not specify minimum circuit
depth, parameter count, compiled gate count, or measurement cost.

\begin{figure}[!htbp]
\centering
\includegraphics[width=0.75\textwidth]{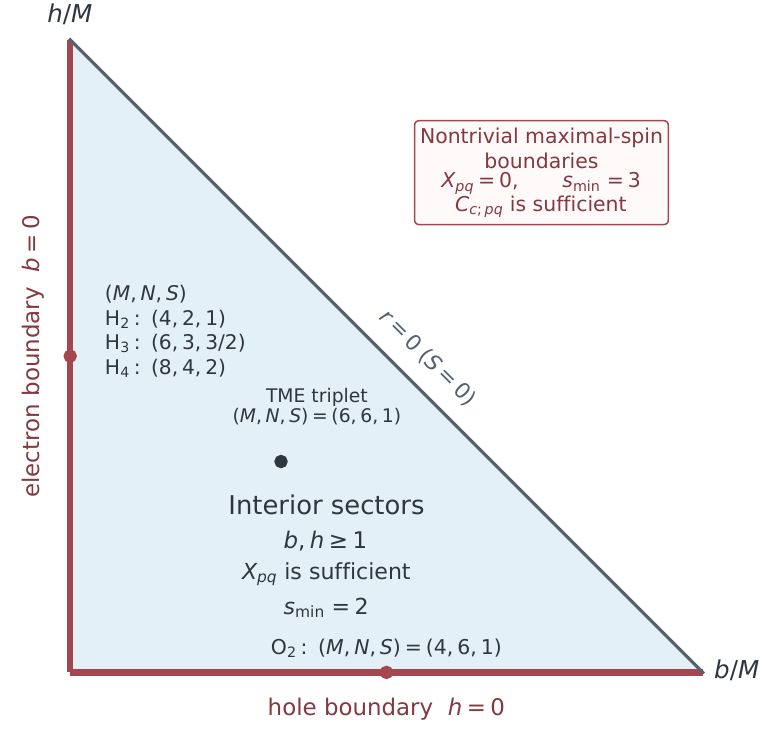}
\caption{Sufficient additions and minimum orbital support across the sector
family. The coordinates \(b,r,h\) in \cref{eq:bhr} count double, single,
and empty occupations in a minimum-seniority highest-weight pattern.
The interior, including singlets on \(r=0\), requires two orbitals and is
completed by any fixed \(X_{pq}\) with \(p\ne q\). On the nontrivial parts of
\(b=0\) and \(h=0\), \(X_{pq}=0\), and any fixed \(C_{c;pq}\) with pairwise
distinct \(c,p,q\) attains the three-orbital minimum. Endpoint exceptions are listed in
Table~\ref{tab:classification}. The marked H$_2$, H$_3$, H$_4$ sectors
have \((b/M,h/M)=(0,1/2)\); the O$_2$ \(\pi\)-CAS(6e,4o) sector has
\((1/2,0)\), and the TME triplet sector has \((1/3,1/3)\) in the interior.
The continuous simplex organizes normalized integer counts;
only physically allowed discrete sectors are classified.}
\label{fig:boundary-classification}
\end{figure}

\begin{table}[!htbp]
\centering\small
\caption{Complete real-rotation requirements with connected orbital rotations
already available. Orbital counts refer to one additional generator before
sector restriction. The full boundary criterion for arbitrary admissible
one- and two-body additions is \cref{eq:boundary-action-criterion}.}
\label{tab:classification}
\renewcommand{\arraystretch}{1.25}
\setlength{\tabcolsep}{4pt}
\begin{tabular}{@{}llll@{}}
\toprule
\tblcell{0.27\textwidth}{Sector}
&\tblcell{0.27\textwidth}{Sufficient addition}
&\tblcell{0.16\textwidth}{Resulting algebra}
&\tblcell{0.17\textwidth}{Minimum orbital count}\\
\midrule
\tblcell{0.27\textwidth}{Interior: \(b,h\ge1\)}
&\tblcell{0.27\textwidth}{Any fixed pair transfer \(X_{pq}\), \(p\ne q\)}
&\tblcell{0.16\textwidth}{\(\so(d)\)}
&\tblcell{0.17\textwidth}{\(2\)}\\
\tblcell{0.27\textwidth}{Nontrivial boundary: \(bh=0\), \(2\le k\le M-2\)}
&\tblcell{0.27\textwidth}{Any fixed conditioned rotation \(C_{c;pq}\), distinct \(c,p,q\)}
&\tblcell{0.16\textwidth}{\(\so(d)\)}
&\tblcell{0.17\textwidth}{\(3\)}\\
\tblcell{0.27\textwidth}{Boundary: \(k=1,M-1\)}
&\tblcell{0.27\textwidth}{No addition required}
&\tblcell{0.16\textwidth}{\(\so(M)=\so(d)\)}
&\tblcell{0.17\textwidth}{Not required}\\
\tblcell{0.27\textwidth}{Boundary: \(k=0,M\)}
&\tblcell{0.27\textwidth}{No addition required}
&\tblcell{0.16\textwidth}{\(d=1\), \(\so(d)=0\)}
&\tblcell{0.17\textwidth}{Not required}\\
\bottomrule
\end{tabular}
\end{table}

%% file: sections/03_mechanism.tex
\section{Mechanism and proof strategy}
\label{sec:mechanism}

The proofs explain how a small set of operations can distinguish all the
independent changes of a many-electron state. We work here in the
highest-weight component \(M_S=S\). The normalized spin-lowering map is
real orthogonal in compatible real CSF bases and intertwines every allowed
generator and commutator (Appendix~\ref{app:magnetic-equivalence}), so the
conclusions apply to every magnetic component. In this section \(\Hsp\)
denotes that real highest-weight space. The notation used in the detailed
proofs is collected in Table~\ref{tab:notation-guide}.
Figure~\ref{fig:controllability-mechanism} illustrates the interior and
boundary mechanisms developed below.

\subsection{From one pair transfer to the complete interior sector}
\label{sec:interior-control}

Assume \(b,h\ge1\). A fixed pair-transfer generator already provides all
pair-transfer coordinates in the presence of orbital rotations: conjugation
by a signed orbital permutation carries \(X_{pq}\) to \(\pm X_{rs}\).
The permutation can be chosen in \(SO(M)\), and the generated Lie algebra
is invariant under these conjugations. Appendix~\ref{app:pair-transport}
gives the sign and closure details. We may therefore work with
\[
\mathfrak g_{LX}
=\Lie\bigl(\mathfrak g_{\mathrm{orb}},\{X_{pq}|_{\Hsp}:p<q\}\bigr)
\]
without changing the algebra generated by any single fixed \(X_{pq}\).
Restrictions to \(\Hsp\) are implicit below.

\paragraph{Conditioning and isolating pair transfers.}
The first step makes a transfer respond differently to different electron
configurations. For three distinct orbitals,
\begin{equation}
\label{eq:px-conditional}
[X_{pc},X_{cq}]=(1-n_c)X_{pq}=Z_cX_{pq},\qquad Z_c=1-n_c.
\end{equation}
The occupation factor is \(+1\), \(0\), or \(-1\) when the third orbital
is empty, singly occupied, or doubly occupied. The transfer thus acquires
a condition on electrons that it does not itself move. Iterating this
identity produces \(Z_\Omega X_{pq}\), where
\(Z_\Omega=\prod_{a\in\Omega}(1-n_a)\) and
\(\Omega\subseteq[M]\setminus\{p,q\}\). These higher-body actions
are generated by commutators of the original two-body transfers;
Appendix~\ref{app:conditional-pairs} proves the induction.

To use these conditions, fix a set \(R\) of \(r=2S\) singly occupied
\(\alpha\) orbitals and let pairs redistribute over the remaining
orbitals \(E_R=[M]\setminus R\). This defines the subspace
\begin{equation}
\label{eq:interior-slice}
K_R=\Span_{\R}\{\ket{B;R}:B\subseteq E_R,\ |B|=b\}.
\end{equation}
Here \(B\) labels doubly occupied orbitals; the other orbitals are empty.
For distinct \(p,q\in E_R\), the operator
\(Z_{E_R\setminus\{p,q\}}X_{pq}\) kills every single
occupation in \(E_R\). Any surviving pattern can therefore have at most
\(r\) singly occupied orbitals, all in \(R\). Total spin \(S=r/2\)
requires at least \(r\) such orbitals, and \(M_S=S\) aligns all of them
with \(\alpha\) spin. Particle number then fixes exactly \(b\) pairs.
Consequently the operator acts only on \(K_R\) and vanishes on
\(K_R^\perp\).

Within \(K_R\), the conditioned transfers become exchanges of hard-core
pairs, up to a common sign. Occupation projectors separate their individual
configuration-to-configuration rotations, giving
\begin{equation}
\label{eq:interior-supported-slice}
\so(K_R)\oplus0_{K_R^\perp}\subseteq\mathfrak g_{LX}.
\end{equation}
Hereafter, \(\so(Q)\) for \(Q\subseteq\Hsp\) denotes its zero extension
to \(Q^\perp\) inside \(\Hsp\).
The hard-core-pair step uses the fixed-occupation construction of
Ref.~\cite{Stergiou2026}. The additional spin-sector argument above makes
its action strictly supported on \(K_R\), rather than merely recovering
the right matrix after compression to that subspace. Both steps are proved
in Appendices~\ref{app:hardcore-pairs} and~\ref{app:strict-localization}.
This strict support is what allows the next construction.

\paragraph{A three-state rotation block.}
Choose the standard highest-weight occupation vector
\begin{equation}
\label{eq:interior-highest-vector}
v_\lambda=P_1^\dagger\cdots P_b^\dagger
a_{b+1,\alpha}^\dagger\cdots a_{b+r,\alpha}^\dagger\ket0,
\qquad \lambda=(2^b,1^r,0^h).
\end{equation}
It has a doubly occupied orbital \(p\) and an empty orbital \(q\).
Remove that pair to obtain a normalized spectator state \(\ket\chi\).
With the other electrons unchanged, consider
\[
u=P_p^\dagger\ket\chi=v_\lambda,\qquad
v=P_q^\dagger\ket\chi,\qquad
w=\frac{a_{p\alpha}^\dagger a_{q\beta}^\dagger
-a_{p\beta}^\dagger a_{q\alpha}^\dagger}{\sqrt2}\ket\chi.
\]
The pair is on \(p\), on \(q\), or split between them as a singlet
[Fig.~\ref{fig:controllability-mechanism}(a)].
All three states belong to the same target sector and are orthonormal.
Let \(Q_0=\Span_{\R}\{u,v,w\}\). Equation~\eqref{eq:interior-supported-slice}
provides the single supported rotation
\(A=\ket u\bra v-\ket v\bra u\). Since
\(L_{pq}u=-\sqrt2\,w\) and \(L_{pq}v=\sqrt2\,w\), the three operators
\[
A,\qquad [A,L_{pq}],\qquad [A,[A,L_{pq}]]
\]
are supported on \(Q_0\) and span its three independent plane rotations.
Thus \(\so(Q_0)\subseteq\mathfrak g_{LX}\), with
\(v_\lambda\in Q_0\) and \(\dim Q_0=3\).
Appendix~\ref{app:three-dimensional-block} verifies the support and signs.
Three here counts many-electron states, not spatial orbitals: the pair is
rearranged on just \(p,q\). The same construction includes
\((M,N,S)=(2,2,0)\), for which \(Q_0=\Hsp\).

\paragraph{Extending the rotations to the complete sector.}
Two facts ensure that this independently rotatable space can be extended
to the whole sector. First, the real linear span of all vectors obtained
by applying finite products of represented orbital generators to
\(v_\lambda\) is \(\Hsp\). Equivalently, any linear subspace containing
\(v_\lambda\) and invariant under all orbital rotations is already the
whole sector. Appendix~\ref{app:highest-weight-cyclicity} proves this
cyclicity statement.

Second, suppose all zero-extended rotations on a subspace \(Q\) of
dimension at least three are available. If an allowed
generator has a nonzero block \(B:Q\to Q^\perp\), commutators with those
rotations give
\begin{equation}
\label{eq:interior-infection}
\so(Q\oplus\im B)
\subseteq\mathfrak g_{LX}.
\end{equation}
This expansion result, including the strict support and the requirement
\(\dim Q\ge3\), is proved in Appendix~\ref{app:infection-proof}.

Now take a largest subspace containing \(Q_0\) on which all supported
rotations are available. If it were smaller than \(\Hsp\), cyclicity
would force some orbital rotation to connect it to the outside. The
expansion result would then produce a larger such subspace, a contradiction.
Hence \(\mathfrak g_{LX}=\so(\Hsp)\). Together with coordinate transport,
this proves the interior part of \cref{thm:main}. The full-pool theorem
and single-coordinate corollary are stated and proved in
Appendix~\ref{app:interior-completion}.

\subsection{Why nontrivial boundaries require a third orbital}
\label{sec:boundary-mechanism}

On the electron maximal-spin boundary, the highest-weight component has
only \(\alpha\) electrons, so there is no opposite-spin pair to remove.
On the hole boundary, every spatial orbital already has an \(\alpha\)
electron, blocking pair creation elsewhere. Thus every \(X_{pq}\) with
\(p\ne q\) vanishes. Appendix~\ref{app:px-boundary-proof} also proves the
converse: each such coordinate is nonzero in every interior sector.

The boundary spaces reduce to \(\bigwedge^k\R^M\), with \(k\) spinless
particles or holes. Orbital rotations can leave real states inaccessible
there. To see this directly, let \(\gamma_\psi\) be a state's one-particle
density matrix. An orbital rotation changes it to
\(g\gamma_\psi g^{\mathsf T}\), preserving its spectrum. A determinant
has \(k\) occupations equal to one and the rest zero. For
\(2\le k\le M-2\), choose two determinants that differ by replacing
two occupied orbitals with two empty ones. Their normalized equal
superposition has instead
\[
\operatorname{spec}(\gamma_\psi)
=\bigl(1^{k-2},(1/2)^4,0^{M-k-2}\bigr).
\]
The one-body transition density between the determinants vanishes, so this
spectrum follows directly. They cannot lie on the same orbital-rotation
orbit. Particle--hole duality gives the same obstruction on the hole
boundary. This establishes actual missing state access, beyond comparing
Lie-algebra dimensions.

For pairwise distinct \(c,p,q\), the third orbital in \(C_{c;pq}\)
supplies the missing dependence on the
surrounding configuration. On the electron boundary its restriction is
\((n_c^{\mathrm{sl}}-1)L_{pq}\); on the hole boundary it is
\((1-n_c^{\mathrm h})L_{pq}^{\mathrm h}\). Since orbital rotations are
already available, either gives an occupation-conditioned rotation, up to
sign. Orbital conjugation supplies every choice of pairwise distinct condition
and rotated-pair labels. Further commutators add occupation conditions
until they isolate a single pair of configurations
[Fig.~\ref{fig:controllability-mechanism}(b)]:
\begin{equation}
\label{eq:edge}
\begin{aligned}
n_\Omega^{\mathrm{sl}}L_{pq}
&=\pm(\ket I\bra J-\ket J\bra I),\\
I&=\Omega\cup\{p\},\quad J=\Omega\cup\{q\},\\
\Omega&\subseteq[M]\setminus\{p,q\},\quad p\ne q,\\
|\Omega|&=k-1,\quad 2\le k\le M-2.
\end{aligned}
\end{equation}
Here \(n_\Omega^{\mathrm{sl}}=\prod_{a\in\Omega}n_a^{\mathrm{sl}}\).
The configurations are connected by exchanges of occupied and empty
orbitals, and commutators along these connections produce all plane
rotations. This independently establishes completion by one fixed
\(C_{c;pq}\) with pairwise distinct \(c,p,q\).
Appendices~\ref{app:boundary-car}--\ref{app:higher-conditioning}
give the explicit constructions, including shared-orbital and disjoint
triplet-pair coordinates; Appendix~\ref{app:min-support-proof} records
the precise particle and hole relations between those coordinates and
\(C\). Those relations hold after boundary restriction, with orbital
terms and signs retained.

There is also a short reason why no genuinely two-orbital alternative can
do the same job. On a maximal-spin boundary, two spatial orbitals reduce
to two spinless modes. Number conservation splits their local space into
blocks of dimensions \(1,2,1\). A real skew-symmetric operator vanishes
on either one-dimensional block; on the two-dimensional block it has
only the single orbital-rotation direction. True two-orbital support
requires the same coefficient for every spectator configuration. Thus
every such generator reduces to \(\kappa L_{pq}\), regardless of its
body rank. Appendix~\ref{app:two-orbital-boundary-proof} states and proves
the restriction precisely. Any collection of two-orbital additions remains
inside the orbital algebra, while \(C\) attains the three-orbital lower
bound.

\begin{figure}[!htbp]
\centering
\includegraphics[width=\textwidth]{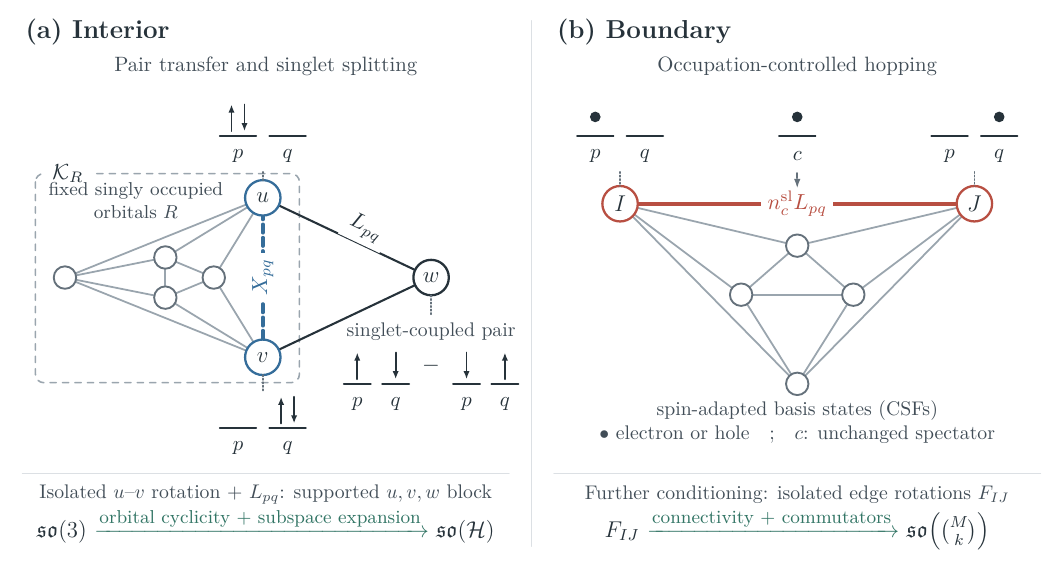}
\caption{Physical transitions underlying the two controllability mechanisms.
Vertices are spin-adapted many-electron basis states (CSFs) in the
\(M_S=S\) component; horizontal bars denote orbitals rather than energies.
(a) In the dashed subspace \(K_R\) of Eq.~\eqref{eq:interior-slice}, pair
transfer \(X_{pq}\) connects \(u=P_p^\dagger\ket\chi\) and
\(v=P_q^\dagger\ket\chi\), while an orbital rotation links them to \(w\).
The paired spin sketches represent one normalized singlet CSF. Isolating the
\(u\)--\(v\) rotation gives all rotations on \(\Span\{u,v,w\}\); orbital
cyclicity and subspace expansion extend them to the complete sector.
(b) On either nontrivial maximal-spin boundary, configurations form the
Johnson graph \(J(M,k)\) of spinless particles or holes. The highlighted
hopping is conditioned on the binary occupation \(n_c^{\mathrm{sl}}\) of an
unchanged spectator \(c\notin\{p,q\}\), with \(p\ne q\). This reduced
occupation differs from the full
spin-summed \(n_c\). The generator \(C_{c;pq}\), together with orbital
rotations, supplies \(n_c^{\mathrm{sl}}L_{pq}\) up to sign; further
conditions isolate rotations \(F_{IJ}=\ket I\bra J-\ket J\bra I\).
Connectivity and commutators generate all plane rotations. The conclusions
apply to every allowed \(M_S\).}
\label{fig:controllability-mechanism}
\end{figure}

\FloatBarrier
\subsection{Classification of general one- and two-body additions}
\label{sec:pair-space}

The preceding construction gives a sufficient operation. To decide whether
an arbitrary admissible one- or two-body addition is sufficient, we first
separate its two-electron coefficients according to the spin of the pair.
Let \(V=\R^M\). A spin-zero pair has a symmetric spatial state, whereas
a spin-one pair has an antisymmetric spatial state. Such pair tensors
are used in spin-adapted operator representations~\cite{LiPaldus1993};
their orbital spaces are
\[
\mathcal P_0=\operatorname{Sym}^2V,\qquad
\mathcal P_1=\bigwedge^2V.
\]
These labels describe the coupled pair; they do not restrict the total
spin of the many-electron target to zero or one.

Every admissible generator of body rank at most two has the unique
full-Fock-space representation
\begin{equation}
\label{eq:at-most-two-body-seed}
\begin{split}
K(\eta,A_0,A_1)&=d\Gamma(\eta)+\operatorname{Lift}_0(A_0)
+\operatorname{Lift}_1(A_1),\\
\eta&\in\so(V),\qquad A_j\in\so(\mathcal P_j).
\end{split}
\end{equation}
Here \(d\Gamma(\eta)=\sum_{p,q}\eta_{pq}E_{pq}\), and
\(\operatorname{Lift}_j\) applies a matrix of pair transfers, summed
over the pair's magnetic components. Normalized pair creators, the full
lift formula, completeness and coefficient counts appear in
Appendix~\ref{app:pair-space}.

For example, let \(e,f\in\binom{[M]}2\) be distinct, with the pair
orientations fixed, and let \(T_{e,m}^{\dagger}\) create a normalized
triplet pair on \(e\). The corresponding spin-scalar transfer is
\begin{equation}
\label{eq:dint1-Y}
Y_{e,f}=\sum_{m=-1}^{1}
\bigl(T_{e,m}^{\dagger}T_{f,m}-T_{f,m}^{\dagger}T_{e,m}\bigr).
\end{equation}
For orthonormal pair-basis vectors \(\ket\mu,\ket\nu\in\mathcal P_j\),
the skew pair coordinate is
\[
K^{(j)}_{\mu\nu}
=\operatorname{Lift}_j\!\left(\ket\mu\bra\nu-\ket\nu\bra\mu\right).
\]
The matrix inside the lift acts on \(\mathcal P_j\), whereas
\(K^{(j)}_{\mu\nu}\) acts on the full fermionic Fock space. The relevant
generators satisfy
\begin{equation}
\label{eq:pair-space-dictionary}
X_{pq}=K^{(0)}_{pp,qq},\qquad Y_{e,f}=K^{(1)}_{e,f},\qquad
n_cL_{pq}=K^{(0)}_{cp,cq}+K^{(1)}_{cp,cq}.
\end{equation}
The last identity assumes distinct \(c,p,q\) and consistent pair
orientations. These coordinates connect the classification to the
spin-adapted operators studied in
Refs.~\cite{Magoulas2025SpinAdapted,Jain2026Exact}. Exact conventions and
comparisons with prior universality arguments are given in the
Supplemental Material~\cite{SupplementalMaterial}, Sec.~1.

On a boundary, some pair coefficients merely reproduce orbital rotations.
Let \(j_1:\so(V)\to\so(\mathcal P_1)\) be the induced orbital action and
separate those directions from their Frobenius-orthogonal complement:
\begin{equation}
\label{eq:triplet-effective-splitting}
\mathcal J_1=j_1(\so(V)),\qquad
\so(\mathcal P_1)=\mathcal J_1\oplus\mathcal W.
\end{equation}
The component in \(\mathcal W\) can be computed directly from the
triplet-pair coefficient matrix. Write
\(\ell_{pq}=e_{pq}-e_{qp}\in\so(V)\) for the spatial matrix whose second
quantization is \(L_{pq}\). Then
\begin{equation}
\label{eq:PW-explicit}
P_{\mathcal W}(A_1)=A_1-\sum_{p<q}
\frac{\langle A_1,j_1(\ell_{pq})\rangle_F}{2(M-2)}j_1(\ell_{pq}).
\end{equation}
Here \(e_{pq}\) is the elementary matrix on \(V\), and in the orthonormal
pair basis \(\langle B,C\rangle_F=\operatorname{tr}(B^{\mathsf T}C)
=\sum_{\mu,\nu}B_{\mu\nu}C_{\mu\nu}\).
The matrices \(j_1(\ell_{pq})\) are mutually orthogonal, each with squared
Frobenius norm \(2(M-2)\). Thus \cref{eq:PW-explicit} subtracts the component
of \(A_1\) along every orbital-induced direction.

On either boundary, the real identification
\(\Hsp\simeq\bigwedge^kV\) sends \(\mathfrak g_{\mathrm{orb}}\) to
\(\mathfrak k=d\Gamma_k(\so(V))\). Here \(d\Gamma_k\) lifts orbital matrices
to one-body actions on \(\bigwedge^kV\), while
\(\gamma_k:\so(\mathcal P_1)\to\so(\bigwedge^kV)\) lifts triplet-pair
coefficient matrices to two-body actions. Under the same identification,
\(K|_{\Hsp}\) becomes \(\rho_{\mathrm e}K\) or \(\rho_{\mathrm h}K\).
Their formulas for \cref{eq:at-most-two-body-seed} are
\begin{align}
\rho_{\mathrm e}K&=d\Gamma_k(\eta)+\gamma_k(A_1),
\label{eq:general-electron-restriction}\\
\rho_{\mathrm h}K&=d\Gamma_k(\eta_{\mathrm{eff}})+\gamma_k(-A_1).
\label{eq:general-hole-restriction}
\end{align}
The singlet-pair coefficient \(A_0\) vanishes on the electron boundary;
on the hole boundary it generally contributes to
\(\eta_{\mathrm{eff}}\in\so(V)\), an already available orbital direction.
Also \(\gamma_k(j_1(\xi))=(k-1)d\Gamma_k(\xi)\), so only
\(P_{\mathcal W}(A_1)\) can provide a new direction.

Two facts make this a complete criterion. For \(M\ge4\),
\(\mathcal W\) is a real irreducible orbital module: the span generated
from any nonzero element by orbital commutators is all of \(\mathcal W\).
Furthermore, for \(2\le k\le M-2\), its lift is injective after
quotienting out \(\mathfrak k\). Thus a nonzero coefficient in this
space cannot disappear or become an orbital rotation upon restriction.
Appendix~\ref{app:complete-boundary-classification} proves both statements,
including the \(M=4,6\) real forms of the representation-theoretic
decomposition \cite{Avetisyan2020X2}.

It follows that a nonzero effective component supplies the whole lifted
\(\mathcal W\). Together with the available orbital directions, this
recovers a shared-pair coordinate giving \(n_c^{\mathrm{sl}}L_{pq}\)
up to sign and an orbital term. The independently established boundary
construction then completes the algebra. Conversely, a zero effective
component leaves the addition inside \(\mathfrak k\). Thus, on either
nontrivial boundary and for any admissible at-most-two-body
\(K=K(\eta,A_0,A_1)\), the coefficient form of
\cref{eq:boundary-action-criterion} is
\begin{equation}
\label{eq:complete-boundary-criterion}
\boxed{\Lie\bigl(\mathfrak g_{\mathrm{orb}},K|_{\Hsp}\bigr)
=\so(\Hsp)
\quad\Longleftrightarrow\quad P_{\mathcal W}(A_1)\ne0.}
\end{equation}
The condition \(P_{\mathcal W}(A_1)\ne0\) is equivalent to
\(K|_{\Hsp}\notin\mathfrak g_{\mathrm{orb}}\), recovering
\cref{eq:boundary-action-criterion}. The full linear failure set in
coefficient space is
\begin{equation}
\label{eq:boundary-failure-subspace}
\{(\eta,A_0,A_1):A_1\in\mathcal J_1\}.
\end{equation}
For a fixed orbital space, the same coefficient test applies to every
nontrivial electron and hole maximal-spin boundary.

Every shared-orbital or disjoint \(Y_{e,f}\) coordinate with distinct pair
labels \(e\ne f\) has a nonzero \(\mathcal W\) component and completes a
nontrivial boundary. By
contrast, choosing \(A_1=j_1(\xi)\), with arbitrary \(\eta,A_0\),
gives a failing linear combination. The direct linear space of restricted
one- and two-body actions need not equal its Lie closure; its dimension is
recorded in \cref{cor:boundary-two-body-image} in
Appendix~\ref{app:complete-boundary-classification}.
These statements complete the arbitrary-generator boundary part of
\cref{thm:main}. They do not extend to arbitrary interior additions:
for example, every \(Y\) vanishes on a two-electron singlet, whereas
each \(X_{pq}\) with \(p\ne q\) remains sufficient there.

%% file: sections/04_implementation.tex
\section{Three-orbital operations and their exact implementation}
\label{sec:implementation}

For pairwise distinct \(c,p,q\), we implement the boundary-optimal generator
\(C_{c;pq}\) of Eq.~\ref{eq:conditioned-single} and construct a single addition
that, with connected orbital rotations, suffices in every sector separately.

\subsection{Rotation conditioned on orbital occupation}
\label{sec:occupation-rotation}

General spin-adapted double-excitation unitaries and their exact Lie-algebraic
factorizations have been studied in
Refs.~\cite{Magoulas2025SpinAdapted,Jain2026Exact}.  The boundary-optimal
conditioned single below is a different three-orbital seed whose commuting
structure yields a particularly direct factorization.

Let \(\Pi_0,\Pi_1,\Pi_2\) be the spectral projectors of \(n_c=E_{cc}\)
for occupations \(0,1,2\), respectively; their explicit polynomials are
given in Appendix~\ref{app:support-factorization}.
Because \([n_c,L_{pq}]=0\),
\[
C_{c;pq}=-\Pi_0L_{pq}+\Pi_2L_{pq},
\]
and therefore
\begin{equation}
\label{eq:controlled-givens}
\boxed{
\e^{\theta C_{c;pq}}
=\Pi_0\e^{-\theta L_{pq}}+\Pi_1+\Pi_2\e^{\theta L_{pq}}.
}
\end{equation}
Writing \(G_{pq}(\theta)=e^{\theta L_{pq}}\), this gives an exact
controlled Givens rotation.
Controlled single-excitation gates provide universal primitives for
particle-conserving unitaries \cite{Arrazola2022}.  The factorization above
uses total spatial-orbital occupation as a spin-scalar control and attains
the boundary support optimum of \cref{thm:min-support}.

\subsection{Exact decomposition into commuting Pauli rotations}
\label{sec:pauli-implementation}

Under any fixed Jordan--Wigner ordering \cite{JordanWigner1928}, for spin
orbitals \(u<v\),
\begin{align*}
a_u^\dagger a_v-a_v^\dagger a_u
=\frac{i}{2}\bigl(&X_uZ_{u+1}\cdots Z_{v-1}Y_v\\
&-Y_uZ_{u+1}\cdots Z_{v-1}X_v\bigr),
\end{align*}
while \(E_{cc}-1=-\tfrac12(Z_{c\alpha}+Z_{c\beta})\).  Hence two choices of
condition spin, two choices of hopping spin, and two Majorana channels give
\(2\times2\times2=8\) nonzero terms.  More explicitly,
\begin{equation}
\label{eq:exact-eight-pauli}
C_{c;pq}=-\frac{i}{4}
\sum_{\tau,\sigma\in\{\alpha,\beta\}}
\sum_{\kappa=0}^{1}
s_{\tau\sigma\kappa}P_{\tau\sigma\kappa},
\qquad s_{\tau\sigma\kappa}\in\{\pm1\},
\end{equation}
where \(P_{\tau\sigma\kappa}\) is the Hermitian Jordan--Wigner Pauli image
of the corresponding quartic Majorana monomial.  The argument below shows that
the eight strings are distinct as well as mutually commuting.

\begin{proposition}[Commuting eight-string decomposition]
\label{prop:eight-pauli}
All Pauli strings in this Jordan--Wigner decomposition commute pairwise.  Thus
\[
\e^{\theta C_{c;pq}}
=\prod_{\ell=1}^{8}
\exp\left(-\frac{i\theta s_\ell}{4}P_\ell\right),
\qquad s_\ell\in\{\pm1\},
\]
with no Trotter error.
\end{proposition}

\begin{proof}
Each term is the Jordan--Wigner image of a quartic Majorana monomial:
a condition bilinear on \(c\) times a hopping bilinear on \(p,q\).
Two such supports share zero, two, or four Majorana indices, so the
monomials commute. The two choices of condition spin, two of hopping
spin, and two hopping channels have distinct supports, giving eight
distinct Pauli strings. Appendix~\ref{app:exact-factorization} supplies
the explicit support and sign calculation.
\end{proof}

The commutativity is internal to one \(C\) block. Different orbital and
conditioned blocks need not commute; combining those blocks produces the
additional directions required for completeness.
The symmetry guarantee applies to the complete fermionic
\(e^{\theta C_{c;pq}}\) block; its individual Pauli-rotation factors need not
preserve particle number or total spin.

\subsection{A common construction across sectors}
\label{sec:common-construction}

The smallest sufficient addition depends on the sector, but one may also
want to fix the same generator before choosing the target quantum numbers.
Combining pair transfer with an occupation-weighted orbital rotation gives
such a choice while retaining body rank two.

\begin{corollary}[A common two-body generator across sectors]
\label{cor:sector-uniform-seed}
Let \(M\ge3\), let \(G\) be connected, and choose pairwise distinct
\(c,p,q\).  For any \(\xi,\zeta\in\R\setminus\{0\}\), the fixed
three-orbital generator
\begin{equation}
\boxed{
\widetilde Z_{c;pq}
=\xi X_{pq}+\zeta n_cL_{pq}}
\label{eq:sector-uniform-seed}
\end{equation}
is an admissible pure two-body generator. For every physically allowed
complete fixed-\((N,S,M_S)\) sector of the same \(M\)-orbital Fock
space, its restriction satisfies
\begin{equation}
\Lie\bigl(\mathfrak g_{\mathrm{orb}},\widetilde Z_{c;pq}|_{\Hsp}\bigr)
=\so(\Hsp).
\label{eq:sector-uniform-completion}
\end{equation}
For \(M\ge4\), three-orbital support is minimal among single seeds required
to complete every exact-spin sector.
\end{corollary}

\par\noindent\emph{Proof.}
Because \(c\notin\{p,q\}\), \(n_cL_{pq}\) is pure, normal-ordered
two-body and commutes with \(L_{pq}\).  The full-Fock-space identities
\[
[L_{pq},[L_{pq},\widetilde Z_{c;pq}]]=-4\xi X_{pq},
\qquad
C_{c;pq}=\frac{\widetilde Z_{c;pq}-\xi X_{pq}}{\zeta}-L_{pq}
\]
therefore recover both \(X_{pq}\) and \(C_{c;pq}\) from the connected
backbone and \(\widetilde Z_{c;pq}\).  The interior and boundary parts of
\cref{thm:main} then give completion; at the remaining endpoints the
backbone is already complete or the target algebra is zero.  For \(M\ge4\),
the boundary lower bound in \cref{thm:min-support} proves optimality. \qed

The same fixed generator completes each sector separately; independent
simultaneous control of the direct sum is not asserted.

The conditioned rotation alone has a useful but different scope.
The Supplemental Material~\cite{SupplementalMaterial}, Sec.~3, proves that for \(M\ge3\), every fixed
\(C_{c;pq}\) with pairwise distinct \(c,p,q\) also completes a connected
orbital graph in every physically allowed positive-spin interior sector
\((b,h\ge1,\ S>0)\). In a nontrivial singlet sector, however, the orbital
rotations
and all elementary \(C\) generators share a nonzero fixed vector, so
their algebra is incomplete. That extension uses the already established
interior theorem; it is not an ingredient of the main classification.
The mixed generator above supplies a common sufficient choice by
recovering both \(X\) and \(C\). Its exponential is not covered by the
eight-commuting-string formula for \(C\) alone.

%% file: sections/05_benchmarks.tex
\section{Finite-depth consequences of the classification}
\label{sec:numerical-protocol}

Section~\ref{sec:results} establishes complete real-state rotations but does not determine the efficiency of finite-depth circuits. We test whether the conditioned rotation removes the observed boundary energy plateaus and whether a fixed pair transfer remains useful in sparse interior circuits.

The boundary tasks realize the two maximal-spin limits in different orbital
spaces. Linear H$_4$ has four electrons in eight spatial orbitals, with
\(S=M_S=2\), and lies on the electron boundary. The O$_2$ model has six
electrons in four \(\pi/\pi^\ast\) orbitals, with \(S=M_S=1\); the two
holes in this active space place it on the hole boundary. The interior
tasks all use six electrons in six spatial orbitals: linear H$_6$ in its
singlet sector, tetramethyleneethane (TME) in separate singlet and triplet
sectors, and N$_2$ in its singlet sector at eight bond lengths. These
interior calculations use \(M_S=0\), including the TME triplet.
Every task uses a complete fixed-\((N,S,M_S)\) space without an additional
point-group restriction. Molecular geometries, basis sets and reference
states are specified in Appendix~\ref{sec:benchmark-systems}.

Tree+one $X$ uses the $M-1$ path rotations with edges
$(0,1),(1,2),\ldots,(M-2,M-1)$ and one fixed pair transfer; Tree+$C$
replaces it with one fixed conditioned rotation. Dense $L$ contains all
orbital-pair rotations.
We compare with QNP-Q, tUPS, and tUPS with a perfect-pairing reference,
denoted tUPS (PP ref.) \cite{Anselmetti2021,Burton2024TUPS}, using fixed
orbitals and the stated local optimizer. Resources are parameter counts and
ansatz-only controlled-NOT (CNOT) counts compiled after parameter binding
under a common all-to-all connectivity and gate-decomposition protocol.
Figures~\ref{fig:benchmark-boundary-repair} and~\ref{fig:benchmark-sota}
test boundary repair and interior resource-to-accuracy endpoints,
respectively; their sample sizes and selection rules are summarized in
Table~\ref{tab:benchmark-protocols}. Additional backbone/seed ablations and
QNP-F comparisons are reported in the Supplemental Material.

\subsection{Removing boundary energy-error plateaus}
\label{sec:boundary-benchmarks}
A common-protocol boundary scan compares Dense $L$, Tree+$C$, boundary-effective tUPS, and tUPS+$C$ (\cref{fig:benchmark-boundary-repair}). Each checkpoint runs five independently initialized VQD chains and selects one complete chain by the sum of its final root objectives. On these strict maximal-spin boundaries, $X$ vanishes, so Dense $L$ and Dense $L+X$ have identical effective actions. Dense $L$ and unaugmented tUPS retain energy-error plateaus within the tested ranges, whereas the $C$-augmented products reach maximum-root errors below 0.1 mHa. Tree+$C$ and tUPS+$C$ both reach this criterion with 480 CNOTs for H$_4$ and 220 CNOTs for O$_2$. These calculations illustrate the distinction between repeating an insufficient operation set and adding the type of operation required by the classification. An independent common-deflator, single-root control is reported separately in Supplemental Material.

\begin{figure}[!htbp]
\centering
\includegraphics[width=\textwidth]{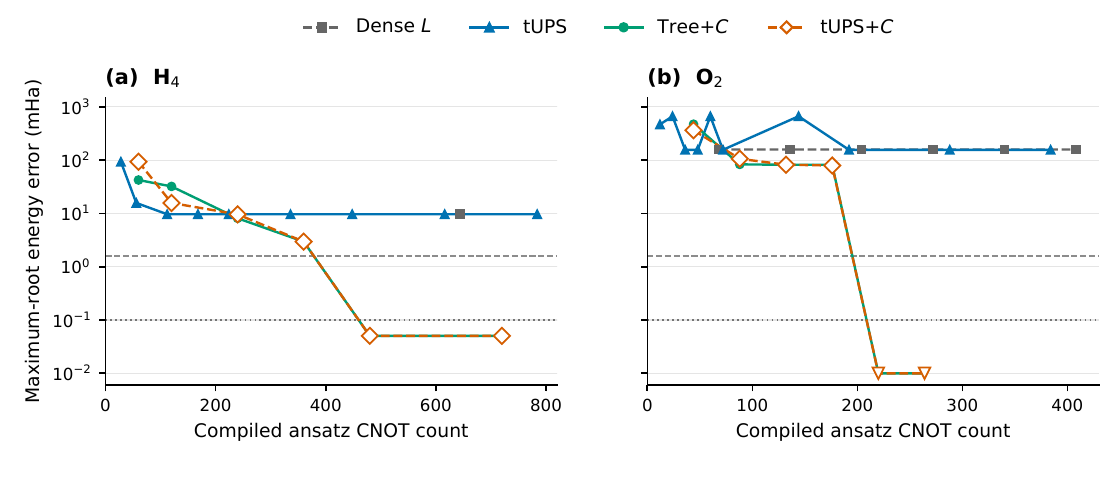}
\caption{A conditioned rotation removes the observed boundary energy plateau. (a) H$_4$ roots 0--1 and (b) O$_2$ roots 0--5. Each checkpoint selects one complete chain from five starts by the summed final VQD objectives. Energy error and ansatz CNOT count are maxima over the roots of that chain; the maximizing roots need not coincide. All 50 checkpoints within the displayed windows are shown. Five additional H$_4$ Dense $L$ checkpoints retain the approximately 9.67-mHa plateau up to 7728 CNOTs outside the window. Downward markers indicate errors below the 0.01-mHa display floor. Dashed and dotted lines mark 1.6 and 0.1 mHa; lines connect tested checkpoints. The optimization and compilation procedures are specified in Appendix~\ref{app:numerical-methods}.}
\label{fig:benchmark-boundary-repair}
\end{figure}

\subsection{Accuracy and compiled resources in interior sectors}
\label{sec:interior-benchmarks}
For each interior geometry and spin sector, each circuit size contributes the lowest-variational-energy output among 20 independently initialized starts and the compiled CNOT count of that same parameter-bound circuit (\cref{fig:benchmark-sota,tab:benchmark-headline}). The reported endpoint is the smallest observed CNOT count among tested circuits satisfying the 0.1-mHa task-energy criterion. For TME, the task error is the larger singlet or triplet energy error and the resource is the sum for both circuits. For $\mathrm{N}_2$, both quantities are maxima over the eight geometries at a common repetition count. At the reported endpoints, the selected optimizations satisfied an optimizer stopping condition before reaching the iteration and evaluation limits.

On linear \(\mathrm H_6\)/STO-3G, Tree+one \(X\) attains the 0.1-mHa
criterion with slightly fewer CNOTs than tUPS, but more parameters
(144 versus 120; \cref{tab:benchmark-headline}).

\Needspace{9\baselineskip}
For TME, Tree+one \(X\) uses both fewer parameters and 58.8\% fewer CNOT
gates than the fixed-orbital tUPS controls at the observed 16-repeat,
best-of-20 energy endpoint; its singlet and triplet outputs are selected
independently. At that size, none of the 20 paired starts meets the separate
joint strict-success criterion. Tree+one \(X\) first reaches that criterion
at 24 repeats (20/20 paired starts, 1632 CNOTs, maximum two-qubit depth 411),
compared with tUPS at 12 repeats (19/20, 2640 CNOTs, depth 384;
the TME joint-initialization table in
Supplemental Material).
The selected singlet--triplet gap error is 0.029 mHa. Because both
variational energies are upper bounds to their respective sector ground
states, the larger state-energy error also bounds the gap error.

\begingroup
\widowpenalty=10000
Across the eight-point \(\mathrm N_2\) dissociation curve, Tree+one \(X\)
attains the full-curve energy criterion with 22.7\% fewer CNOTs than
either tUPS reference. QNP-Q does not attain that criterion within the
displayed scan. The 16-layer tUPS error rebound reflects an independent local
search rather than a loss of expressivity: zero-padding a shallower solution
preserves its state in the deeper circuit, as checked in Supplemental Material.
\par
\endgroup

\begin{figure}[!htbp]
\centering
\includegraphics[width=\textwidth]{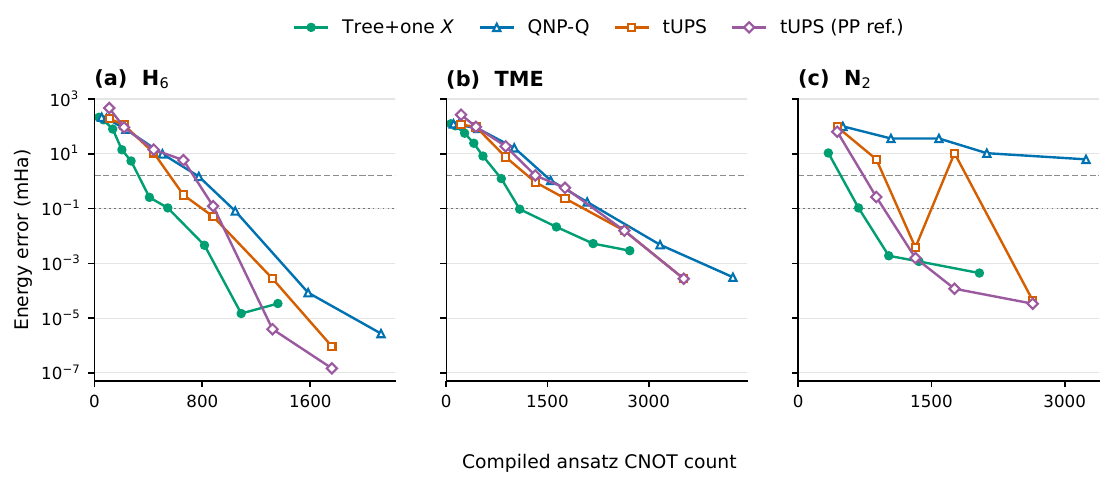}
\caption{Accuracy versus compiled ansatz CNOT count in complete fixed-\((N,S,M_S)\) interior sectors. Each point uses the variational-energy-selected output at a tested circuit size and the count of that same parameter-bound circuit. (a) Linear-$\mathrm{H}_6$ ground-state error and single-state count. (b) The larger TME singlet or triplet energy error and the summed count for both states. (c) Maximum energy error and maximum count over eight $\mathrm{N}_2$ geometries at a common repetition count; the two maxima may occur at different geometries. Lines connect sampled circuit sizes; dashed and dotted lines mark 1.6 and 0.1 mHa. All controls use fixed orbitals; PP ref. denotes the perfect-pairing reference. Supplemental Material gives initialization statistics and independent circuit validation.}
\label{fig:benchmark-sota}
\end{figure}

\FloatBarrier
\begin{table}[!htbp]
\centering
\small
\setlength{\tabcolsep}{4pt}
\caption{Observed interior resource-to-accuracy endpoints at 0.1~mHa.
The task error and ansatz-only CNOT count follow Fig.~\ref{fig:benchmark-sota}.
A star marks the largest tested checkpoint when the threshold is not reached.}
\label{tab:benchmark-headline}
\input{benchmark/main_interior_endpoints.tex}
\par\smallskip
\begin{minipage}{\textwidth}
\footnotesize\raggedright
$^{\dagger}$\enspace For TME, the common triplet result determines the reported
maximum error in both tUPS rows; the PP reference changes only the
singlet calculation.
\end{minipage}
\end{table}

\FloatBarrier

%% file: benchmark/main_interior_endpoints.tex
\begin{tabular*}{\textwidth}{@{\extracolsep{\fill}}llrrrr@{}}
\toprule
Task & Method & Repeats & Params. & CNOTs & Error (mHa) \\
\midrule
H$_6$ & Tree+one $X$ & 24 & 144 & 816 & 0.004571 \\
 & QNP-Q & 23 & 116 & 1044 & 0.081197 \\
 & tUPS & 8 & 120 & 880 & 0.050981 \\
 & tUPS (PP ref.) & 12 & 180 & 1320 & 0.000004 \\
TME & Tree+one $X$ & 16 & 192 & 1088 & 0.094966 \\
 & QNP-Q & 35 & 352 & 3168 & 0.004772 \\
 & tUPS$^{\dagger}$ & 12 & 360 & 2640 & 0.015516 \\
 & tUPS (PP ref.)$^{\dagger}$ & 12 & 360 & 2640 & 0.015516 \\
N$_2$ & Tree+one $X$ & 30 & 180 & 1020 & 0.001888 \\
 & QNP-Q$^{*}$ & 72 & 360 & 3240 & 6.265837 \\
 & tUPS & 12 & 180 & 1320 & 0.003800 \\
 & tUPS (PP ref.) & 12 & 180 & 1320 & 0.001565 \\
\bottomrule
\end{tabular*}

%% file: sections/06_discussion.tex
\section{Discussion}
\label{sec:discussion}

Complex electronic states can be represented by repeated combinations of
simple operations, but exact spin symmetry changes which simple operations
are sufficient. Within a complete fixed-\((N,S,M_S)\) sector, the extra
operation can remain confined to two spatial orbitals throughout the
interior. On a nontrivial maximal-spin boundary, every admissible two-orbital
generator restricts to an orbital rotation, and a third orbital is necessary. The
conditioned rotation attains this lower bound while still using only
one- and two-body terms. The change in required orbital support therefore
does not require increasing the elementary interaction's body rank.

The physical distinction is between producing higher-body terms and
obtaining enough independent transformations to rotate the entire allowed
real state space. Occupation conditioning separates changes that an
elementary operation applies simultaneously to different electron
configurations.

This distinction is useful when designing a finite circuit. If all of its
operations remain inside the insufficient orbital algebra on a nontrivial
boundary, increasing their repetition count cannot remove the restriction.
The conditioned rotation supplies the required addition and can be
implemented exactly as eight commuting Pauli rotations. The boundary
calculations demonstrate that adding this operation removes the observed
energy-error plateaus. In the interior H$_6$, TME, and N$_2$ calculations,
the fixed pair-transfer construction can reduce compiled ansatz CNOT
counts at the reported accuracy endpoints.

Minimal elementary operations do not imply a short circuit for an
arbitrary target state. The required repetitions and adjustable angles
depend on the task, the reference, and the optimization procedure. The
H$_6$ and N$_2$ constructions use fewer CNOT gates but greater compiled
two-qubit depth than tUPS at the reported energy endpoints
(see the resource comparison in Supplemental Material); H$_6$ also uses
more parameters.

The classification concerns complete fixed-\((N,S,M_S)\) spaces. Additional
spatial symmetry can invalidate the orbital transport and cyclicity used
here and leave singles-plus-pairing products with extra conserved
quantities \cite{Magoulas2026Dilemmas}. Point-group-resolved spaces and
Hamiltonians requiring complex amplitudes therefore need separate
analyses. Within the present real, spin-free setting, the next resource
question is how to assemble and optimize the classified operations
efficiently for chemically relevant states at finite depth.

%% file: appendices/appendix_a_preliminaries.tex
\section{Notation and algebraic preliminaries}

This appendix collects the common notation, the construction of the complete
fixed-spin sector, magnetic-component equivalence, and the orbital-covariance
facts used in the main classification. It also proves the complete pair-space
representation of pure two-body spin scalars.

\input{appendices/notation_guide.tex}
\input{appendices/01_preliminaries.tex}
\input{appendices/02_pair_space.tex}

%% file: appendices/notation_guide.tex
\begingroup
\small
\renewcommand{\arraystretch}{1.0}
\setlength{\LTleft}{0pt}
\setlength{\LTright}{0pt}
\setlength{\LTcapwidth}{\textwidth}
\newcommand{\ntrow}[2]{%
  \tblcell{0.28\textwidth}{#1}&\tblcell{0.66\textwidth}{#2}}
\begin{longtable}{@{}ll@{}}
\caption{Notation and symbol guide. The proof sections use the
highest-weight component \(M_S=S\) and the complete fixed-spin space,
without additional spatial-irrep restriction.}
\label{tab:notation-guide} \\
\toprule
\tblcell{0.28\textwidth}{Symbol}&\tblcell{0.66\textwidth}{Meaning or convention}\\
\midrule
\endfirsthead
\multicolumn{2}{@{}l}{\textbf{Table \thetable} \textit{(continued)}} \\
\toprule
\tblcell{0.28\textwidth}{Symbol}&\tblcell{0.66\textwidth}{Meaning or convention}\\
\midrule
\endhead
\midrule
\multicolumn{2}{r@{}}{\textit{Continued on the next page}} \\
\endfoot
\bottomrule
\endlastfoot
\multicolumn{2}{@{}l}{\textbf{Sector and occupation coordinates}} \\*
\ntrow{\(M,\ N\)}{Numbers of spatial orbitals and electrons, respectively.}\\
\ntrow{\(S,\ M_S\)}{Total spin and spin projection; the sector label also uses \(m=M_S\).}\\
\ntrow{\(\Hsp_{\C},\ \Hsp\)}{\(\Hsp_{\C}=\Hsp^{\C}_{M,N,S,M_S}\); \(\Hsp\) is its real CSF span. Proofs choose \(M_S=S\).}\\
\ntrow{\(J_{\mathrm{occ}}\)}{Coefficientwise complex conjugation in the fixed occupation basis.}\\
\ntrow{\(d=d_{M,N,S}\)}{Multiplicity dimension: \(\dim_{\R}\Hsp=\dim_{\C}\Hsp_{\C}\).}\\
\ntrow{\(b=(N-2S)/2\)}{Doubly occupied orbitals in the highest-weight pattern.}\\
\ntrow{\(r=2S\)}{Singly occupied orbitals in that pattern.}\\
\ntrow{\(h=(2M-N-2S)/2\)}{Empty orbitals in that pattern; \(b+r+h=M\), \(N=2b+r\).}\\
\addlinespace[0.35em]
\multicolumn{2}{@{}l}{\textbf{Fermionic operators and generators}} \\*
\ntrow{\(a^\dagger_{p\sigma},\ a_{p\sigma}\)}{Electron creation and annihilation; \(\sigma\in\{\alpha,\beta\}\).}\\
\ntrow{\(E_{pq},\ n_p\)}{\(E_{pq}=\sum_\sigma a^\dagger_{p\sigma}a_{q\sigma}\); \(n_p=E_{pp}\) is spin-summed occupation.}\\
\ntrow{\(P_{pq}^{0\dagger},\ P_p^\dagger\)}{Normalized singlet-coupled pair creation, \(p\le q\); the on-site pair is \(P_p^\dagger=P_{pp}^{0\dagger}=a^\dagger_{p\alpha}a^\dagger_{p\beta}\).}\\
\ntrow{\(T_{pq,m}^\dagger\)}{Triplet-coupled pair creation, \(p<q\); here \(m=-1,0,1\) labels pair spin projection.}\\
\ntrow{\(\ell_{pq},\ L_{pq}\)}{Spatial matrix \(\ell_{pq}=e_{pq}-e_{qp}\in\so(\R^M)\), with \(e_{pq}\) an elementary matrix, and its spin-free Fock-space lift \(L_{pq}=d\Gamma(\ell_{pq})=E_{pq}-E_{qp}\).}\\
\ntrow{\(X_{pq}\)}{Perfect-pairing transfer: \(X_{pq}=P_p^\dagger P_q-P_q^\dagger P_p\).}\\
\ntrow{\(Y_{e,f}\)}{Spin-scalar transfer between triplet-coupled pairs \(e,f\); the orbital pairs may overlap or be disjoint.}\\
\ntrow{\(C_{c;pq}\)}{Conditioned rotation: \((E_{cc}-1)L_{pq}\), with \(c,p,q\) distinct.}\\
\pagebreak
\addlinespace[0.35em]
\multicolumn{2}{@{}l}{\textbf{Generator pools, Lie algebras, and orbital support}} \\*
\ntrow{\(\mathcal P_0,\ \mathcal P_1\)}{Real singlet and triplet orbital pair spaces: \(\operatorname{Sym}^2V\) and \(\bigwedge^2V\), respectively.}\\
\ntrow{\(\mathcal K_2\)}{Full-Fock-space real vector space of pure two-body, number-conserving, spin-scalar, real-skew controls.}\\
\ntrow{\(K(\eta,A_0,A_1)\)}{General admissible real at-most-two-body seed: one-body coefficient \(\eta\) and singlet/triplet pair coefficients \(A_0,A_1\).}\\
\ntrow{\(\mathcal J_1,\ \mathcal W\)}{Orbital-induced triplet subspace \(\mathcal J_1=j_1(\so(V))\) and its Frobenius-orthogonal complement \(\mathcal W\) in \(\so(\mathcal P_1)\).}\\
\ntrow{\(P_{\mathcal W}\)}{Explicit coefficient-space projection onto the effective boundary triplet module \(\mathcal W\).}\\
\ntrow{\(G,\ \operatorname{Ed}(G),\ \Lpool(G)\)}{Orbital graph, its edge set, and the real span of its edge rotations \(L_{pq}\); \(\Lpool\) denotes the complete-graph span.}\\
\ntrow{\(\Xpool,\ \Ypool,\ \Ypooldisj\)}{Real spans of all \(X\) coordinates, all \(Y\) coordinates, and disjoint \(Y\) coordinates, respectively; see Supplemental Material, Sec.~1.2.}\\
\ntrow{\(\Pi_{M,N,S},\ K^{(M,N,S)}\)}{Highest-weight projector and restricted generator: \(K^{(M,N,S)}=\Pi_{M,N,S}K\Pi_{M,N,S}\).}\\
\ntrow{\(\mathfrak g_{\mathrm{orb}}\)}{Represented orbital algebra on the target sector; connected orbital edges generate its full image.}\\
\ntrow{\(\rho_{M,N,S}\)}{Orbital representation \(\so(M)\longrightarrow\so(\Hsp)\).}\\
\ntrow{\(\mathfrak g_{M,N,S}(\mathcal A),\ \Lie\)}{Restricted dynamical Lie algebra (DLA) of pool \(\mathcal A\); \(\Lie\) denotes real Lie closure under commutators.}\\
\ntrow{\(\mathfrak g_{LX}\)}{DLA generated by the orbital backbone and the full \(X\) pool on the target multiplicity space.}\\
\ntrow{\(\so(d)\)}{Real skew-symmetric matrix algebra; the full-DLA target for real-state controllability.}\\
\ntrow{\(s_{\min}(M,N,S)\)}{For an incomplete orbital backbone: minimum number of spatial orbitals supporting one admissible completing seed, before sector restriction.}\\
\addlinespace[0.35em]
\multicolumn{2}{@{}l}{\textbf{Repeated proof notation}} \\*
\ntrow{\(K_R,\ K_{\min}\)}{Minimum-seniority slice with singly occupied set \(R\), \(|R|=r\); \(K_{\min}=\bigoplus_{|R|=r}K_R\).}\\
\ntrow{\(D,\ Q_0\)}{Slice dimension \(D=\binom{b+h}{b}\); initial supported control block \(Q_0\) has dimension three and contains \(v_\lambda\).}\\
\ntrow{\(v_\lambda\)}{Standard highest-weight occupation vector, with \(\lambda=(2^b,1^r,0^h)\); cyclic for the orbital algebra.}\\
\ntrow{\(Z_p,\ D_a\)}{Occupation shifts: \(Z_p=1-n_p\), \(D_a=E_{aa}-1=-Z_a\).}\\
\ntrow{\(q_a\)}{Single-occupation projector: \(q_a=1-D_a^2=E_{aa}(2-E_{aa})\).}\\
\end{longtable}

\noindent
\begin{minipage}{\textwidth}
\small
\textit{Boundary key.} \(b=0\): electron maximal-spin boundary, with
\(k=N\); \(h=0\): hole maximal-spin boundary, with \(k=2M-N\).
Here \(k\) counts spinless particles or holes in the boundary reduction.
\(b,h\ge1\): interior (including nontrivial singlet sectors).
The counts \((b,r,h)\) describe the minimum-seniority highest-weight
occupation pattern, not a common occupation pattern of every state in
the sector. CSF denotes configuration state function.
\end{minipage}
\par\medskip
\endgroup

%% file: appendices/01_preliminaries.tex
\label{app:preliminaries}

\subsection{Sector construction and magnetic-component equivalence}
\label{app:magnetic-equivalence}
Let \(V_{\R}=\R^M\), let \(V_{\C}=V_{\R}\otimes_{\R}\C\simeq\C^M\), and let
\(W=\C^2=\Span_{\C}\{\ket\alpha,\ket\beta\}\).  The one-electron space is
\(V_{\C}\otimes_{\C}W\). The \(N\)-electron fermionic space is
\[
\Fsp_N=\bigwedge_{\C}^{N}(V_{\C}\otimes_{\C}W),
\qquad \dim_{\C}\Fsp_N=\binom{2M}{N}.
\]
For an allowed total spin \(S\) and magnetic quantum number \(m\), define
\[
\Hsp^{\C}_{M,N,S,m}
=\{\psi\in\Fsp_N:\widehat S^2\psi=S(S+1)\psi,
\ \widehat S_z\psi=m\psi\}.
\]
The full spin-\(S\) space decomposes as
\[
\Hsp^{\C,\mathrm{full}}_{M,N,S}
\simeq \mathcal M_{M,N,S}\otimes\mathcal V_S,
\]
where \(\mathcal V_S\) is the \((2S+1)\)-dimensional irreducible spin
representation and \(\mathcal M_{M,N,S}\) is the multiplicity space.

The generator families studied below conserve particle number and commute
with the full spin action, equivalently with all three components
\(\widehat S_x,\widehat S_y,\widehat S_z\).  They therefore preserve the
simultaneous quantum numbers \(N\), \(S\), and \(M_S=m\).

Fix an ordering of the spin orbitals and let \(J_{\mathrm{occ}}\) be coefficientwise complex
conjugation in the resulting occupation basis.  The real Fock space and real
spin sector are
\[
\Fsp_N^{\R}=\bigwedge_{\R}^{N}(\R^M\otimes_{\R}\R^2),
\qquad
\Hsp^{\R}_{M,N,S,m}=\Hsp^{\C}_{M,N,S,m}\cap\Fsp_N^{\R}.
\]
The matrices of \(\widehat S_z\) and
\(\widehat S_- =\sum_p a_{p\beta}^{\dagger}a_{p\alpha}\) are real in this
basis, as is \(\widehat S_+=\widehat S_-^{\dagger}\).  Hence
\(\widehat S^2=\widehat S_z^2+\widehat S_z+
\widehat S_-\widehat S_+\) commutes with \(J_{\mathrm{occ}}\).  Since the eigenvalues
\(S(S+1)\) and \(m\) are real, the real and imaginary parts of each vector
in \(\Hsp^{\C}_{M,N,S,m}\) satisfy the same eigenvalue equations.  Thus
\(\Hsp^{\C}_{M,N,S,m}\) is the complexification of
\(\Hsp^{\R}_{M,N,S,m}\).

The multiplicity dimension is independent of the magnetic component. With
\[
N_\alpha=\frac N2+S,
\qquad
N_\beta=\frac N2-S,
\]
the multiplicity dimension is
\begin{equation}
\label{eq:dimension}
d_{M,N,S}=
\binom{M}{N_\alpha}\binom{M}{N_\beta}
-\binom{M}{N_\alpha+1}\binom{M}{N_\beta-1},
\end{equation}
where out-of-range binomial coefficients are zero
\cite[Chap.~2, pp.~34--79]{Helgaker2000}.

For the proof notation, \(\Pi_{M,N,S}\) projects onto the complete
highest-weight space and \(K^{(M,N,S)}=\Pi_{M,N,S}K\Pi_{M,N,S}\).
Write \(\Lpool(G)=\Span_{\R}\{L_{pq}:\{p,q\}\in\operatorname{Ed}(G)\}\), with
\(\Lpool\) for the complete graph, and
\(\mathfrak g_{M,N,S}(\mathcal A)=\Lie\{K^{(M,N,S)}:K\in\mathcal A\}\).
Restrictions are implicit in Lie-algebra expressions using these pools.
We abbreviate the represented orbital action by \(\rho=\rho_{M,N,S}\)
when the sector labels are fixed.

\begin{lemma}[Equivalence of magnetic components]
\label{lem:magnetic-equivalence}
Let every \(K\) in a set \(\mathcal A\) conserve particle number and commute
with the full spin action.  For every allowed \(m\), the restrictions of
\(\mathcal A\) to \(\Hsp^{\C}_{M,N,S,m}\) and to the highest-weight space
\(\Hsp^{\C}_{M,N,S,S}\) are unitarily conjugate, and conjugation identifies
their generated Lie algebras.  If, in addition, every \(K\in\mathcal A\)
commutes with \(J_{\mathrm{occ}}\) and is skew-adjoint, its restrictions preserve the real
spin sectors; the conjugacy is real orthogonal and identifies their real
dynamical Lie algebras inside \(\so(d_{M,N,S})\).
\end{lemma}

\begin{proof}
The normalized lowering map
\[
U_m=c_{S,m}(\widehat S_-)^{S-m}:
\Hsp^{\C}_{M,N,S,S}\longrightarrow\Hsp^{\C}_{M,N,S,m}
\]
is unitary for
\[
c_{S,m}=\left[\frac{(S+m)!}{(2S)!(S-m)!}\right]^{1/2}.
\]
Indeed, on every copy of the spin-\(S\) representation, the squared norm of
\((\widehat S_-)^{S-m}\ket{S,S}\) is
\((2S)!(S-m)!/(S+m)!\).  Thus \(U_m\) is unitary on the whole multiplicity
space, including when its dimension is greater than one.  Number
conservation and spin commutation make both restrictions well defined, and
\([K,\widehat S_-]=0\) gives
\(K^{(m)}U_m=U_mK^{(S)}\), hence
\(K^{(m)}=U_mK^{(S)}U_m^{-1}\).  Conjugation preserves real linear
combinations and commutators.

For the real claim, \(J_{\mathrm{occ}}\widehat S_-=\widehat S_-J_{\mathrm{occ}}\) and \(c_{S,m}\) is real,
so \(U_m\) maps \(\Hsp^{\R}_{M,N,S,S}\) onto
\(\Hsp^{\R}_{M,N,S,m}\).  Its complex unitarity makes this restriction a
real isometry, represented by an orthogonal matrix in orthonormal real CSF
bases.  Each \(K\) commuting with \(J_{\mathrm{occ}}\) restricts to a real operator; if it
is also skew-adjoint, that restriction is real skew-symmetric.  The displayed
intertwining identity therefore identifies the real Lie closures by
orthogonal conjugation.  This applies in particular to all admissible
generator families considered here.
\end{proof}

\subsection{Orbital conjugacy and boundary seed orbits}
\label{app:orbital-transport}
On the boundary exterior-power space, write \(R_k(g)=\bigwedge^k g\)
and \(\rho_k=dR_k=d\Gamma_k\) for its differential, the orbital
one-body action on \(\bigwedge^k V_{\R}\).
\begin{lemma}[Orbit transport]
\label{lem:seed-orbit}
If \(\mathfrak g_Z=\Lie(\rho_k(\so(M)),Z)\), then
\(R_k(g)ZR_k(g)^{-1}\in\mathfrak g_Z\) for every \(g\in SO(M)\).  When
\(M\ge4\), elementary seeds form two signed-permutation orbits:
overlapping seeds \(\pm D_{cp,cq}\) and disjoint seeds
\(\pm D_{ab,cd}\).  For \(M=3\), only the overlapping orbit exists.
\end{lemma}

\begin{proof}
For \(X\in\rho_k(\so(M))\), the space \(\mathfrak g_Z\) is invariant under
\(\ad_X\), hence under \(\e^{t\ad_X}\) by Cayley--Hamilton.  Since connected
\(SO(M)\) is generated by exponentials, the conjugacy statement follows.
When \(M\ge4\), signed permutations act transitively within each
intersection type and preserve whether the two orbital pairs overlap; for
\(M=3\), two distinct pairs necessarily overlap.  If the underlying
permutation has determinant \(-1\), changing one coordinate sign makes the signed permutation
have determinant \(+1\); this changes the transported seed by at most an
irrelevant overall sign.  Thus the required representatives lie in \(SO(M)\),
not merely in \(O(M)\).
\end{proof}

\subsection{Transport of perfect-pairing coordinates}
\label{app:pair-transport}
\begin{lemma}[One seed gives the full coordinate pool]
\label{lem:px-orbit}
For every \(p\ne q\),
\(\Lie(\rho(\so(M)),X_{pq})\) contains every coordinate \(X_{rs}\)
with \(r\ne s\).
\end{lemma}

\begin{proof}
The generated algebra is invariant under \(SO(M)\) conjugation.  A signed
permutation sends \(P_i^\dagger\) to \(P_{\pi(i)}^\dagger\), because the orbital
sign is squared, and hence sends \(X_{ij}\) to \(X_{\pi(i)\pi(j)}\), up to the
irrelevant orientation sign.  Such transformations are transitive on unordered
orbital pairs.  As in \cref{lem:seed-orbit}, a coordinate sign can be adjusted
so that the signed permutation has determinant \(+1\); for \(M=2\) there is
only one unordered pair and no transport is needed.
\end{proof}

The connected-backbone commutator is stated in the main text,
Eq.~\ref{eq:connected-backbone}; it does not require a separate
representation-faithfulness assumption.

%% file: appendices/02_pair_space.tex
\subsection{Complete pair-space representation of pure two-body spin scalars}
\label{app:pair-space}

Symmetric and antisymmetric two-particle tensors provide a standard
spin-adapted representation of one- and two-body
operators~\cite{LiPaldus1993}. Here we use normalized pair bases and real,
anti-Hermitian generators to prove the full-Fock-space decomposition in
\cref{eq:pair-space-decomposition} and describe its restriction to the
target sector.

Let \(V=\R^M\), and write \(V_{\C}=V\otimes_{\R}\C\).  The Cauchy
decomposition for two fermions gives
\begin{equation}
\bigwedge\nolimits^2\!\left(V_{\C}\otimes\C^2\right)
\simeq
\left(\operatorname{Sym}^2V_{\C}\right)\otimes\mathcal V_0
\oplus
\left(\bigwedge\nolimits^2V_{\C}\right)\otimes\mathcal V_1,
\label{eq:app-two-electron-decomposition}
\end{equation}
where \(\mathcal V_0\) and \(\mathcal V_1\) are the spin-zero and spin-one
irreducible representations.  The corresponding real orbital pair spaces are
\[
\mathcal P_0=\operatorname{Sym}^2V,
\qquad
\mathcal P_1=\bigwedge\nolimits^2V.
\]
Use the normalized singlet creators
\begin{equation}
P_{pq}^{0\dagger}
=\frac{a_{p\alpha}^{\dagger}a_{q\beta}^{\dagger}
-a_{p\beta}^{\dagger}a_{q\alpha}^{\dagger}}
{\sqrt{2(1+\delta_{pq})}},
\qquad p\le q,
\label{eq:app-singlet-pair-basis}
\end{equation}
and the triplet creators \(P^{1\dagger}_{pq,m}=T^\dagger_{pq,m}\) for
\(p<q\).

For \(p<q\), let \(T_{pq,m}^\dagger\), \(m=-1,0,+1\), be the usual triplet-pair
creation operators,
\begin{align*}
T_{pq,+1}^\dagger&=a_{p\alpha}^\dagger a_{q\alpha}^\dagger,\\
T_{pq,0}^\dagger&=\frac{a_{p\alpha}^\dagger a_{q\beta}^\dagger
+a_{p\beta}^\dagger a_{q\alpha}^\dagger}{\sqrt2},\\
T_{pq,-1}^\dagger&=a_{p\beta}^\dagger a_{q\beta}^\dagger.
\end{align*}
  Pair-basis labels \(\mu,\nu\) in \(\mathcal P_0\) represent symmetric
orbital pairs \((p,q)\) with \(p\le q\), while labels in \(\mathcal P_1\)
represent oriented antisymmetric pairs with \(p<q\).  When unsorted labels are convenient, we
extend them by
\[
P^{0\dagger}_{qp}=P^{0\dagger}_{pq},
\qquad
P^{1\dagger}_{qp,m}=-P^{1\dagger}_{pq,m}.
\]
Define
\begin{align}
Q^{(j)}_{\mu\nu}
&=\sum_{m=-j}^{j}P^{j\dagger}_{\mu,m}P^j_{\nu,m},
\label{eq:app-pair-matrix-unit}\\
K^{(j)}_{\mu\nu}
&=Q^{(j)}_{\mu\nu}-Q^{(j)}_{\nu\mu}.
\label{eq:app-pair-skew-coordinate}
\end{align}
For \(j=0\), the sum in \cref{eq:app-pair-matrix-unit} has the single term
\(m=0\).  More generally, for a real matrix
\(A\in\End(\mathcal P_j)\), let
\begin{equation}
\operatorname{Lift}_j(A)
=\sum_{\mu,\nu}A_{\mu\nu}Q^{(j)}_{\mu\nu}.
\label{eq:app-pair-lift}
\end{equation}

\begin{proposition}[Complete pure two-body control space]
\label{prop:pair-space-completeness}
Let \(\mathcal K_2\) be the real vector space of pure, normal-ordered
two-body, particle-number-conserving operators that commute with the full
spin action and are real skew-symmetric in the occupation basis.  Then the
map
\[
(A_0,A_1)\longmapsto
\operatorname{Lift}_0(A_0)+\operatorname{Lift}_1(A_1)
\]
is a real-linear isomorphism
\begin{equation}
\so(\mathcal P_0)\oplus\so(\mathcal P_1)
\xrightarrow{\ \simeq\ }\mathcal K_2.
\label{eq:pair-space-decomposition}
\end{equation}
Consequently,
\[
\dim_{\R}\mathcal K_2
=\binom{M(M+1)/2}{2}+\binom{M(M-1)/2}{2}
=\frac{M^2(M^2-1)}4.
\]
\end{proposition}

\begin{proof}
A pure normal-ordered number-conserving two-body operator is uniquely
determined by its restriction to the two-electron space: the quartic
monomials act there as matrix units between antisymmetric two-electron basis
states.  Hence no nonzero pure two-body coefficient tensor can vanish on the
entire two-electron space.

By \cref{eq:app-two-electron-decomposition}, an operator commuting with the
full spin action has, on the two-electron space, the block form
\[
A_0\otimes I_{\mathcal V_0}
\oplus
A_1\otimes I_{\mathcal V_1},
\qquad
A_j\in\End(\mathcal P_j).
\]
Indeed, the spin-zero and spin-one irreducible representations are
inequivalent, and Schur's lemma leaves an arbitrary operator on each orbital
multiplicity space and the identity on its spin factor.  In the normalized
pair bases, \(Q^{(j)}_{\mu\nu}\) is precisely the second-quantized lift of the
corresponding matrix unit, contracted over the spin factor.  Requiring real
coefficients and skew symmetry is therefore equivalent to
\(A_j^{\mathsf T}=-A_j\).  Conversely, every pair of real skew matrices gives
through \cref{eq:app-pair-lift} a pure two-body, number-conserving,
full-spin-scalar real-skew operator.  Injectivity on the two-electron space
gives uniqueness, completing the proof.
\end{proof}

\noindent
The isomorphism in \cref{eq:pair-space-decomposition} is an isomorphism of
real coefficient spaces.  It is not a Lie-algebra isomorphism for
full-Fock-space operators, because a commutator of pure two-body operators
can contain higher-body terms.

\paragraph{Coordinate count by exact orbital support.}
In a fixed orbital basis, the skew pair-edge coordinates can be grouped by
the number of distinct spatial orbitals in their support.  On each chosen
two-orbital set there are three singlet edges; on each three-orbital set
there are six singlet and three triplet edges; and on each four-orbital set
there are three singlet and three triplet disjoint-pair edges.  Hence
\begin{equation}
\begin{aligned}
N^{\mathrm{coord}}_2&=3\binom M2,&
N^{\mathrm{coord}}_3&=9\binom M3,&
N^{\mathrm{coord}}_4&=6\binom M4,\\
\sum_{s=2}^{4}N^{\mathrm{coord}}_s
&=\frac{M^2(M^2-1)}4=\dim\mathcal K_2.
\end{aligned}
\label{eq:app-pair-support-count}
\end{equation}
These numbers count coordinate-basis directions of exact support two, three,
and four, respectively.  They are not asserted to define orbital-invariant
submodules.

\paragraph{Restriction is not a second direct-sum theorem.}
For the highest-weight fixed-spin sector, define the linear map
\[
\operatorname{res}_{M,N,S}:\mathcal K_2\longrightarrow\so(\Hsp),
\qquad
K\longmapsto\Pi_{M,N,S}K\Pi_{M,N,S}.
\]
This map can have a nontrivial kernel.  Moreover, even when the two
full-Fock-space coefficient blocks in
\cref{eq:pair-space-decomposition} are distinct, their restricted images
can overlap each other or the represented one-body orbital algebra.  An
explicit source of this redundancy is as follows.  For
\(\eta\in\so(V)\), let
\(j_0(\eta)=d(\operatorname{Sym}^2)(\eta)\) and
\(j_1(\eta)=d(\bigwedge^2)(\eta)\) be the induced actions on the two pair
spaces, and write
\[
d\Gamma(\eta)=\sum_{p,q}\eta_{pq}E_{pq}.
\]
Direct second quantization gives the full-Fock-space identity
\begin{equation}
\operatorname{Lift}_0\!\left(j_0(\eta)\right)
+\operatorname{Lift}_1\!\left(j_1(\eta)\right)
=(\widehat N-1)d\Gamma(\eta).
\label{eq:app-fixed-N-redundancy}
\end{equation}
On the \(N\)-electron sector, the right-hand side is
\((N-1)d\Gamma(\eta)\), an orbital one-body direction.  Thus the unique
coefficient decomposition before restriction should not be read as a
classification of linearly independent restricted controls or of arbitrary
single-seed Lie completion.

\paragraph{Dictionary for the generators used in the article.}
Choose consistent orientations for antisymmetric pair labels.  The diagonal
singlet-pair coordinates and the triplet block give
\[
K^{(0)}_{pp,qq}
=P_p^\dagger P_q-P_q^\dagger P_p
=X_{pq},
\qquad
K^{(1)}_{e,f}=Y_{e,f}.
\]
For pairwise distinct \(c,p,q\), resolving the spin of the electron on the
shared orbital \(c\) into the singlet and triplet channels yields
\begin{equation}
Q^{(0)}_{cp,cq}+Q^{(1)}_{cp,cq}=n_cE_{pq}.
\label{eq:app-shared-pair-resolution}
\end{equation}
Antisymmetrizing in \(p,q\) therefore gives
\begin{equation}
K^{(0)}_{cp,cq}+K^{(1)}_{cp,cq}=n_cL_{pq},
\qquad
C_{c;pq}=K^{(0)}_{cp,cq}+K^{(1)}_{cp,cq}-L_{pq}.
\label{eq:app-XYC-dictionary}
\end{equation}
Thus \(X\) and \(Y\) are individual skew coordinates in the singlet and
triplet pair blocks, respectively, while the pure two-body part of \(C\) is
an equal-weight shared-orbital direction across the two blocks.

%% file: appendices/appendix_b_boundary.tex
\section{Boundary completion and classification}
\label{app:boundary-completion}

\subsection{Perfect-pairing boundary criterion}
\label{app:px-boundary-proof}

\begin{theorem}[Perfect-pairing boundary theorem]
\label{thm:px-boundary}
For every physically allowed \((M,N,S)\), all projected \(X_{pq}\) with
\(p\ne q\) vanish if and only if \(b=0\) or \(h=0\).  If \(b,h\ge1\), every
such coordinate \(X_{pq}\) is nonzero.
\end{theorem}

\begin{proof}[Proof of \cref{thm:px-boundary}]
By magnetic-component equivalence, it is enough to work in the highest-weight
component \(M_S=S\).  If \(b=0\), then \(N_\beta=0\) there, so every pair
annihilator \(P_q\) vanishes.  If \(h=0\), then \(N_\alpha=M\); every orbital
already contains an \(\alpha\) electron, and pair creation on a different
orbital is Pauli blocked.  Thus \(X_{pq}=0\) on either boundary.

Conversely, suppose \(b,h\ge1\) and fix \(p\ne q\).  Choose a determinant in
which \(q\) is doubly occupied, \(p\) is empty, another \(b-1\) orbitals are
doubly occupied, \(r=2S\) orbitals are singly occupied by \(\alpha\)
electrons, and \(h-1\) orbitals are empty.  It obeys \(S_z=S\) and
\(S_+=0\), hence lies in the highest-weight spin-\(S\) space.  On this
determinant, \(P_p^\dagger P_q\ne0\) and
\(P_q^\dagger P_p=0\), so \(X_{pq}\ne0\).
\end{proof}

\subsection{CAR derivations for the boundary lemmas}
\label{app:boundary-car}

This appendix fixes the signs and distinct-index assumptions used in
the first-conditioned-rotations lemma and the hole-overlap equation.  In
this section only, spinless
labels denote mutually anticommuting fermionic modes and
\(A_{uv}^\dagger=a_u^\dagger a_v^\dagger\),
\(A_{uv}=a_v a_u\), with the displayed order defining the pair orientation.
The canonical anticommutation relations are
\[
\{a_i,a_j^\dagger\}=\delta_{ij},\qquad
\{a_i,a_j\}=\{a_i^\dagger,a_j^\dagger\}=0.
\]

On the electron boundary \(\Hsp\cong\bigwedge^k\R^M\), let
\(\rho_k=d(\bigwedge^k)\) denote the orbital action and
\(n_p^{\mathrm{sl}}=a_p^\dagger a_p\). Define
\begin{equation}
\label{eq:D}
D_{uv,xy}=A_{uv}^\dagger A_{xy}-A_{xy}^\dagger A_{uv}.
\end{equation}

\begin{lemma}[First conditioned rotations]
\label{lem:first-conditioned}
For pairwise distinct \(c,p,q\), the first identity holds; for pairwise
distinct \(a,b,c,d\), the second holds, up to orientation signs:
\[
D_{cp,cq}=n_c^{\mathrm{sl}}L_{pq},
\qquad
[L_{bc},D_{ab,cd}]=(n_c^{\mathrm{sl}}-n_b^{\mathrm{sl}})L_{ad}.
\]
Consequently, if \(2\le k\le M-2\), any individual intermediate-triplet
coordinate \(D_{e,f}\) with \(e,f\in\binom{[M]}2\) and \(e\ne f\),
together with the singles, generates \(n_c^{\mathrm{sl}}L_{pq}\) for all
pairwise distinct \(c,p,q\).
\end{lemma}

The following CAR calculations, together with the orbital transport of
Appendix~\ref{app:orbital-transport}, prove this lemma independently of
the coefficient-space classification.

\subsubsection{Overlapping and disjoint two-body seeds}

Let \(c,p,q\) be pairwise distinct.  Moving the rightmost \(a_c\) through
the two distinct modes gives
\begin{align*}
A_{cp}^\dagger A_{cq}
&=a_c^\dagger a_p^\dagger a_q a_c
=-a_c^\dagger a_p^\dagger a_c a_q
=a_c^\dagger a_c a_p^\dagger a_q
=n_c^{\mathrm{sl}}a_p^\dagger a_q.
\end{align*}
The reverse product is its Hermitian conjugate.  Subtraction gives
the full-Fock-space identity
\begin{equation}
D_{cp,cq}=n_c^{\mathrm{sl}}L_{pq}.
\label{eq:app-overlap-car}
\end{equation}

Now let \(a,b,c,d\) be mutually distinct.  From the CAR,
\begin{align*}
[L_{bc},a_x^\dagger]
&=\delta_{cx}a_b^\dagger-\delta_{bx}a_c^\dagger,\\
[L_{bc},a_x]
&=-\delta_{bx}a_c+\delta_{cx}a_b.
\end{align*}
Consequently,
\begin{align*}
[L_{bc},A_{ab}^\dagger]&=-A_{ac}^\dagger,
&[L_{bc},A_{cd}]&=A_{bd},\\
[L_{bc},A_{cd}^\dagger]&=A_{bd}^\dagger,
&[L_{bc},A_{ab}]&=-A_{ac}.
\end{align*}
Applying the ordinary Leibniz rule (both factors are even) gives
\begin{equation}
[L_{bc},D_{ab,cd}]=D_{ab,bd}-D_{ac,cd}.
\label{eq:app-disjoint-step}
\end{equation}
Pair antisymmetry and \cref{eq:app-overlap-car} give
\begin{align*}
D_{ab,bd}
&=-D_{ba,bd}
=-n_b^{\mathrm{sl}}L_{ad},\\
D_{ac,cd}
&=-D_{ca,cd}
=-n_c^{\mathrm{sl}}L_{ad}.
\end{align*}
Substitution into \cref{eq:app-disjoint-step} proves
\begin{equation}
[L_{bc},D_{ab,cd}]
=(n_c^{\mathrm{sl}}-n_b^{\mathrm{sl}})L_{ad}.
\end{equation}
Reversing any pair orientation multiplies the corresponding displayed seed
and the resulting identity by the same overall sign.

For completeness, consider the four local occupation states of modes \(p,q\).
The operator \(L_{pq}\) is zero on \(\ket{00}\) and \(\ket{11}\), and rotates
the span of \(\ket{10},\ket{01}\).  Hence
\begin{equation}
(n_p^{\mathrm{sl}}+n_q^{\mathrm{sl}})L_{pq}=L_{pq}
\label{eq:app-local-number}
\end{equation}
on the complete Fock space, not only on the nonzero support of a selected
matrix element.  Since \(L_{pq}\) conserves total particle number, restriction
to \(\bigwedge^k\R^M\) gives the operator identity
\begin{align}
\sum_{d\notin\{p,q\}}n_d^{\mathrm{sl}}L_{pq}
&=\left(\widehat N-n_p^{\mathrm{sl}}-n_q^{\mathrm{sl}}\right)L_{pq}\notag\\
&=(k-1)L_{pq}.                                    \label{eq:app-fixed-k}
\end{align}
Equations \eqref{eq:app-local-number}--\eqref{eq:app-fixed-k} justify the
recovery of one absolute conditioned rotation from the differences produced
by a disjoint seed.

\subsubsection{Shared-orbital seed on the hole boundary}

Let \(c,p,q\) be pairwise distinct and let \(\Pi_\alpha\) project onto the
filled \(\alpha\)-sea.  Pair orientations are fixed so that the final
\(\beta\)-spin hop is \(L_{pq}^{(\beta)}\); reversing a pair changes all three
displayed contributions by the same overall sign.  For \(m=+1\),
\[
\Pi_\alpha A_{cp,\alpha}^\dagger A_{cq,\alpha}\Pi_\alpha
=\Pi_\alpha n_{c\alpha}a_{p\alpha}^\dagger a_{q\alpha}\Pi_\alpha=0.
\]
After \(a_{q\alpha}\) creates a hole at \(q\), the mode \(p\alpha\) is still
occupied, so creation at \(p\alpha\) is Pauli blocked; the Hermitian-conjugate
term vanishes in the same way.

For \(m=0\), expand both normalized triplet operators.  Of the four products
in \(T_{cp,0}^\dagger T_{cq,0}\), only the term that removes and restores the
shared \(c\alpha\) electron returns to the filled \(\alpha\)-sea:
\begin{align*}
\Pi_\alpha T_{cp,0}^\dagger T_{cq,0}\Pi_\alpha
&=\frac12\Pi_\alpha
a_{c\alpha}^\dagger a_{p\beta}^\dagger
a_{q\beta}a_{c\alpha}\Pi_\alpha\\
&=\frac12a_{p\beta}^\dagger a_{q\beta}\Pi_\alpha.
\end{align*}
The other three products either attempt creation in an occupied \(\alpha\)
mode or leave an \(\alpha\) hole and are killed by the final projector.
Subtracting the reverse product gives
\[
\Pi_\alpha\bigl(T_{cp,0}^\dagger T_{cq,0}-\mathrm{h.c.}\bigr)\Pi_\alpha
=\tfrac12L_{pq}^{(\beta)}\Pi_\alpha.
\]

For \(m=-1\), \cref{eq:app-overlap-car} applies directly to the three
\(\beta\) modes:
\[
\Pi_\alpha\bigl(T_{cp,-1}^\dagger T_{cq,-1}-\mathrm{h.c.}\bigr)\Pi_\alpha
=n_{c\beta}L_{pq}^{(\beta)}\Pi_\alpha.
\]
With
\(h_p^\dagger=a_{p\beta}\), \(h_p=a_{p\beta}^\dagger\), and
\(n_p^{\mathrm h}=1-n_{p\beta}\), the CAR give
\(L_{pq}^{\mathrm h}=L_{pq}^{(\beta)}\) for \(p\ne q\).
Summing the three magnetic components therefore proves
\begin{equation}
\label{eq:hole-overlap}
\Pi_\alpha Y_{\{c,p\},\{c,q\}}\Pi_\alpha
=\left(n_{c\beta}+\frac12\right)L_{pq}^{(\beta)}
=\left(\frac32-n_c^{\mathrm h}\right)L_{pq}^{\mathrm h}.
\end{equation}

The singlet-single algebra already contains
\(\tfrac32L_{pq}^{\mathrm h}\); subtracting it isolates
\(-n_c^{\mathrm h}L_{pq}^{\mathrm h}\).  The irrelevant minus sign may be
removed by reversing the seed orientation.  Thus the local calculation gives
exactly the conditioned hole rotation needed by the Johnson-graph proof.

\subsection{Higher conditioning and half filling}
\label{app:higher-conditioning}
We give the full induction and the independent coordinate-completion
conclusion used in Section~\ref{sec:boundary-mechanism} and below.
Set \(n_\Omega^{\mathrm{sl}}=\prod_{a\in\Omega}n_a^{\mathrm{sl}}\).
\begin{lemma}[Higher-order conditioning]
\label{lem:si-higher-condition}
\label{lem:higher-condition}
Fix \(2\le k\le M-2\) and an intermediate-triplet coordinate \(D_{e,f}\)
with distinct pair labels \(e\ne f\). For every \(p\ne q\) and every
\(\Omega\subseteq[M]\setminus\{p,q\}\) with \(0\le|\Omega|\le k-1\),
\(\Lie(\rho_k(\so(M)),D_{e,f})\) contains
\(n_\Omega^{\mathrm{sl}}L_{pq}\).
\end{lemma}

\begin{proof}
The cases \(|\Omega|=0,1\) are the orbital rotations and first-conditioned rotations established above.  For \(|\Omega|\ge2\), take
\(c\in\Omega\).  Since \(|\Omega|\le k-1\) and \(k\le M-2\),
\[
|\Omega\cup\{p,q\}|\le k+1\le M-1,
\]
so one may choose \(s\notin\Omega\cup\{p,q\}\).  Induction gives
\[
n_{\Omega\setminus\{c\}}^{\mathrm{sl}}L_{ps},
\qquad n_c^{\mathrm{sl}}L_{sq},
\]
where each conditioner is disjoint from the two indices of its rotation.
Their commutator is, up to sign,
\(n_\Omega^{\mathrm{sl}}L_{pq}\).
For \(k=2\), the range \(|\Omega|\le k-1\) stops at the already proved base
case.  For \(k=M-2\), the displayed inequality still leaves at least one
choice of \(s\), so both endpoint values are included.
\end{proof}

\begin{theorem}[Electron-boundary single-seed theorem]
\label{thm:electron-boundary}
For \(2\le k\le M-2\) and every intermediate-triplet coordinate
\(D_{e,f}\) with distinct pair labels \(e\ne f\),
\[
\Lie\left(\rho_k(\so(M)),D_{e,f}\right)
=\so\!\left(\binom Mk\right).
\]
\end{theorem}

\begin{proof}
Orbital transport and \cref{lem:first-conditioned} give all first
occupation-conditioned rotations from any one shared or disjoint
coordinate. Lemma~\ref{lem:higher-condition} gives all
\(n_\Omega^{\mathrm{sl}}L_{pq}\) for \(p\ne q\) and
\(\Omega\subseteq[M]\setminus\{p,q\}\), up to \(|\Omega|=k-1\).
At that order, only the configurations \(I=\Omega\cup\{p\}\) and
\(J=\Omega\cup\{q\}\) survive, so the operator is
\(\pm F_{IJ}\), where \(F_{IJ}=\ket I\bra J-\ket J\bra I\).
The graph of \(k\)-subsets connected by these single replacements is
the connected Johnson graph \(J(M,k)\). For three distinct vertices,
\begin{equation}
\label{eq:connected-edges}
[F_{IJ},F_{JK}]=F_{IK}.
\end{equation}
Induction along paths therefore gives all elementary skew matrices,
which span \(\so(\binom Mk)\).
\end{proof}

The hole contraction \eqref{eq:hole-overlap} supplies the same conditioned
rotation after subtracting an available orbital term. Disjoint triplet
coordinates likewise reduce to disjoint hole transfers, up to sign.
Thus the theorem also holds on the hole boundary. In particular this
proof is independent of the coefficient-space criterion proved next.

When \(M=2k\), the Hodge star commuting with the bare \(\so(M)\) action causes
no exception.  For adjacent \(I,J\), the generated \(F_{IJ}\) obeys
\([F_{IJ},*]\ne0\), because \(*\ket I=\pm\ket{I^c}\) and
\(I^c\notin\{I,J\}\).  Thus the additional seed breaks the Hodge invariant;
the Johnson-graph argument already gives the full algebra.  The appearance of
special orthogonal or symplectic intermediate groups from particle--hole
invariants at half filling has a broader complex passive-linear-optics
precedent in Ref.~\cite{Oszmaniec2017}; the displayed commutator verifies that
the concrete real seed used here does not preserve the relevant Hodge
invariant.

\subsection{Complete classification of at-most-two-body boundary seeds}
\label{app:complete-boundary-classification}

This subsection proves the coefficient-space criterion in
\cref{eq:complete-boundary-criterion} for both maximal-spin boundaries.  Let
\(V=\R^M\), \(\mathcal P_1=\bigwedge^2V\), and let
\[
j_1:\so(V)\longrightarrow\so(\mathcal P_1)
\]
be the induced orbital action. Write
\(\ell_{pq}=e_{pq}-e_{qp}\in\so(V)\), where \(e_{pq}\) is the elementary
matrix on \(V\) and \(L_{pq}=d\Gamma(\ell_{pq})\). In the orthonormal
unordered-pair basis set
\[
\mathcal J_1=j_1(\so(V)),
\qquad
\mathcal W=\mathcal J_1^\perp.
\]
The matrix \(j_1(\ell_{pq})\) rotates the \(M-2\) pair-label planes
\(\{c,p\}\leftrightarrow\{c,q\}\).  These matrices are mutually orthogonal
and have squared Frobenius norm \(2(M-2)\).  Hence
\[
P_{\mathcal W}(A)
=A-\sum_{p<q}
\frac{\langle A,j_1(\ell_{pq})\rangle_F}{2(M-2)}j_1(\ell_{pq}),
\]
which proves \cref{eq:PW-explicit}.
Write \(P_{\mathcal J_1}=I-P_{\mathcal W}\) for the complementary
orthogonal projection.

\begin{lemma}[Effective triplet module]
\label{lem:boundary-effective-module}
For \(M\ge4\),
\[
\so(\mathcal P_1)=\mathcal J_1\oplus\mathcal W,
\qquad
\dim\mathcal W
=\frac{M(M-1)(M-3)(M+2)}8,
\]
and \(\mathcal W\) is irreducible as a real \(\so(V)\)-module.
\end{lemma}

\begin{proof}
Using \(\bigwedge^2V\simeq\so(V)\), the orbital action on
\(\mathcal P_1\) is the adjoint representation, and
\(\so(\mathcal P_1)\simeq\bigwedge^2(\operatorname{ad})\).  For \(M=5\)
and \(M\ge7\), the standard decomposition
\(\bigwedge^2(\operatorname{ad})=\operatorname{ad}\oplus X_2\)
has irreducible complementary summand \(X_2\)
\cite{Avetisyan2020X2}; its natural real form is \(\mathcal W\), hence is
real irreducible.

For \(M=4\), the Hodge splitting
\(\bigwedge^2\R^4=V_+\oplus V_-\), with \(\dim V_\pm=3\), gives
\[
\so(\bigwedge^2V)
=\so(V_+)\oplus\so(V_-)\oplus(V_+\otimes V_-).
\]
The first two summands form \(\mathcal J_1\), while
\(V_+\otimes V_-\) is the real irreducible \(3\otimes3\) module and has
dimension nine.  For \(M=6\), use
\(\so(6)\simeq\mathfrak{su}(4)\): the complexification of \(\mathcal W\)
is the sum of two inequivalent conjugate 45-dimensional irreducibles, which
complex conjugation exchanges.  A real invariant subspace would complexify
to a conjugation-stable sub-sum, so it is either zero or the whole module.
Thus the 90-dimensional real module is irreducible.  Finally, subtracting
\(\dim\mathcal J_1=\binom M2\) from
\(\dim\so(\mathcal P_1)=\binom{\binom M2}{2}\) gives the displayed
dimension.
\end{proof}

Let \(b_p^\dagger,b_p\) be the creation and annihilation operators of the
abstract spinless exterior-power space \(\bigwedge V\). On the electron
boundary they represent the surviving \(\alpha\) modes; before the
particle--hole change on the filled-\(\alpha\) sea, they represent the
\(\beta\) modes. For the spinless exterior-power representation, define
\begin{equation}
\gamma_k(A)=
\left.\sum_{p<q,\,r<s}
A_{pq,rs}b_p^\dagger b_q^\dagger b_s b_r
\right|_{\bigwedge^kV},
\qquad A\in\so(\bigwedge^2V).
\label{eq:app-spinless-lift}
\end{equation}
It is equivariant under orbital rotations, and counting the \(k-1\) pairs
containing each occupied orbital gives
\begin{equation}
[d\Gamma_k(\eta),\gamma_k(A)]=\gamma_k([j_1(\eta),A]),
\qquad
\gamma_k(j_1(\eta))=(k-1)d\Gamma_k(\eta).
\label{eq:app-lift-covariance}
\end{equation}

On the electron boundary, every singlet-pair annihilator removes a
\(\beta\) electron and therefore vanishes, while only the \(m=+1\) triplet
component survives.  This proves
\[
\rho_{\mathrm e}K=d\Gamma_k(\eta)+\gamma_k(A_1).
\]
The hole boundary is not obtained by simply declaring the singlet block
zero.  Let \(\Pi_\alpha\) project onto the filled-\(\alpha\) subspace.  A
representative mixed-spin contraction is
\begin{equation}
\Pi_\alpha
a^\dagger_{p\alpha}a^\dagger_{q\beta}
a_{s\beta}a_{r\alpha}\Pi_\alpha
=\delta_{pr}a^\dagger_{q\beta}a_{s\beta}\Pi_\alpha.
\label{eq:app-filled-alpha-contraction}
\end{equation}
After CAR reordering, all other mixed terms reduce in the same way.  Thus the
singlet block and the mixed-spin part of the triplet block contribute only a
real-skew one-body matrix \(H_{\mathrm{mix}}(A_0,A_1)\); the remaining
\(\beta\beta\) two-body kernel is \(A_1\).

To convert the \(\beta\) electrons to holes, set
\(d_p^\dagger=b_p\) and \(d_p=b_p^\dagger\).  Two CAR reorderings give
\begin{align}
d_pd_qd_s^\dagger d_r^\dagger
={}&\delta_{qs}\delta_{pr}-\delta_{ps}\delta_{qr}
-\delta_{qs}d_r^\dagger d_p+\delta_{ps}d_r^\dagger d_q\notag\\
&+\delta_{qr}d_s^\dagger d_p-\delta_{pr}d_s^\dagger d_q
+d_r^\dagger d_s^\dagger d_qd_p .
\label{eq:app-hole-car}
\end{align}
For a skew pair coefficient \(A\), the scalar contraction vanishes.  Writing
\(e_{uv}\) for the elementary matrix on \(V\), the one-body coefficient
matrix is
\begin{equation}
\mathfrak c(A)=
\sum_{p<q,\,r<s}A_{pq,rs}
\left(
-\delta_{qs}e_{rp}+\delta_{ps}e_{rq}
+\delta_{qr}e_{sp}-\delta_{pr}e_{sq}
\right).
\label{eq:app-hole-contraction}
\end{equation}
Consequently,
\[
\gamma^{(b)}(A)
=\gamma^{(d)}(-A)+d\Gamma_d(\mathfrak c(A)),
\]
and restriction to the \(k\)-hole sector yields
\begin{equation}
\rho_{\mathrm h}K
=d\Gamma_k(\eta_{\mathrm{eff}})+\gamma_k(-A_1),
\qquad
\eta_{\mathrm{eff}}
=\eta+H_{\mathrm{mix}}(A_0,A_1)+\mathfrak c(A_1).
\label{eq:app-general-hole-restriction}
\end{equation}
This proves the two boundary reduction formulas used in the main text.  In
particular, \(A_0\) generally does not vanish on the hole boundary; it changes
only the already available orbital backbone.

\begin{lemma}[Survival of the effective module]
\label{lem:boundary-quotient-injective}
For \(2\le k\le M-2\), the equivariant map
\[
\mathcal W\longrightarrow
\so(\bigwedge^kV)/d\Gamma_k(\so(V)),
\qquad A\longmapsto[\gamma_k(A)]
\]
is injective.
\end{lemma}

\begin{proof}
By \cref{lem:boundary-effective-module}, its kernel is either zero or all of
\(\mathcal W\).  It remains to show that the map is nonzero.  Let \(B\) be
the shared-pair coefficient for which
\(\gamma_k(B)=n_c^{\mathrm{sl}}L_{pq}\), with \(c,p,q\) distinct.  For
\(R\subseteq[M]\setminus\{p,q\}\), \(|R|=k-1\), the matrix element from
\(\ket{q,R}\) to \(\ket{p,R}\), apart from the common fermionic sign, is
\(\mathbf1_{c\in R}\).  The corresponding matrix element of any one-body
action \(d\Gamma_k(\eta)\) is instead the spectator-independent coefficient
\(\eta_{pq}\).  The bounds on \(k\) permit choices of \(R\) both containing
and omitting \(c\), so \(n_c^{\mathrm{sl}}L_{pq}\) is not a one-body
action.  By \cref{eq:app-lift-covariance},
\(\gamma_k(P_{\mathcal J_1}B)\) lies in the backbone.  Hence
\([\gamma_k(P_{\mathcal W}B)]\ne0\), and irreducibility excludes a nonzero
kernel.
\end{proof}

\begin{proof}[Proof of the boundary part of \cref{thm:main}]
Under the real boundary identification,
\(\mathfrak k=d\Gamma_k(\so(V))\) represents
\(\mathfrak g_{\mathrm{orb}}\), and \(\rho_{\mathrm{e/h}}K\) represents
\(K|_{\Hsp}\).
Write \(A_1=A_{\mathcal J_1}+A_{\mathcal W}\) according to
\cref{eq:triplet-effective-splitting}.  The electron and hole reduction
formulas, together with \cref{eq:app-lift-covariance}, show that modulo
\(\mathfrak k=d\Gamma_k(\so(V))\), the seed has class
\([\gamma_k(A_{\mathcal W})]\) or its negative.  If
\(A_{\mathcal W}=0\), the restricted seed is already in \(\mathfrak k\), a
proper subalgebra for \(2\le k\le M-2\), and cannot complete the DLA.

If \(A_{\mathcal W}\ne0\), orbital commutators and
\cref{lem:boundary-effective-module,eq:app-lift-covariance} generate
\(\gamma_k(\mathcal W)\).  The backbone already contains
\(\gamma_k(\mathcal J_1)\), so the generated algebra contains the lift of
every triplet-pair coefficient, in particular a shared coordinate whose
restriction is \(n_c^{\mathrm{sl}}L_{pq}\) up to an orbital term and an
overall sign.  The independently proved conditioned-rotation and
Johnson-graph construction in the preceding subsections then gives
\(\so\!\left(\binom Mk\right)\).  Finally,
\cref{lem:boundary-quotient-injective} shows that
\(P_{\mathcal W}(A_1)\ne0\) is equivalent to the restricted seed lying
outside \(\mathfrak k\).  This proves
\cref{eq:complete-boundary-criterion,eq:boundary-failure-subspace} on both
boundaries.
\end{proof}

The failure set is therefore a linear subspace of codimension
\(\dim\mathcal W\) in the coefficient space and is the same for every
nontrivial electron and hole maximal-spin boundary; only the exterior-power
representation of its surviving component depends on \(k\).

\begin{corollary}[Boundary image of the at-most-two-body control space]
\label{cor:boundary-two-body-image}
For \(M\ge4\) and \(2\le k\le M-2\), the vector-space image of all
admissible seeds in \cref{eq:at-most-two-body-seed} under either maximal-spin
boundary restriction is
\begin{equation}
\boxed{
\bigl\{\rho_{\mathrm{e/h}}K(\eta,A_0,A_1)\bigr\}
=\mathfrak k\oplus\gamma_k(\mathcal W),
\qquad
\dim\bigl\{\rho_{\mathrm{e/h}}K\bigr\}
=\binom{r_1}{2},
\quad r_1=\binom M2 .}
\label{eq:boundary-two-body-image}
\end{equation}
Thus, for fixed \(M\), the linear control-space dimension is independent of
\(k\) throughout \(2\le k\le M-2\). The sum in
\cref{eq:boundary-two-body-image} is a vector-space decomposition. Its image
equals \(\so(d)\) for \(k=2\) or \(M-2\), and is a proper subspace for
\(3\le k\le M-3\); in the latter case, a completing seed and the backbone
generate the remaining directions through iterated commutators.
\end{corollary}

\par\noindent\emph{Proof.}
The boundary reduction formulas place every image in
\(\mathfrak k+\gamma_k(\mathcal W)\), while
\(\gamma_k(\mathcal J_1)\subseteq\mathfrak k\). Conversely, choosing
\(\eta\) and \(A_1\) realizes arbitrary orbital terms and
\(\gamma_k(\mathcal W)\), with \(\eta\) absorbing the hole-side one-body
contribution. Injectivity of
\(\gamma_k\) on \(\mathcal W\) modulo \(\mathfrak k\) makes the sum direct.
The represented orbital action and \(j_1\) are faithful in the stated range,
so its dimension is
\(\dim\mathcal J_1+\dim\mathcal W=\dim\so(\mathcal P_1)=\binom{r_1}{2}\).
\qed

For completeness, the coordinate consequences can be read directly in pair
coefficient space.  A disjoint pair edge is orthogonal to every
\(j_1(\ell_{pq})\), hence lies in \(\mathcal W\).  A shared edge \(B\) has
\(\|B\|_F^2=2\) and overlaps only the corresponding orbital-induced
direction, with overlap magnitude two, so
\begin{equation}
\|P_{\mathcal W}B\|_F^2
=2-\frac{2}{M-2}
=\frac{2(M-3)}{M-2}>0.
\label{eq:app-shared-W-norm}
\end{equation}
The triplet coefficient of \(C_{c;pq}\) is the same shared edge. Thus every
shared or disjoint \(Y_{e,f}\) with \(e\ne f\), and every \(C_{c;pq}\) with
pairwise distinct \(c,p,q\), succeeds.
This coordinatewise fact does not mean that every nonzero coordinate sum
succeeds: with consistent orientations,
\[
\sum_{c\notin\{p,q\}}B_{cp,cq}=j_1(\ell_{pq})\in\mathcal J_1,
\]
so the sum lies exactly in the failure subspace.

\subsection{An alternative root-string description}
\begin{remark}
For \(E_{pq}=\sum_\sigma a_{p\sigma}^\dagger a_{q\sigma}\), the canonical
anticommutation relations give
\[
E_{pq}^2=2P_p^\dagger P_q,
\qquad
X_{pq}=\tfrac12(E_{pq}^2-E_{qp}^2).
\]
On the Schur module \(\mathbb S_{(2^b,1^r)}V\), the corresponding root string
has length at least two exactly when \(b,h>0\).
\end{remark}

%% file: appendices/appendix_c_interior.tex
\section{Full proof of interior full-pool completion}
\label{app:interior-full-proof}

We use the notation of the main text throughout.  On the fixed
\((N,S,M_S=S)\) multiplicity space, write
\[
\mathfrak g_{LX}
 =\Lie\bigl(\rho_{M,N,S}(\so(M)),\{X_{pq}:p<q\}\bigr).
\]
This appendix proves Theorem~\ref{thm:interior-full-pool} and its
single-coordinate consequence in \cref{thm:main}: hard-core-pair completion, strict
localization, a three-dimensional initial block, highest-weight cyclicity,
and one-sided infection.

\subsection{Conditional pair transfers}
\label{app:conditional-pairs}
Set \(Z_p=1-n_p\). We give the full spectator induction underlying
\cref{eq:px-conditional}.
\begin{lemma}[Conditional perfect-pairing commutator]
\label{lem:si-px-conditional}
\label{lem:px-conditional}
For distinct \(p,c,q\),
\begin{equation}
\label{eq:si-px-conditional}
[X_{pc},X_{cq}]=Z_cX_{pq}.
\end{equation}
Consequently, \(Z_\Omega X_{pq}\) belongs to the generated algebra for every
\(\Omega\subseteq[M]\setminus\{p,q\}\), where
\(Z_\Omega=\prod_{a\in\Omega}Z_a\).
\end{lemma}

\begin{proof}
The only nonzero overlapping terms are
\([
P_p^\dagger P_c,P_c^\dagger P_q]=Z_cP_p^\dagger P_q
\)
and its reverse counterpart, which is
\(-Z_cP_q^\dagger P_p\).  This proves \cref{eq:si-px-conditional}.  For the
higher-order statement, suppose \(c\in\Omega\).  Every
\(d\in\Omega\setminus\{c\}\) is distinct from \(p,c,q\), so
\([Z_d,X_{cq}]=0\).  The induction hypothesis is applicable to the pair
\((p,c)\), because
\(\Omega\setminus\{c\}\subseteq[M]\setminus\{p,c\}\), and therefore
\begin{align*}
 [Z_{\Omega\setminus\{c\}}X_{pc},X_{cq}]
 &=Z_{\Omega\setminus\{c\}}[X_{pc},X_{cq}]\\
 &=Z_{\Omega\setminus\{c\}}Z_cX_{pq}
 =Z_\Omega X_{pq}.
\end{align*}
Starting from \(|\Omega|=0\) and \cref{eq:si-px-conditional} for
\(|\Omega|=1\), this proves the claim for every allowed spectator set.  In
particular, no additional terms arise: all previously accumulated spectator
factors have indices disjoint from the two overlapping pair transfers in the
new commutator.
\end{proof}

\subsection{Hard-core-pair completion}
\label{app:hardcore-pairs}

Consider \(m\) sites occupied by exactly \(t\) hard-core pairs, with basis
\(\ket B\), \(B\subseteq[m]\), \(|B|=t\).  Let \(x_{pq}\) transfer a pair from
\(q\) to \(p\) skew-symmetrically, and set \(z_c=1-2\nu_c\), where \(\nu_c\) is
the pair-occupation projector.

\begin{lemma}[Hard-core-pair orthogonal algebra]
\label{lem:hardcore}
If \(1\le t\le m-1\), then
\[
\Lie\{x_{pq}:p\ne q\}=\so\!\left(\binom mt\right).
\]
\end{lemma}

\begin{proof}
Identifying empty and doubly occupied sites with effective qubit states,
this is the fixed-Hamming-weight construction of Proposition~1 in
Ref.~\cite{Stergiou2026}, expressed in pair-occupation variables with the
full pair-exchange pool.
The local identity \([x_{pc},x_{cq}]=z_cx_{pq}\) recursively generates every
\(z_Ax_{pq}\).  For \(C\subseteq[m]\setminus\{p,q\}\), \(|C|=t-1\), the
spectator projector
\[
\pi_C=\prod_{c\in C}\frac{1-z_c}{2}
\prod_{d\notin C\cup\{p,q\}}\frac{1+z_d}{2}
\]
gives \(\pi_Cx_{pq}\) as a linear combination of such conditioned terms and
isolates one edge of \(J(m,t)\). For distinct configurations,
\[
[F_{AB},F_{BC}]=F_{AC},\qquad F_{AB}=\ket A\bra B-\ket B\bra A.
\]
The graph is connected, so commutators along paths give every skew edge
and hence the full orthogonal algebra.
\end{proof}

\subsection{Strict localization on minimum-seniority slices}
\label{app:strict-localization}

Fix an \(r\)-subset \(R\subseteq[M]\), let \(E_R=[M]\setminus R\), and define
\[
K_R=\Span\{\ket{B;R}:B\subseteq E_R,\ |B|=b\}.
\]
Here \(B\) labels doubly occupied orbitals, each orbital in \(R\) carries one
\(\alpha\) electron, and the rest are empty.  Thus
\(
\dim K_R=D=\binom{b+h}{b}
\).

\medskip
\noindent\textbf{Occupation-seniority spin bound.}\quad
In a fixed occupation pattern with \(s\) singly occupied orbitals, empty
and doubly occupied orbitals are spin singlets, so the spin representation
is \((\C^2)^{\otimes s}\). Hence \(S\le s/2\). If \(s=2S\), the
highest-weight condition \(M_S=S=s/2\) forces every singly occupied orbital
to carry \(\alpha\) spin.
\medskip

\begin{proposition}[Strict localization on a slice]
\label{prop:slice-local}
For every \(R\),
\[
\so(K_R)\oplus0_{K_R^\perp}\subseteq\mathfrak g_{LX}.
\]
\end{proposition}

\begin{proof}
For \(p,q\in E_R\), set
\[
\widetilde X_{pq}^{(R)}=Z_{E_R\setminus\{p,q\}}X_{pq}.
\]
It belongs to the generated algebra by
Lemma~\ref{lem:px-conditional}.  If it is
nonzero on an occupation-pattern component, the spectator factors eliminate
that component if it has a singly occupied orbital in
\(E_R\setminus\{p,q\}\), while \(X_{pq}\) annihilates it if either \(p\) or
\(q\) is singly occupied.  Hence every singly occupied orbital of a surviving
component lies in \(R\), so its number \(s\) satisfies \(s\le |R|=r\).
On the other hand, the target sector has \(S=r/2\), and the
occupation-seniority spin bound just established gives \(s\ge2S=r\).  Thus
\(s=r\), the
singly occupied set is exactly \(R\), and in the chosen \(M_S=S\)
highest-weight component all of these orbitals carry \(\alpha\) spin.  The
particle-number identity \(N=2b+r\) then forces exactly \(b\) doubly occupied
orbitals in \(E_R\).  Every surviving occupation component therefore belongs
to \(K_R\), and the operator vanishes on \(K_R^\perp\).
On \(K_R\),
\(
\widetilde X_{pq}^{(R)}=(-1)^{b-1}x_{pq}
\), and \cref{lem:hardcore} gives the full supported \(\so(K_R)\).
\end{proof}

Let \(K_{\min}=\bigoplus_{|R|=r}K_R\).

\subsection{A three-dimensional initial block}
\label{app:three-dimensional-block}

The standard highest-weight vector is
\begin{equation}
\label{eq:highest-vector}
v_\lambda=P_1^\dagger\cdots P_b^\dagger
a_{b+1,\alpha}^\dagger\cdots a_{b+r,\alpha}^\dagger\ket0.
\end{equation}
It lies in \(K_{R_\lambda}\), where
\(R_\lambda=\{b+1,\ldots,b+r\}\), empty when \(r=0\).

\begin{lemma}[Three-dimensional initial block]
\label{lem:three-dimensional-block}
If \(b,h\ge1\), there is a three-dimensional subspace
\(Q_0\subseteq\Hsp\) containing \(v_\lambda\) such that
\[
\so(Q_0)\oplus0_{Q_0^\perp}\subseteq\mathfrak g_{LX}.
\]
\end{lemma}

\begin{proof}
Choose \(p\in\{1,\ldots,b\}\) and
\(q\in\{b+r+1,\ldots,M\}\). The normalized state
\(\ket\chi=P_pv_\lambda\) has both \(p\) and \(q\) empty. Define
\begin{align*}
u&=P_p^\dagger\ket\chi=v_\lambda,
&v&=P_q^\dagger\ket\chi,\\
w&=P_{pq}^{0\dagger}\ket\chi
=\frac{a_{p\alpha}^\dagger a_{q\beta}^\dagger
-a_{p\beta}^\dagger a_{q\alpha}^\dagger}{\sqrt2}\ket\chi.
\end{align*}
These states are orthonormal. The spectator state has spin
\(S=r/2\) and projection \(M_S=S\), and each added two-electron state is a
singlet. Hence all three states lie in the same complete target sector.
Set \(Q_0=\Span_{\R}\{u,v,w\}\). Although \(w\) lies outside
\(K_{R_\lambda}\), no restriction to that slice is required for \(Q_0\).

For orthonormal vectors \(x,y\), write
\(F_{xy}=\ket x\bra y-\ket y\bra x\), extended by zero on their orthogonal
complement in \(\Hsp\). Since \(u,v\in K_{R_\lambda}\),
\cref{prop:slice-local} gives \(A=F_{uv}\in\mathfrak g_{LX}\).
The canonical anticommutation relations yield
\[
L_{pq}u=-\sqrt2\,w,\qquad L_{pq}v=\sqrt2\,w.
\]
The commutator \(B=[A,L_{pq}]\) is strictly supported on \(Q_0\).
Indeed, for every \(z\in Q_0^\perp\), one has \(Az=0\) and
\[
\langle u,L_{pq}z\rangle=-\langle L_{pq}u,z\rangle=0,
\qquad
\langle v,L_{pq}z\rangle=0.
\]
Thus \(AL_{pq}z=0\), so \(Bz=0\); skew-symmetry also places the range of
\(B\) in \(Q_0\). Its action and one further commutator are
\begin{align}
B&=-\sqrt2\,(F_{uw}+F_{vw}),\notag\\
[A,B]&=-\sqrt2\,(F_{uw}-F_{vw}).
\label{eq:three-dimensional-rotations}
\end{align}
Together with \(A=F_{uv}\), these span the three independent plane
rotations on \(Q_0\), proving the supported \(\so(3)\) claim.
This construction also applies to \((M,N,S)=(2,2,0)\), where
\(Q_0=\Hsp\).
\end{proof}

\subsection{Highest-weight cyclicity}
\label{app:highest-weight-cyclicity}

By skew Schur--Weyl duality (equivalently, the exterior Cauchy decomposition)
\cite{FultonHarris1991},
\[
\Hsp_{\C}\cong\mathbb S_\lambda\C^M,
\qquad \lambda=(2^b,1^r),
\]
with standard highest-weight vector \(v_\lambda\) from
\eqref{eq:highest-vector}.

\begin{lemma}[Singlet-single cyclicity]
\label{lem:cyclicity}
The highest-weight vector is cyclic for the complexified orbital-rotation
algebra:
\[
U(\so(M,\C))v_\lambda=\mathbb S_\lambda\C^M.
\]
Consequently, the real singlet-single cyclic span of \(K_{\min}\) is all of
\(\Hsp\).
\end{lemma}

\begin{proof}
Let \(e_i=E_{i,i+1}\), \(f_i=E_{i+1,i}\), and
\(\ell_i=e_i-f_i\in\so(M,\C)\).  Grade the highest-weight module by the height
of \(\lambda-\mu\): if \(V_\mu\) is the weight space of weight \(\mu\), set
\[
\mathcal V^{\mathrm{ht}}_n=
\bigoplus_{\operatorname{ht}(\lambda-\mu)=n}V_\mu,
\qquad \mathcal V^{\mathrm{ht}}_0=\C v_\lambda.
\]
We use here the standard highest-weight structure of
\(\mathbb S_\lambda\C^M\) as an irreducible
\(\mathfrak{sl}(M,\C)\)-module.  The
Poincar\'e--Birkhoff--Witt theorem gives
\(\mathbb S_\lambda\C^M=U(\mathfrak n_-)v_\lambda\), where
\(\mathfrak n_-\) is the negative-root subalgebra.
The simple lowering operators \(f_i\) generate \(U(\mathfrak n_-)\) as an
associative algebra, and each increases weight height by one. Consequently,
for \(n\ge1\),
\begin{equation}
\label{eq:height-recursion}
\mathcal V^{\mathrm{ht}}_n=\sum_i f_i\mathcal V^{\mathrm{ht}}_{n-1}.
\end{equation}

Let \(\mathcal C=U(\so(M,\C))v_\lambda\).  Suppose inductively that
\(\mathcal V^{\mathrm{ht}}_m\subseteq\mathcal C\) for every \(m<n\).  For
\(w\in\mathcal V^{\mathrm{ht}}_{n-1}\), the raising operator satisfies
\(e_iw\in\mathcal V^{\mathrm{ht}}_{n-2}\) or \(e_iw=0\), while
\(\ell_iw\in\mathcal C\) because \(w\in\mathcal C\) and
\(\ell_i\in\so(M,\C)\).  The identity
\[
f_iw=e_iw-\ell_iw
\]
therefore puts \(f_iw\) in \(\mathcal C\).  Equation
\eqref{eq:height-recursion} completes the induction and proves complex
cyclicity.

No irreducibility of the restricted \(\so(M,\C)\) representation is assumed.

Finally, set
\[
\mathcal C_{\R}=U(\so(M,\R))v_\lambda\subseteq\Hsp.
\]
The vector \(v_\lambda\) and all represented orbital generators are real, so
\[
\mathcal C_{\R}\otimes_{\R}\C
=U(\so(M,\C))v_\lambda
=\Hsp_{\C}.
\]
Consequently \(\dim_{\R}\mathcal C_{\R}=\dim_{\R}\Hsp\), whence
\(\mathcal C_{\R}=\Hsp\).  Since \(v_\lambda\in K_{\min}\), the real
singlet-single cyclic span of \(K_{\min}\) is likewise all of \(\Hsp\).
\end{proof}

\subsection{One-sided infection}
\label{app:infection-proof}

We first record the matrix-algebra fact needed in the expansion step.

\begin{lemma}[Associative envelope]
\label{lem:assoc-envelope}
For \(n\ge3\), the associative algebra generated by the standard
representation of \(\so(n)\) is \(\End(\R^n)\), and the identity is a linear
combination of positive-length words in \(\so(n)\).
\end{lemma}

\begin{proof}
Let \(e_{ij}\) be the standard matrix units in \(\End(\R^n)\), and set
\(F_{ij}=e_{ij}-e_{ji}\). Distinct \(i,j,k\) give
\[
F_{ij}F_{jk}=e_{ik},
\qquad
e_{ii}=-\tfrac12(F_{ij}^2+F_{ik}^2-F_{jk}^2).
\]
Thus all matrix units occur.  Moreover,
\(\sum_{i<j}F_{ij}^2=-(n-1)I_n\).
\end{proof}

\begin{lemma}[One-sided infection]
\label{lem:infection}
Let \(\mathcal H\) be a finite-dimensional real inner-product space and
\(Q\subseteq\mathcal H\) a subspace with \(\dim Q\ge3\). Suppose
\(\mathfrak g\subseteq\so(\mathcal H)\) is a real Lie subalgebra containing
\(\so(Q)\oplus0\). If some \(H\in\mathfrak g\), written in
\(\mathcal H=Q\oplus Q^\perp\), has a nonzero
off-diagonal block
\[
H=\begin{pmatrix}H_{11}&-C^{\mathsf T}\\C&H_{22}\end{pmatrix},
\qquad Y=\im C,
\]
then \(\mathfrak g\) contains the supported algebra \(\so(Q\oplus Y)\).
\end{lemma}

\begin{proof}
For \(K\in\so(Q)\), write
\[
\widehat K=\begin{pmatrix}K&0\\0&0\end{pmatrix},
\qquad
\widehat H_{11}=\begin{pmatrix}H_{11}&0\\0&0\end{pmatrix},
\qquad
\mathcal C_T=\begin{pmatrix}0&-T^{\mathsf T}\\T&0\end{pmatrix}.
\]
Because \(H_{11}\in\so(Q)\), both \(\widehat K\) and
\(\widehat H_{11}\) belong to \(\mathfrak g\).  Direct block multiplication
gives
\[
[\widehat K,H]-[\widehat K,\widehat H_{11}]
=\mathcal C_{-CK}.
\]
Repeated commutation with supported elements of \(\so(Q)\) therefore gives
\(\mathcal C_{CK_1\cdots K_m}\in\mathfrak g\).  By
\cref{lem:assoc-envelope}, including its positive-length representation of
the identity, linear combinations of these words give
\[
\mathcal C_{CB}\in\mathfrak g
\qquad\text{for every }B\in\End(Q).
\]

The map \(C:Q\to Y\) is surjective.  Choose a linear right inverse
\(R:Y\to Q\), so that \(CR=I_Y\).  For any
\(T\in\Hom(Q,Y)\), the endomorphism \(B=RT\) obeys
\(CB=CRT=T\).  Hence every cross block \(\mathcal C_T\),
\(T\in\Hom(Q,Y)\), belongs to \(\mathfrak g\).

It remains to obtain a diagonal algebra supported only on \(Y\).  For
\(S,T\in\Hom(Q,Y)\),
\[
[\mathcal C_S,\mathcal C_T]=
\begin{pmatrix}
T^{\mathsf T}S-S^{\mathsf T}T&0\\
0&TS^{\mathsf T}-ST^{\mathsf T}
\end{pmatrix}.
\]
Choose a unit vector \(q\in Q\) and arbitrary \(y_1,y_2\in Y\), and set
\(S=y_1q^{\mathsf T}\), \(T=y_2q^{\mathsf T}\).  Then the \(Q\)-block
vanishes, while the \(Y\)-block is
\[
y_2y_1^{\mathsf T}-y_1y_2^{\mathsf T}.
\]
Such rank-two skew matrices span \(\so(Y)\).  Thus \(\mathfrak g\) contains
the supported \(\so(Q)\), all \(Q\)--\(Y\) cross blocks, and the supported
\(\so(Y)\), which together are exactly \(\so(Q\oplus Y)\).  The assumption
\(\dim Q\ge3\) is used only in \cref{lem:assoc-envelope}.
\end{proof}

\subsection{Completion of the interior theorem}
\label{app:interior-completion}

\begin{theorem}[Interior full-pool completion]
\label{thm:interior-full-pool}
If \(b,h\ge1\), then
\[
\Lie\bigl(\rho_{M,N,S}(\so(M)),
\{X_{rs}:1\le r<s\le M\}\bigr)=\so(d_{M,N,S}).
\]
\end{theorem}

\begin{proof}[Proof of \cref{thm:interior-full-pool}]
By \cref{lem:three-dimensional-block}, choose \(Q_0\) with
\(v_\lambda\in Q_0\), \(\dim Q_0=3\), and supported
\(\so(Q_0)\subseteq\mathfrak g_{LX}\).

Among subspaces \(Q\supseteq Q_0\) with supported
\(\so(Q)\subseteq\mathfrak g_{LX}\), choose one of maximal dimension;
such a choice exists because \(\Hsp\) is finite dimensional. If
\(Q\ne\Hsp\), \cref{lem:cyclicity} implies that some orbital generator
has a nonzero block \(C:Q\to Q^\perp\), since otherwise \(Q\) would
contain the entire cyclic span of \(v_\lambda\). Then
\cref{lem:infection} gives supported \(\so(Q\oplus\im C)\), contradicting
maximality because \(\im C\ne0\). Hence \(Q=\Hsp\) and
\(\mathfrak g_{LX}=\so(\Hsp)\).
\end{proof}

\begin{corollary}[Interior single-coordinate universality]
\label{thm:interior-single-seed}
Let \(G\) be a connected graph on \([M]\), and let \(b,h\ge1\). Then, for every
\(p\ne q\),
\[
\Lie\bigl(\Lpool(G),X_{pq}\bigr)=\so(d_{M,N,S}).
\]
\end{corollary}

\begin{proof}
Connected orbital rotations generate the full represented orbital algebra.
Lemma~\ref{lem:px-orbit} in Appendix~\ref{app:pair-transport} then puts
every \(X\) coordinate in the algebra generated by any one fixed
\(X_{pq}\). Theorem~\ref{thm:interior-full-pool} completes the proof.
\end{proof}

%% file: appendices/appendix_d_support_factorization.tex
\section{Minimum-support and exact-factorization proofs}
\label{app:support-factorization}

\subsection{Interior dimension gap}
\begin{lemma}[Interior dimension gap]
\label{lem:interior-dimension-gap}
If \(b,h\ge1\), then \(d_{M,N,S}>M\).
\end{lemma}

\begin{proof}
Recall
\[
K_{\min}=\bigoplus_{|R|=r}K_R.
\]
Since \(b,h\ge1\), we have \(1\le b\le M-1\), so
\[
\begin{aligned}
\dim K_{\min}&=\binom Mr\binom{b+h}{b}\\
&=\binom Mb\binom{M-b}{r}\ge\binom Mb\ge M.
\end{aligned}
\]
The nonzero singlet-split vector \(w\) constructed in the proof of
\cref{lem:three-dimensional-block} lies in \(\Hsp\) and has
\(r+2\) singly occupied orbitals, whereas every occupation-basis vector
in \(K_{\min}\) has exactly \(r\).  Thus \(w\perp K_{\min}\), and
\[
d_{M,N,S}\ge\dim K_{\min}+1\ge M+1>M.
\]
\end{proof}

\subsection{Two-orbital restriction on maximal-spin boundaries}
\label{app:two-orbital-boundary-proof}

\begin{lemma}[Two-orbital boundary restriction]
\label{lem:two-orbital-boundary}
Let \(K\) be a number-conserving operator genuinely supported on two spatial
orbitals \(p,q\).  Suppose its restriction preserves an electron or hole
maximal-spin boundary and is real skew there.  Then on that boundary
\[
K=\kappa L_{pq}
\]
for some \(\kappa\in\R\), with \(L_{pq}\) understood in the particle or hole
representation, respectively.
\end{lemma}

\begin{proof}[Proof of \cref{lem:two-orbital-boundary}]
After maximal-spin projection the two orbitals are two spinless modes.
Number conservation decomposes their local boundary space into
particle-number blocks of dimensions \(1,2,1\).  A real skew matrix vanishes
on each one-dimensional block, while the real skew matrices on the ordered
one-particle block \((\ket p,\ket q)\) form the one-dimensional space
\[
\R\begin{pmatrix}0&1\\-1&0\end{pmatrix}
=\R L_{pq}.
\]
Genuine two-orbital support excludes spectator-dependent coefficients hidden
in global projectors or occupation functions.  Tensoring with any spectator
configuration therefore gives the same constant multiple of this local
direction.  Particle--hole conjugation gives the hole-boundary statement.
\end{proof}

\subsection{Spin symmetry of the conditioned rotation}
\begin{proposition}[Symmetries of the conditioned single]
\label{prop:C-symmetry}
The operator in text Eq.~\ref{eq:conditioned-single} obeys
\begin{align*}
[C_{c;pq},\widehat N]&=0,\\
[C_{c;pq},\widehat S_\mu]&=0\quad(\mu=x,y,z),\\
C_{c;pq}^\dagger&=-C_{c;pq}.
\end{align*}
\end{proposition}

\begin{proof}
Write \(E_{rs}=d\Gamma(|r\rangle\langle s|\otimes I_2)\) and
\(\widehat S_\mu=d\Gamma(I_M\otimes\sigma_\mu/2)\).  The identity
\([d\Gamma(A),d\Gamma(B)]=d\Gamma([A,B])\) shows that all \(E_{rs}\) commute
with \(\widehat N\) and every \(\widehat S_\mu\).  Since
\(c\notin\{p,q\}\),
\([E_{cc},L_{pq}]=0\).  The product of the commuting Hermitian
\(E_{cc}-1\) and anti-Hermitian \(L_{pq}\) is anti-Hermitian.
\end{proof}

\subsection{Proof of the minimum-support theorem}
\label{app:min-support-proof}

\begin{proof}[Proof of \cref{thm:min-support}]
On a nontrivial boundary, \(d=\binom Mk>M\), and in the interior
\cref{lem:interior-dimension-gap} gives \(d>M\).  In either case,
\[
\dim\rho_{M,N,S}(\so(M))
\le \frac{M(M-1)}2
<\frac{d(d-1)}2
=\dim\so(d),
\]
so the orbital backbone is not already complete; no faithfulness assumption
on \(\rho_{M,N,S}\) is needed.

A one-orbital number-conserving spin scalar is scalar on each local
irreducible block
\[
(N,S)=(0,0),\qquad(1,\tfrac12),\qquad(2,0)
\]
by Schur's lemma \cite{FultonHarris1991}.  It is therefore diagonal in the
occupation basis and cannot be a nonzero real skew-symmetric matrix.  Every
useful seed has support at least two.

On a maximal-spin boundary,
\cref{lem:two-orbital-boundary} shows that every genuinely two-orbital
admissible real seed restricts to a multiple of \(L_{pq}\), already in the
orbital algebra.  Boundary support is therefore at least three.  The seed
\(C_{c;pq}\) is admissible by \cref{prop:C-symmetry} and has genuine
three-orbital support.  On the electron boundary,
\begin{equation}
C_{c;pq}^{\mathrm e}=(n_c^{\mathrm{sl}}-1)L_{pq},
\qquad
D_{cp,cq}=n_c^{\mathrm{sl}}L_{pq}
=C_{c;pq}^{\mathrm e}+L_{pq}.
\label{eq:C-electron-Y-equivalence}
\end{equation}
On the hole boundary, \cref{eq:hole-overlap} and direct projection give
\begin{equation}
\begin{aligned}
C_{c;pq}^{\mathrm h}
&=(1-n_c^{\mathrm h})L_{pq}^{\mathrm h},\\
\Pi_\alpha Y_{\{c,p\},\{c,q\}}\Pi_\alpha
&=\left(\frac32-n_c^{\mathrm h}\right)L_{pq}^{\mathrm h}
=C_{c;pq}^{\mathrm h}+\frac12L_{pq}^{\mathrm h}.
\end{aligned}
\label{eq:C-hole-Y-equivalence}
\end{equation}
If \(A\in\mathfrak g\), then
\(\Lie(\mathfrak g,Z)=\Lie(\mathfrak g,Z+A)\).  The differences in
\cref{eq:C-electron-Y-equivalence,eq:C-hole-Y-equivalence} lie in the
orbital algebra, so each \(C_{c;pq}^{\mathrm{bdry}}\) generates the same
Lie algebra as the corresponding overlapping intermediate-triplet
coordinate.  The coordinate boundary construction proves
\cref{eq:C-coordinatewise-universality}, attaining the boundary upper bound.
In the interior, \cref{thm:interior-single-seed} supplies the genuinely
two-orbital seed \(X_{pq}\), attaining the interior upper bound.
\end{proof}

Fixed discrete gates are outside the generator minimum in
Eq.~\ref{eq:minimum-support-definition}. For example, the two-orbital gate
\(U_{ab}=e^{i\pi n_an_b}\), with \(a\ne b\), preserves particle number,
full spin symmetry, and real amplitudes, and its conjugation of orbital
rotations supplies occupation-conditioned rotations on the maximal-spin
boundaries. Its phase family \(e^{i\theta n_an_b}\) does not preserve real
amplitudes at general angles.

\subsection{Spectral projectors and eight distinct commuting strings}
\label{app:exact-factorization}
The spectral projectors of \(n_c=E_{cc}\), with eigenvalues \(0,1,2\), are
\[
\Pi_0=\frac{(n_c-1)(n_c-2)}2,\qquad
\Pi_1=n_c(2-n_c),\qquad
\Pi_2=\frac{n_c(n_c-1)}2.
\]
They give the occupation-controlled Givens action in text
Eq.~\ref{eq:controlled-givens}. For completeness, the following verifies
both commutativity and distinctness in the eight-string factorization of
text Proposition~\ref{prop:eight-pauli}.
\begin{proof}
With Majoranas
\(\gamma_{j,0}=a_j+a_j^\dagger\) and
\(\gamma_{j,1}=-i(a_j-a_j^\dagger)\),
\[
a_r^\dagger a_s-a_s^\dagger a_r
=\tfrac12(\gamma_{r,0}\gamma_{s,0}
+\gamma_{r,1}\gamma_{s,1}).
\]
Up to its scalar coefficient \(\pm i/4\), each term is a quartic Majorana
monomial.  More generally, for
canonically ordered Majorana monomials
\(\Gamma_A=\prod_{j\in A}\gamma_j\) and
\(\Gamma_B=\prod_{j\in B}\gamma_j\), reordering gives
\[
\Gamma_A\Gamma_B
=(-1)^{|A||B|-|A\cap B|}\Gamma_B\Gamma_A.
\]
Here a quartic support has the form
\begin{align*}
A=\{&(c,\tau,0),(c,\tau,1),\\
     &(p,\sigma,\kappa),(q,\sigma,\kappa)\},\\
&\tau,\sigma\in\{\alpha,\beta\},\qquad \kappa\in\{0,1\}.
\end{align*}
The assumption \(c\notin\{p,q\}\) prevents a condition index from coinciding
with a hopping index.  Hence two such supports intersect in zero indices, in
the two indices of a shared condition or hopping bilinear, or in all four
indices.  Thus \(|A\cap B|\in\{0,2,4\}\), and the displayed sign is always
positive.  Moreover, the triple \((\tau,\sigma,\kappa)\) uniquely determines
the support \(A\), because \(c,p,q\) are pairwise distinct.  The
Jordan--Wigner map is injective on the Clifford monomial basis up to its
nonzero phase, so these supports give exactly eight distinct Pauli strings.
They all commute pairwise.
\end{proof}

%% file: appendices/appendix_f_methods.tex
\Needspace{24\baselineskip}
\section{Numerical methods, validation, and formal verification}
\label{app:numerical-methods}

\begin{table}[!htbp]
\centering\small
\caption{Sample and selection rules for the two main numerical comparisons and the supplementary boundary sweep.
All resources are parameter-bound ansatz-only CNOT counts.}
\label{tab:benchmark-protocols}
\setlength{\tabcolsep}{4pt}
\renewcommand{\arraystretch}{1.15}
\begin{tabular}{@{}lll@{}}
\toprule
\tblcell{.11\textwidth}{Location}
& \tblcell{.42\textwidth}{Samples and output selection}
& \tblcell{.38\textwidth}{Reported quantities}\\
\midrule
\tblcell{.11\textwidth}{Fig.~\ref{fig:benchmark-boundary-repair}}
& \tblcell{.42\textwidth}{Five VQD chains per checkpoint; one entire chain minimizes summed final VQD objectives.}
& \tblcell{.38\textwidth}{Maximum-root error and maximum-root CNOT count of that same chain.}\\
\tblcell{.11\textwidth}{SI}
& \tblcell{.42\textwidth}{Twenty VQD chains; all contribute to statistics. A representative minimizes summed variational root energies.}
& \tblcell{.38\textwidth}{Ensemble median/IQR and strict-success count; representative-chain maximum-root CNOT count.}\\
\tblcell{.11\textwidth}{Fig.~\ref{fig:benchmark-sota}}
& \tblcell{.42\textwidth}{Twenty starts per geometry/sector; select the lowest variational energy independently at each tested size.}
& \tblcell{.38\textwidth}{Selected task error and matched CNOT count: one state for H$_6$, larger error and summed CNOT count for TME, maxima over eight geometries for N$_2$.}\\
\bottomrule
\end{tabular}
\end{table}
\FloatBarrier

\subsection{Sector construction and sequential VQD}
\label{sec:vqd-protocol}

Full-Fock generators are restricted first to fixed \((N,M_S)\), then to
\(S^2=S(S+1)\). Their real-skew matrices are closed under commutators until
the span rank stabilizes. All DLA and variational calculations use complete
fixed-spin spaces without a point-group projector; the magnetic component
is specified for each task. Closure checks are reported in Supplemental
Material, Sec.~4.1.

Variational calculations use Q\textsuperscript{2}Chemistry~\cite{Fan2022Q2Chemistry}.
Sequential VQD minimizes
\begin{equation}
\Phi_n(\theta)
=\langle\psi_n(\theta)|H|\psi_n(\theta)\rangle
+\beta\sum_{j<n}|\langle\widetilde\psi_j|\psi_n(\theta)\rangle|^2.
\label{eq:vqd-objective}
\end{equation}
For each independent benchmark chain, \(\widetilde\psi_j\) is its accepted
lower-root variational output; each root starts from the prescribed reference
with its own initialized parameters. References and generator coordinates
are fixed before optimization. Exact states enter only the subsequent
energy-error and fidelity diagnostics.

\subsection{Benchmark systems and circuit families}
\label{sec:benchmark-systems}

The electron-boundary task is linear H$_4$/6-31G at 1.50~\AA{} adjacent
spacing, with \((M,N,S,M_S)=(8,4,2,2)\), \(d=70\), and roots 0--1.
The hole-boundary task is O$_2$/cc-pVDZ at 1.21~\AA{}, with a
\(\pi/\pi^\ast\) CAS(6e,4o), \(S=M_S=1\), \(d=6\), and roots 0--5.
Interior tasks use STO-3G CAS(6e,6o): linear H$_6$ at 1.50~\AA{} spacing,
TME at the geometry of Ref.~\cite{Burton2024TUPS}, and N$_2$ at eight
independently optimized bond lengths. H$_6$ and N$_2$ are singlets;
TME targets separate singlet and triplet ground states. All use \(M_S=0\),
with \(d=175\) for singlets and \(d=189\) for the TME triplet.

Circuit families are defined in Section~\ref{sec:numerical-protocol}.
One repetition means a backbone-plus-seed unit for the proposed products,
an alternating even or odd half-layer for QNP-Q/F, and a full even--odd
cycle for tUPS. All comparisons use fixed orbitals. Complete Hamiltonian,
reference-state, frozen-orbital, and coordinate specifications are given
with the optimization protocol in Supplemental Material, Sec.~4.4.

\subsection{Sampling and output selection}
\label{app:benchmark-optimization}
\label{sec:benchmark-optimization}
\label{sec:additional-controls}

The five-chain boundary scan and twenty-start interior comparison use the
sample and output-selection rules in Table~\ref{tab:benchmark-protocols}.
H$_4$ and O$_2$ use \(\beta=2\)~Ha. Double-precision searches use independent
initializations and L-BFGS-B with analytic gradients. Deterministic
initialization keys, budgets, supplementary strict-success criteria, and
record-adoption rules are specified in Supplemental Material,
Secs.~4.3--4.4.

For the interior comparison, the reported endpoint is the smallest compiled
CNOT count among tested circuits with task error at most 0.1~mHa.
The supplementary ensemble reports initialization statistics separately.

\subsection{Circuit resources and validation}
\label{sec:circuit-validation}

CNOT counts are obtained after binding optimized parameters and compiling
the complete ansatz with Q\textsuperscript{2}Chemistry under the common
all-to-all connectivity, wire-order, and decomposition protocol. Counts
exclude reference preparation, measurements, and VQD overlaps. They are
maxima over roots for H$_4$/O$_2$, a single-state count for H$_6$, a sum
for the two TME states, and maxima over N$_2$ geometries. The compilation
settings and independent state-replay checks are given in Supplemental
Material, Secs.~4.4 and~4.8.

\subsection{Formal verification}
\label{app:formal-verification}

We used Lean 4.33.1~\cite{deMoura2021Lean4} and mathlib
4.33.1~\cite{Mathlib2020} to verify the highest-weight real-sector forms of
the classification and minimum-orbital-support results in
Theorems~\ref{thm:main} and~\ref{thm:min-support}. This covers every fixed
interior pair-transfer coordinate \(X_{pq}\) with \(p\ne q\), the complete criterion for arbitrary admissible one- and
two-body additions on nontrivial maximal-spin boundaries, and the endpoint
sectors. The support minimum is verified for globally admissible generators
of arbitrary body rank.

\begingroup
\widowpenalty=10000
Additional verified results include the complete one- and two-body
coefficient representation, complex-unitary and real-orthogonal equivalence
of magnetic restrictions and their respective complex and real Lie spans,
and the sector-uniform construction in \cref{cor:sector-uniform-seed}.
The real Lean theorem assumes that each full-Fock generator conserves particle
number, commutes with the full spin action, and preserves the occupation-real
form. Its admissible-seed specialization also verifies that anti-Hermitian
generators restrict to real skew-adjoint operators and that their generated
real Lie algebras lie in \(\so(d_{M,N,S})\) in every magnetic component.
For \(C\),
we verified its occupation-controlled exponential and eight commuting
factors in the CAR/Majorana representation underlying
\cref{prop:eight-pauli}. The transitivity of the special orthogonal group
on normalized real states is also verified.
The positive-spin fixed-\(C\) extension in Supplemental Material,
Sec.~3, is not included in the Lean formalization.
\par
\endgroup

\begin{samepage}
\paragraph*{AI assistance.}
GPT-5.6 Sol assisted with selected intermediate mathematical derivations
and searches for small counterexamples; the authors independently verified
the resulting arguments and take full responsibility for the mathematical claims.
\par
\end{samepage}

%% file: supplementary/si_body.tex
\clearpage
\begingroup
\pagenumbering{arabic}
\makeatletter
\let\addcontentsline\@gobblethree
\makeatother
\pdfbookmark[0]{Supplementary Information}{supplementary-start}
\thispagestyle{plain}

\begin{center}
  {\Large\bfseries Supplementary Information for\par}
  \vspace{0.35em}
  {\Large\bfseries Minimal building blocks for molecular quantum circuits\\
  with exact spin symmetry\par}
  \vspace{1.1em}
  {\large Mengwei Liu\textsuperscript{1,2} and Zhenyu Li\textsuperscript{1,2,*}\par}
  \vspace{0.5em}
  {\normalsize \textsuperscript{1}State Key Laboratory of Precision and Intelligent Chemistry,\par
  University of Science and Technology of China, Hefei 230026, China\par}
  \vspace{0.35em}
  {\normalsize \textsuperscript{2}Hefei National Laboratory,\par
  University of Science and Technology of China, Hefei 230088, China\par}
  \vspace{0.35em}
  {\normalsize \textsuperscript{*}Contact author:
  \href{mailto:zyli@ustc.edu.cn}{zyli@ustc.edu.cn}\par}
\end{center}
\vspace{1.25em}

\ifdefined\StandaloneSI
  \newcommand{\mainref}[1]{\ref*{main-#1}}
\else
  \newcommand{\mainref}[1]{\ref{#1}}
\fi
\newcommand{\maincite}[1]{\cite{#1}}
\ifdefined\StandaloneSI
  \newcommand{\maineqref}[1]{\eqref{main-#1}}
  \newcommand{\maincref}[1]{\cref{main-#1}}
  \newcommand{\Maincref}[1]{\Cref{main-#1}}
\else
  \newcommand{\maineqref}[1]{\eqref{#1}}
  \newcommand{\maincref}[1]{\cref{#1}}
  \newcommand{\Maincref}[1]{\Cref{#1}}
\fi

\makeatletter
\def\@sectioncntformat#1{\csname the#1\endcsname.\quad}
\def\@hangfrom@section#1#2#3{\@hangfrom{#1#2}\MakeTextUppercase{#3}}
\makeatother
\setcounter{section}{0}
\renewcommand{\thesection}{\arabic{section}}
\renewcommand{\thesubsection}{\thesection.\arabic{subsection}}
\numberwithin{equation}{section}
\renewcommand{\theequation}{S\arabic{section}.\arabic{equation}}
\setcounter{table}{0}
\setcounter{figure}{0}
\renewcommand{\thetable}{\arabic{table}}
\renewcommand{\thefigure}{\arabic{figure}}
\renewcommand{\tablename}{Supplementary Table}
\renewcommand{\figurename}{Supplementary Fig.}
\setcounter{theorem}{0}
\setcounter{notation}{0}
\setcounter{fact}{0}
\renewcommand{\thetheorem}{S\arabic{theorem}}
\renewcommand{\thenotation}{S\arabic{notation}}
\renewcommand{\thefact}{S\arabic{fact}}
\providecommand{\theHtheorem}{}
\providecommand{\theHnotation}{}
\providecommand{\theHfact}{}
\renewcommand{\theHtheorem}{supplementary.theorem.\arabic{theorem}}
\renewcommand{\theHnotation}{supplementary.notation.\arabic{notation}}
\renewcommand{\theHfact}{supplementary.fact.\arabic{fact}}

\providecommand{\theHsection}{}
\providecommand{\theHsubsection}{}
\providecommand{\theHsubsubsection}{}
\providecommand{\theHfigure}{}
\providecommand{\theHtable}{}
\providecommand{\theHequation}{}
\renewcommand{\theHsection}{supplementary.\arabic{section}}
\renewcommand{\theHsubsection}{supplementary.\arabic{section}.\arabic{subsection}}
\renewcommand{\theHsubsubsection}{supplementary.\arabic{section}.\arabic{subsection}.\arabic{subsubsection}}
\renewcommand{\theHfigure}{supplementary.\arabic{figure}}
\renewcommand{\theHtable}{supplementary.\arabic{table}}
\renewcommand{\theHequation}{supplementary.\arabic{section}.\arabic{equation}}

\crefname{figure}{Supplementary Fig.}{Supplementary Figs.}
\Crefname{figure}{Supplementary Fig.}{Supplementary Figs.}
\crefname{table}{Supplementary Table}{Supplementary Tables}
\Crefname{table}{Supplementary Table}{Supplementary Tables}

\noindent
The proofs of the main classification and minimum-support results are
included as appendices to the article. This Supplementary Information
contains detailed operator and literature comparisons, the QNP-F boundary
reduction, the positive-spin conditioned-single extension and singlet
obstruction, and additional numerical checks and diagnostics.

\input{supplementary/01_operator_comparisons.tex}
\input{supplementary/02_qnp_f_boundary.tex}
\input{supplementary/03_conditioned_extension.tex}
\input{supplementary/06_numerics.tex}

\FloatBarrier
\SIBibliography
\clearpage
\endgroup

%% file: supplementary/01_operator_comparisons.tex
\section{Operator conventions and relation to previous work}
\label{si:operator-comparisons}

\subsection{Conventions of Magoulas--Evangelista}
\label{app:generator-conventions}

We give the coefficient-level correspondence between the generators used
here and those of Magoulas and Evangelista~\cite{Magoulas2025SpinAdapted}.
For this section, denote their elementary
anti-Hermitian spin-orbital excitations by
\begin{align*}
\mathcal A_{P_\sigma}^{Q_\tau}
&=a_{Q\tau}^\dagger a_{P\sigma}
  -a_{P\sigma}^\dagger a_{Q\tau},\\
\mathcal A_{P_\sigma Q_\tau}^{R_\mu S_\nu}
&=a_{R\mu}^\dagger a_{S\nu}^\dagger
  a_{Q\tau}a_{P\sigma}-\mathrm{h.c.}
\end{align*}
Their normalized singlet single, Eq.~(3) of that reference, is
\[
\mathcal A_P^Q
=\frac{1}{\sqrt2}\left(
\mathcal A_{P_\alpha}^{Q_\alpha}
+\mathcal A_{P_\beta}^{Q_\beta}\right),
\]
and their Eq.~(4) is the perfect-pairing double
\[
\mathcal A_{PP}^{QQ}
=\mathcal A_{P_\alpha P_\beta}^{Q_\alpha Q_\beta}.
\]
For four distinct spatial orbitals, their intermediate-triplet scalar,
Eq.~(7), is
\begin{align}
\label{eq:ME-triplet}
{}^{[1]}\mathcal A_{PQ}^{RS}
=\frac{1}{\sqrt3}\Bigl[&
\mathcal A_{P_\alpha Q_\alpha}^{R_\alpha S_\alpha}
+\mathcal A_{P_\beta Q_\beta}^{R_\beta S_\beta}\notag\\
&+\frac12\Bigl(
\mathcal A_{P_\alpha Q_\beta}^{R_\alpha S_\beta}
+\mathcal A_{P_\alpha Q_\beta}^{R_\beta S_\alpha}\notag\\
&\hspace{3.2em}+
\mathcal A_{P_\beta Q_\alpha}^{R_\alpha S_\beta}
+\mathcal A_{P_\beta Q_\alpha}^{R_\beta S_\alpha}
\Bigr)\Bigr].
\end{align}
The superscript \([1]\) is the common intermediate spin of the annihilated
and created pairs.  The complete operator is a total-spin scalar; the label
does not mean that the generator transforms as a total-spin triplet.

The same reference separates the remaining double excitations into the
three-orbital pair-breaking/formation operator \(\mathcal A_{PP}^{QR}\)
[their Eq.~(5)] and the four-distinct-orbital intermediate-singlet operator
\({}^{[0]}\mathcal A_{PQ}^{RS}\) [their Eq.~(6)].  Equations~(3)--(7) together
form their saGSD pool.  The smaller saGSpD pool contains only the normalized
singlet singles and perfect-pairing doubles of Eqs.~(3)--(4).

\Needspace{12\baselineskip}
\subsection{Exact identification of the generators used here}
\label{app:generator-provenance}

Fix \(p<q\) and \(r<s\).  Subscripts label the annihilated pair and
superscripts the created pair.  Direct normal-ordering gives the following
full-Fock-space identities, with four distinct indices in the third:
\begin{equation}
\label{eq:ME-dictionary}
\boxed{
\begin{aligned}
L_{pq}&=\sqrt2\,\mathcal A_q^p,\\
X_{pq}&=\mathcal A_{qq}^{pp},\\
Y_{\{p,q\},\{r,s\}}&=\sqrt3\,{}^{[1]}\mathcal A_{rs}^{pq}.
\end{aligned}}
\end{equation}
Reversing a source and
destination pair changes the corresponding overall sign.  The coefficient
identities hold on the full Fock space.

\begin{table}[!htbp]
\centering
\caption{Provenance of the primitive coordinate directions.}
\label{tab:generator-provenance}
\renewcommand{\arraystretch}{1.18}
\begin{tabular}{@{}llll@{}}
\toprule
\tblcell{0.13\textwidth}{Present symbol}
& \tblcell{0.19\textwidth}{Ref.~\cite{Magoulas2025SpinAdapted}}
& \tblcell{0.22\textwidth}{Exact relation}
& \tblcell{0.29\textwidth}{Scope in the present work} \\
\midrule
\tblcell{0.13\textwidth}{\(L_{pq}\)}
& \tblcell{0.19\textwidth}{Eq.~(3), normalized singlet single}
& \tblcell{0.22\textwidth}{\(\sqrt2\,\mathcal A_q^p\)}
& \tblcell{0.29\textwidth}{Connected singles form the orbital-rotation backbone.} \\
\tblcell{0.13\textwidth}{\(X_{pq}\)}
& \tblcell{0.19\textwidth}{Eq.~(4), perfect-pairing double}
& \tblcell{0.22\textwidth}{\(\mathcal A_{qq}^{pp}\)}
& \tblcell{0.29\textwidth}{One coordinate seed completes every interior sector.} \\
\tblcell{0.13\textwidth}{\(Y_{\{p,q\},\{r,s\}}\), disjoint pairs}
& \tblcell{0.19\textwidth}{Eq.~(7), intermediate-triplet \([1]\) double}
& \tblcell{0.22\textwidth}{\(\sqrt3\,{}^{[1]}\mathcal A_{rs}^{pq}\)}
& \tblcell{0.29\textwidth}{One coordinate seed completes every nontrivial boundary sector.} \\
\tblcell{0.13\textwidth}{\(Y_{\{c,p\},\{c,q\}}\), shared orbital}
& \tblcell{0.19\textwidth}{Not the four-distinct case of Eq.~(7)}
& \tblcell{0.22\textwidth}{Repeated-index specialization of text Eq.~\mainref{eq:dint1-Y}}
& \tblcell{0.29\textwidth}{Reduces to an occupation-conditioned single on a boundary.} \\
\tblcell{0.13\textwidth}{\(C_{c;pq}\)}
& \tblcell{0.19\textwidth}{Not a primitive operator in Eqs.~(3)--(7)}
& \tblcell{0.22\textwidth}{\((E_{cc}-1)L_{pq}\)}
& \tblcell{0.29\textwidth}{Three-orbital boundary-optimal representative.} \\
\bottomrule
\end{tabular}
\end{table}
\FloatBarrier

Shared-orbital coordinates \(Y_{\{c,p\},\{c,q\}}\) are defined separately
to include the repeated-index contractions absent from the four-distinct
case of \cref{eq:ME-triplet}.  After boundary projection, Pauli blocking produces the conditioned
rotation used in Lemma~\mainref{lem:first-conditioned} and
hole-overlap equation.  It is also different from the
three-orbital operator \(\mathcal A_{PP}^{QR}\) in Eq.~(5) of
Ref.~\cite{Magoulas2025SpinAdapted}, which breaks or forms a doubly occupied
pair.  The factors \(\sqrt2\) and \(\sqrt3\) do not affect a Lie
closure, but they do affect ansatz angles, gradient magnitudes,
regularization, and any finite-depth numerical comparison.

For the pool-level comparisons in this section, define
the real operator spans
\begin{align*}
\Xpool&=\Span_{\R}\{X_{pq}:p<q\},\\
\Ypool&=\Span_{\R}\{Y_{e,f}:e\ne f\},\\
\Ypooldisj&=\Span_{\R}\{Y_{e,f}:e\cap f=\varnothing\},
\end{align*}
where \(e,f\) are unordered pairs of distinct spatial orbitals.
The shorthand \(\Lpool\) continues to denote the complete-graph orbital
span defined in the main text.  These spans are distinct from individual
coordinate seeds such as \(X_{pq}\) and \(Y_{e,f}\).

At the pool level, the all-pair span
\(\Lpool+\Xpool\), after the first rescaling in
\cref{eq:ME-dictionary}, is the saGSpD operator span.  A nearest-neighbor or
brickwork QNP-Q circuit, like a tUPS product, uses a sparse ordered selection
of \(L\) and \(X\) generators from that span.  The five-parameter QNP-F block
has additional independent local axes, as specified in
Section~\ref{app:F-boundary}.  The disjoint subpool \(\Ypooldisj\) belongs
to the \([1]\) sector of saGSD; a local intermediate-triplet placement is only a sparse subset
of that sector.  Full saGSD additionally contains the three-orbital
pair-breaking class and the generalized \([0]\) intermediate-singlet class,
neither of which is a primitive seed in the classification proved here.

\subsection{Relation to previous universality results}
\label{app:universality-comparison}

Exact fermionic wave-function parameterizations using products of low-body
unitary operators provide an earlier foundation for disentangled UCC
\cite{Evangelista2019UCC}.  The present question additionally fixes the
spin-scalar generator family and asks for its restricted DLA and the
minimum orbital support of an additional completing coordinate.

Burton et al.~\cite{Burton2023Exact} give a pool-level finite-product
universality argument for generalized spin-adapted singles and paired doubles,
with explicit occupation-edge exceptions when no doubly occupied or no doubly
unoccupied orbital is available.  Two distinctions matter when relating that
full-Fock-space construction to a fixed-spin restricted DLA.  First, one
displayed synthesis uses individual spin-resolved singles, so restriction by a
sector projector \(\Pi\) obeys
\[
\Pi[A,B]\Pi=[\Pi A\Pi,\Pi B\Pi]
 +\Pi A(1-\Pi)B\Pi-\Pi B(1-\Pi)A\Pi.
\]
The cross-sector terms vanish for the spin scalars used here, but need not
vanish for spin-resolved intermediates.  Second, a displayed commutator
produces an occupation-dressed double,
\[
[\hat\kappa_{p\bar q}^{r\bar s},
  \hat\kappa_{s\bar s}^{q\bar q}]
=\hat\kappa_{ps}^{rq}(\hat n_{\bar s}-\hat n_{\bar q}),
\]
rather than the bare same-spin double.  Lie closure alone does not remove the
occupation factor.  Appendix~\mainref{app:interior-full-proof} therefore treats
these conditional factors directly through strict localization, cyclicity,
and subspace expansion.

Burton's tUPS construction derives arbitrary-pair singles and paired doubles
from nearest-neighbor generators \cite[Appendix C]{Burton2024TUPS}. In an
interior sector, \(b,h\ge1\), sector-internal completion establishes the
operator equality
\[
\Lie\bigl(\rho_{M,N,S}(\so(M)),\{X_{rs}\}\bigr)
=\so(d_{M,N,S})
\]
on the complete fixed-\((N,S,M_S=S)\) multiplicity space.  This restricted
full-algebra result is the step that establishes sectorwise completion.
Orbital covariance and the signed-permutation orbit then reduce the additional
paired pool to one coordinate \(X_{pq}\).

Magoulas and Evangelista classify the spin-adapted operators above and study
exact product representations of their exponentials through finite local
dynamical Lie algebras \cite{Magoulas2025SpinAdapted}.  Their local 5-, 28-,
and 84-dimensional algebras organize exact Wei--Norman factorizations of
individual operator families.  Their projected-action arguments also treat
larger pools containing the full intermediate-singlet \([0]\) family and
related reduced subpools.  The coordinatewise result here instead fixes a
connected \(L\) backbone, determines the many-orbital fixed-spin restricted
DLA, and reduces the additional paired pool to one \(X_{pq}\).

Stergiou and Sawaya~\cite{Stergiou2026} study qubit exchange generators
without nonlocal Jordan--Wigner strings.  Their spectator-conditioning
mechanism is shared with our boundary edge-isolation argument and the
hard-core-pair construction in Lemma~\mainref{lem:hardcore}.  For the complete
molecular spin sector, the interior proof additionally localizes the action on
\(K_R\) and extends that supported block by highest-weight cyclicity and
one-sided infection.  The resulting theorems classify individual
perfect-pairing coordinates in the interior and the full admissible
at-most-two-body coefficient space on the boundary, with orbital support
measured on the full fermionic Fock space.

Their additional fuzzy-sphere constraint concerns orbital azimuthal angular momentum, rather than molecular electronic total spin; it therefore specifies a different target space from the complete fixed-\((N,S,M_S)\) space studied here.

Oszmaniec and Zimbor\'{a}s classify extensions of passive fermionic linear
optics, with a \(U(M)\) backbone acting on complex exterior-power spaces
\cite{Oszmaniec2017}.  The backbone here is the real spin-independent
\(SO(M)\) action on a fixed-spin sector, and the target is \(SO(d)\).
Their extension criteria therefore do not supply the restricted real
classification used above.  Symmetry-preserving state preparation and
particle-conserving primitives
\cite{Gard2020,Arrazola2022}, and compact complete-pool constructions
\cite{Tang2021,Shkolnikov2023,Haidar2025,Viswanathan2026}
provide complementary approaches to controllable circuit design.
Here we fix the real spin-scalar orbital backbone, classify its complete
at-most-two-body boundary seed space and its coordinatewise interior
completion, and measure support on the full Fock space.

\begin{remark}[Boundary failure combinations and the interior quantifier]
The complete boundary criterion includes nonzero linear combinations that
fail.  On an electron \(k\)-particle
boundary, orienting overlapping doubles consistently gives
\[
\sum_{p\notin\{a,b\}}D_{pa,pb}
=\sum_{p\notin\{a,b\}}n_p^{\mathrm{sl}}L_{ab}
=(k-1)L_{ab}.
\]
This nonzero linear combination already lies in the singles algebra.  The
conditioned-single family has the analogous relation
\[
\sum_{c\notin\{p,q\}}C_{c;pq}
=\sum_{c\notin\{p,q\}}(n_c^{\mathrm{sl}}-1)L_{pq}
=(k-M+1)L_{pq}.
\]
For \(2\le k\le M-2\), it is again a nonzero span element that does not
enlarge the orbital algebra.  In pair-coefficient language these combinations
lie in \(\mathcal J_1\), exactly as
\maincref{eq:boundary-failure-subspace} requires.  By contrast, the general
interior result remains coordinatewise and does not assert that every
nonzero element of a generator span completes the DLA.
\end{remark}

\FloatBarrier

%% file: supplementary/02_qnp_f_boundary.tex
\section{Boundary reduction of the local F fabric}
\label{app:F-boundary}

This section determines the action of the local five-parameter \(F\) gate of
Anselmetti et al.\ \cite{Anselmetti2021} on strict maximal-spin boundaries.
For two spatial orbitals \(p,q\), set
\begin{align*}
\ket{P_p}&=P_p^\dagger\ket0,
&\ket{P_q}&=P_q^\dagger\ket0,\\
\ket{S_{pq}}&=\frac{a_{p\alpha}^\dagger a_{q\beta}^\dagger
-a_{p\beta}^\dagger a_{q\alpha}^\dagger}{\sqrt2}\ket0.
\end{align*}
With \(\ket{\Omega_{pq}}=P_p^\dagger P_q^\dagger\ket0\), define explicitly
the one-particle states
\(\ket{p\sigma}=a_{p\sigma}^\dagger\ket0\) and the one-hole states
\(\ket{\overline{p\sigma}}=a_{p\sigma}\ket{\Omega_{pq}}\).  The unbarred and
barred blocks are respectively the local particle and hole sectors used by the
five-axis tangent convention.  The local Fock space splits by particle number
and spin, and the five nontrivial tangent directions act as
\[
\so(2)_{\mathrm{1p}}\oplus
\so(2)_{\mathrm{1h}}\oplus
\so(3)_{\mathrm{2e},S=0},
\]
while the empty, filled, and local two-electron triplet sectors are fixed.  Up
to independent orientation and angle conventions, the tangent space is
\begin{align}
\label{eq:F-five-axes}
\mathcal F_{pq}=\Span_{\R}\{K_{\mathrm{1p}},K_{\mathrm{1h}},
K_{\mathrm{PX}},K_{\mathrm{PBU}},K_{\mathrm{PBL}}\},
\end{align}
where
\begin{align*}
K_{\mathrm{1p}}
&=\sum_{\sigma=\alpha,\beta}
\bigl(\ket{p\sigma}\!\bra{q\sigma}
-\ket{q\sigma}\!\bra{p\sigma}\bigr),\\
K_{\mathrm{1h}}
&=\sum_{\sigma=\alpha,\beta}
\bigl(\ket{\overline{p\sigma}}\!\bra{\overline{q\sigma}}
-\ket{\overline{q\sigma}}\!\bra{\overline{p\sigma}}\bigr),\\
K_{\mathrm{PX}}
&=\ket{P_p}\!\bra{P_q}-\ket{P_q}\!\bra{P_p},\\
K_{\mathrm{PBU}}
&=\ket{P_p}\!\bra{S_{pq}}-\ket{S_{pq}}\!\bra{P_p},\\
K_{\mathrm{PBL}}
&=\ket{P_q}\!\bra{S_{pq}}-\ket{S_{pq}}\!\bra{P_q}.
\end{align*}
Thus \(K_{\mathrm{1p}}\) and \(K_{\mathrm{1h}}\) are the two local rotation
axes, while the last three axes generate the singlet \(\so(3)\) block.
Each of these five axes commutes with the complete local spin action, not only
with \(S^2\) and \(S_z\).  The one-particle and one-hole rotations act with
the same orbital matrix on the two spin components of their
spin-\(\tfrac12\) doublets.  The other three axes act only within the
two-electron \(S=0\) multiplicity block, while every component of the local
triplet is fixed.  Hence every finite factor, and therefore the published
five-parameter block, commutes with \(S_x,S_y,S_z\).

This identification is a block-by-block transcription of Fig.~7 and
Appendices D--E of Ref.~\cite{Anselmetti2021}.  In the ordered singlet basis
\((\ket{P_p},\ket{P_q},\ket{S_{pq}})\), the three published one-parameter
factors have tangent matrices
\begin{align*}
K_{\mathrm{PX}}&=
\begin{pmatrix}0&1&0\\-1&0&0\\0&0&0\end{pmatrix},&
K_{\mathrm{PBU}}&=
\begin{pmatrix}0&0&1\\0&0&0\\-1&0&0\end{pmatrix},&
K_{\mathrm{PBL}}&=
\begin{pmatrix}0&0&0\\0&0&1\\0&-1&0\end{pmatrix}.
\end{align*}
Choosing ``upper'' to mean \(P_p\), the exact name-to-axis map is
\begin{equation}
\begin{array}{c|ccccc}
\text{published factor}&\mathrm{QNP}_{1p}&\mathrm{QNP}_{1h}
&\mathrm{QNP}_{\mathrm{PX}}&\mathrm{QNP}_{\mathrm{PBU}}
&\mathrm{QNP}_{\mathrm{PBL}}\\ \hline
\text{axis}&K_{\mathrm{1p}}&K_{\mathrm{1h}}&K_{\mathrm{PX}}
&K_{\mathrm{PBU}}&K_{\mathrm{PBL}}
\end{array}
\label{eq:F-published-axis-map}
\end{equation}
Thus each named factor is
\(\exp[\vartheta_a(\theta_a)K_a]\), where \(\vartheta_a\) is the actual
plane-rotation angle read from the sine/cosine block in the published matrix.
The occasional half-angle used in the elementary-circuit parameter is only a
reparametrization; orientation changes replace \(\vartheta_a\) by
\(-\vartheta_a\).  This convention keeps the finite-angle statement
independent of circuit-symbol normalization while fixing the tangent and the
rotated CSF plane exactly.

More explicitly, the published finite gate consists of the two indicated
\(SO(2)\) rotations and an element \(R_0\in SO(3)\) on the singlet block.
Choosing any Euler factorization of \(R_0\) in the three axes of
\cref{eq:F-five-axes}, and absorbing orientation signs into the angles, one
may write it as
\begin{equation}
\label{eq:F-five-factor}
F_{pq}(\boldsymbol\vartheta)=
\e^{\vartheta_{\mathrm{1p}}K_{\mathrm{1p}}}
\e^{\vartheta_{\mathrm{1h}}K_{\mathrm{1h}}}
\e^{\vartheta_{\mathrm{PX}}K_{\mathrm{PX}}}
\e^{\vartheta_{\mathrm{PBU}}K_{\mathrm{PBU}}}
\e^{\vartheta_{\mathrm{PBL}}K_{\mathrm{PBL}}},
\end{equation}
where changing the order of the last three factors only changes the Euler
convention inside the singlet block.  Ordering each one-particle and one-hole
doublet as orbital multiplicity \(\otimes\) spin, the exact finite matrix
form of \cref{eq:F-five-factor} on the 16-dimensional local Fock space is
\begin{equation}
\label{eq:F-finite-block}
\begin{split}
F_{pq}(\boldsymbol\vartheta)={}&
\bigl(R_{\mathrm{1p}}(\vartheta_{\mathrm{1p}})\otimes I_2\bigr)
\oplus\bigl(R_{\mathrm{1h}}(\vartheta_{\mathrm{1h}})\otimes I_2\bigr)\\
&\oplus R_{S=0}(\vartheta_{\mathrm{PX}},
\vartheta_{\mathrm{PBU}},\vartheta_{\mathrm{PBL}})
\oplus I_5.
\end{split}
\end{equation}
Here each of the first two plane rotations acts identically on both spin
components; the fixed subspace comprises the empty and filled states and the
three two-electron triplet states.  The matrix \(R_{S=0}\) is the ordered
product of the three indicated plane rotations on
\((\ket{P_p},\ket{P_q},\ket{S_{pq}})\).  Thus
\cref{eq:F-finite-block} is an exact factorization of the finite published
gate, not merely a gate chosen to have the same tangent space.

\begin{proposition}[Boundary restriction of the local \(F\) family]
\label{prop:F-boundary}
On an electron maximal-spin boundary,
\[
\Pi_{M,N,S}\mathcal F_{pq}\Pi_{M,N,S}
=\Span_{\R}\{L_{pq}^{(M,N,S)}\},
\]
and the particle--hole-dual equality holds on a hole maximal-spin boundary.
Consequently, even the all-pair \(F\) tangent pool satisfies
\[
\Lie\!\left(\mathcal F_{pq}^{(M,N,S)}:p<q\right)
=\rho_k\!\left(\so(M)\right).
\]
For \(2\le k\le M-2\), this is a proper subalgebra of
\(\so\!\left(\binom Mk\right)\).
\end{proposition}

\begin{proof}
On the electron maximal-spin boundary, every electron has spin \(\alpha\).
The local one-hole sector requires a \(\beta\) electron and is absent, while
the three local two-electron singlet CSFs are orthogonal to the allowed
two-\(\alpha\) pattern.  Consequently
\begin{align*}
\Pi_{b=0}K_{\mathrm{1h}}\Pi_{b=0}
&=\Pi_{b=0}K_{\mathrm{PX}}\Pi_{b=0}=0,\\
\Pi_{b=0}K_{\mathrm{PBU}}\Pi_{b=0}
&=\Pi_{b=0}K_{\mathrm{PBL}}\Pi_{b=0}=0,
\end{align*}
whereas
\[
\Pi_{b=0}K_{\mathrm{1p}}\Pi_{b=0}
=\Pi_{b=0}L_{pq}\Pi_{b=0}.
\]
Tensoring the local statement with every spectator occupation pattern gives
the exterior-power orbital rotation on the complete boundary sector.

On the hole maximal-spin boundary, every spatial orbital contains an
\(\alpha\) electron and the degrees of freedom are the \(\beta\) holes
\[
h_p^\dagger=a_{p\beta},\qquad
h_p=a_{p\beta}^\dagger,\qquad
n_p^{\mathrm h}=h_p^\dagger h_p.
\]
The local one-particle sector is now absent.  If both \(\beta\) electrons are
missing, the two surviving \(\alpha\) electrons form the \(m=+1\) triplet, not
a local singlet; hence all three singlet-block axes again vanish.  The
one-hole axis is the only surviving direction and, up to the harmless phase
choice in the barred basis,
\[
\Pi_{h=0}K_{\mathrm{1h}}\Pi_{h=0}
=\pm\bigl(h_p^\dagger h_q-h_q^\dagger h_p\bigr)
=\pm L_{pq}^{\mathrm h}.
\]
Thus the all-pair \(F\) tangent pool restricts on either boundary to the same
represented orbital algebra \(\rho_k(\so(M))\).  The apparent availability of
five local parameters does not provide five independent boundary directions:
four act entirely through local blocks missing from the boundary sector.
Finally, for \(2\le k\le M-2\), one has
\(d=\binom Mk>M\), and hence
\(\dim\rho_k(\so(M))=\binom M2<\binom d2=\dim\so(d)\).
\end{proof}

\begin{corollary}[Finite \(F\)-gate boundary restriction]
\label{cor:F-finite-boundary}
Write the finite local gate as
\[
F_{pq}(\vartheta_{\mathrm{1p}},\vartheta_{\mathrm{1h}},R_0),
\qquad R_0\in SO(3),
\]
in the block convention of \cref{eq:F-five-axes}.  On an electron
maximal-spin boundary,
\[
F_{pq}(\vartheta_{\mathrm{1p}},\vartheta_{\mathrm{1h}},R_0)
\big|_{\Hsp}
=\exp\!\left(\vartheta_{\mathrm{1p}}
L_{pq}^{(M,N,S)}\right),
\]
independently of \(\vartheta_{\mathrm{1h}}\) and \(R_0\).  On a hole
maximal-spin boundary,
\[
F_{pq}(\vartheta_{\mathrm{1p}},\vartheta_{\mathrm{1h}},R_0)
\big|_{\Hsp}
=\exp\!\left(\vartheta_{\mathrm{1h}}L_{pq}^{\mathrm h}\right),
\]
independently of \(\vartheta_{\mathrm{1p}}\) and \(R_0\), up to the fixed
orientation convention for the barred basis.  Consequently,
\[
\left\langle F_{pq}\big|_{\Hsp}:p<q\right\rangle
=R_k(SO(M)).
\]
\end{corollary}

\begin{proof}
On an electron maximal-spin boundary, the allowed local space of the two
orbitals is
\[
\mathcal B_{\mathrm e}=\Span\!\left\{
\ket0,\ket{p\alpha},\ket{q\alpha},
a_{p\alpha}^\dagger a_{q\alpha}^\dagger\ket0
\right\}.
\]
The exact direct sum in \cref{eq:F-finite-block} leaves this space invariant
and restricts to
\[
F_{pq}\big|_{\mathcal B_{\mathrm e}}
=1\oplus R_{\mathrm{1p}}(\vartheta_{\mathrm{1p}})\oplus1
=\exp\!\left(\vartheta_{\mathrm{1p}}L_{pq}^{\mathrm{boundary}}\right).
\]
Thus the one-hole and three singlet-block factors act as the identity at
finite angle, proving the first equality without a tangent-space argument.

For the hole boundary, let
\(\ket{\Omega_{pq}}=P_p^\dagger P_q^\dagger\ket0\).  Its allowed local space
is
\[
\mathcal B_{\mathrm h}=\Span\!\left\{
\ket{\Omega_{pq}},
a_{p\beta}\ket{\Omega_{pq}},
a_{q\beta}\ket{\Omega_{pq}},
a_{p\beta}a_{q\beta}\ket{\Omega_{pq}}
\right\},
\]
and the corresponding restriction is
\[
F_{pq}\big|_{\mathcal B_{\mathrm h}}
=1\oplus R_{\mathrm{1h}}(\vartheta_{\mathrm{1h}})\oplus1,
\]
up to the fixed orientation of the barred basis.  This is the
particle--hole-dual equality.  Tensoring either local restriction with
arbitrary spectator occupations and using all orbital pairs gives
\(R_k(SO(M))\).  In particular, no finite choice of the four inactive
parameters can couple a boundary state to a local block absent from that
boundary.
\end{proof}

\begin{remark}[Restricted-matrix validation]
Independent restricted-DLA checks reported in the Supplemental Material validate the five-axis
implementation in \cref{eq:F-five-axes}, but they are not used in either
proof.
\end{remark}

\begin{remark}[Relation to the reported \(F\)-fabric numerics]
\label{rem:F-prior-numerics}
Figure~7 of Ref.~\cite{Anselmetti2021} labels the many-orbital \(F\) fabric as
hypothesized to be universal, and its Appendix~G reports no \(F\)-specific
edge non-universality in the tested cases.  \Cref{cor:F-finite-boundary} shows
that, for the published five-parameter block, the conjectured universality does
not extend to strict maximal-spin boundaries: its restricted action remains in
\(R_k(SO(M))\).  This does not affect the reported performance away from
these boundaries.
\end{remark}

%% file: supplementary/03_conditioned_extension.tex
\section{Positive-spin universality and the singlet obstruction for
conditioned singles}
\label{app:C-positive-spin}

This section records a strengthening of the conditioned-single
result that does not alter the minimum-support classification in the main
text.  The three-orbital generator \(C_{c;pq}\) remains support optimal on a
nontrivial maximal-spin boundary, whereas the two-orbital \(X_{pq}\) remains
support optimal in the interior.  Nevertheless, for \(M\ge3\), every
elementary \(C_{c;pq}\) with pairwise distinct \(c,p,q\in[M]\) is also
sufficient in every physically allowed positive-spin interior sector
\((b,h\ge1,\ S>0)\).  The
proof first constructs a perfect-pairing rotation conditioned on single
occupation, then removes all occupation conditions after restriction to
\(S>0\).

For the full Fock-space argument, fix \(M\ge3\).  All orbital indices below
belong to \([M]\); each elementary \(C_{c;pq}\) has pairwise distinct
\(c,p,q\).  Set
\begin{equation}
D_a=E_{aa}-1,
\qquad
q_a=1-D_a^2=E_{aa}(2-E_{aa}).
\label{eq:C-app-single-projector}
\end{equation}
Thus \(q_a\) is the orthogonal projector onto single occupation of spatial
orbital \(a\), and \(C_{a;rs}=D_aL_{rs}\).  All operators used below commute
with \(\widehat N\) and every \(\widehat S_\mu\).  Consequently their full-space
identities may be restricted at the end to any fixed-\((N,S)\) multiplicity
space.  For each fixed pairwise distinct \(c,p,q\), write
\begin{equation}
\widehat{\mathfrak g}_C(c;p,q)
=\Lie_{\R}\bigl(\{L_{ab}:a<b\},C_{c;pq}\bigr)
\label{eq:C-app-full-algebra}
\end{equation}
for the full-Fock-space algebra and
\begin{equation}
\mathfrak g_C^{(M,N,S)}(c;p,q)
=\Lie_{\R}\bigl(\rho_{M,N,S}(\so(M)),
C_{c;pq}^{(M,N,S)}\bigr)
\label{eq:C-app-sector-algebra}
\end{equation}
for its restricted counterpart.

\begin{lemma}[Coordinate transport in all orbital dimensions]
\label{lem:C-app-transport}
Let \(M\ge3\).  From any elementary \(C_{c;pq}\) with pairwise distinct
\(c,p,q\in[M]\), \(\widehat{\mathfrak g}_C(c;p,q)\) contains every
\(C_{a;rs}\) with pairwise distinct \(a,r,s\in[M]\).
\end{lemma}

\begin{proof}
The algebra in \cref{eq:C-app-full-algebra} is invariant under conjugation by
the connected orbital group \(SO(M)\).  Choose a permutation \(\pi\) sending
\((c,p,q)\) to \((a,r,s)\).  If its permutation matrix has determinant
\(+1\), it already gives the required transport.  If its determinant is
\(-1\), compose it with a sign flip of the target condition orbital \(a\).
This changes the determinant to \(+1\), while leaving both \(E_{aa}\) and
\(L_{rs}\) unchanged because \(a,r,s\) are distinct.  The resulting signed
permutation lies in \(SO(M)\) and conjugates \(C_{c;pq}\) exactly to
\(C_{a;rs}\).  This determinant correction also works for \(M=3\); no fourth
spectator orbital is required.
\end{proof}

\begin{lemma}[Single-occupation-conditioned perfect pairing]
\label{lem:C-app-single-X}
Let \(M\ge3\), and let both \((c_0,p_0,q_0)\) and \((c,p,q)\) be
pairwise-distinct triples in \([M]\).  Then
\begin{equation}
q_cX_{pq}\in\widehat{\mathfrak g}_C(c_0;p_0,q_0)
\label{eq:C-app-qX-membership}
\end{equation}
\end{lemma}

\begin{proof}
By \cref{lem:C-app-transport}, all conditioned coordinates appearing below
belong to the algebra generated by the initial seed.  Define
\begin{align*}
A&=C_{p;qc},\qquad B=C_{c;pq},\\
G&=C_{q;pc},\\
\mathcal R_{c;pq}
&=[A,[B,G]]\\
&\quad+[L_{pq},[A,L_{pc}]]+[G,[B,A]],\\
\mathcal T_{c;pq}
&=[L_{pq},[A,[B,A]]].
\end{align*}
A direct three-orbital CAR calculation gives the full-Fock-space identity
\begin{equation}
\boxed{
\mathcal T_{c;pq}-2\mathcal R_{c;pq}=12q_cX_{pq}.
}
\label{eq:C-app-core-CAR}
\end{equation}
For completeness, substitute \(C_{a;rs}=D_aL_{rs}\), apply the Leibniz rule,
and use
\begin{align*}
[E_{ij},E_{kl}]&=\delta_{jk}E_{il}-\delta_{il}E_{kj},\\
E_{pq}^2&=2P_p^\dagger P_q,
&E_{cc}^2&=E_{cc}+2P_c^\dagger P_c.
\end{align*}
Normal ordering the four commutators on the left of
\cref{eq:C-app-core-CAR} yields
\begin{align*}
\mathcal T_{c;pq}-2\mathcal R_{c;pq}
&=6(2E_{cc}-E_{cc}^2)(E_{pq}^2-E_{qp}^2)\\
&=12\bigl(1-(E_{cc}-1)^2\bigr)X_{pq}
=12q_cX_{pq}.
\end{align*}
Every term on the left of \cref{eq:C-app-core-CAR} is a Lie word in orbital
rotations and elementary conditioned coordinates, proving
\cref{eq:C-app-qX-membership}.
\end{proof}

\begin{lemma}[Multiplication of a spectator condition]
\label{lem:C-app-conditioning}
Let \(M\ge3\) and fix pairwise distinct initial seed indices
\(c_0,p_0,q_0\in[M]\).  Let \(i,j\in[M]\) with \(i\ne j\),
\(A\subseteq[M]\setminus\{i,j\}\), and
\(\ell\in[M]\setminus(A\cup\{i,j\})\). Set
\(q_A=\prod_{a\in A}q_a\), with \(q_\varnothing=1\), and
\(Y=q_AX_{ij}\). If
\(Y\in\widehat{\mathfrak g}_C(c_0;p_0,q_0)\), then
\begin{equation}
q_\ell Y
=Y+\frac14[C_{\ell;ij},[C_{\ell;ij},Y]]
\in\widehat{\mathfrak g}_C(c_0;p_0,q_0).
\label{eq:C-app-condition-multiplied}
\end{equation}
The identity holds on the full Fock space.
\end{lemma}

\begin{proof}
The two-orbital CAR identity used in the proof of
\maincref{cor:sector-uniform-seed} is
\[
[L_{ij},[L_{ij},X_{ij}]]=-4X_{ij}.
\]
Every factor of \(q_A\) commutes with \(L_{ij}\), so
\([L_{ij},[L_{ij},Y]]=-4Y\). Moreover,
\([D_\ell,Y]=[D_\ell,L_{ij}]=0\), and therefore
\begin{align*}
[C_{\ell;ij},[C_{\ell;ij},Y]]
&=D_\ell^2[L_{ij},[L_{ij},Y]]\\
&=-4D_\ell^2Y.
\end{align*}
Using \(q_\ell=1-D_\ell^2\) proves
\eqref{eq:C-app-condition-multiplied}. By \cref{lem:C-app-transport},
\(C_{\ell;ij}\) belongs to the generated algebra. The right-hand side thus
uses only real linear combinations and commutators of available elements.
\end{proof}

\begin{corollary}[All nonempty single-occupation conditions]
\label{cor:C-app-multiconditioned}
Let \(M\ge3\) and fix pairwise distinct \(c,p,q\in[M]\).  For
\(i,j\in[M]\) with \(i\ne j\) and every nonempty
\(A\subseteq[M]\setminus\{i,j\}\), set \(q_A=\prod_{a\in A}q_a\).  Then
\begin{equation}
q_AX_{ij}\in\widehat{\mathfrak g}_C(c;p,q).
\label{eq:C-app-all-qA-X}
\end{equation}
\end{corollary}

\begin{proof}
The case \(\lvert A\rvert=1\) is \cref{lem:C-app-single-X}. If
\(q_AX_{ij}\) is available and \(\ell\notin A\cup\{i,j\}\),
\cref{lem:C-app-conditioning} gives
\(q_\ell q_AX_{ij}=q_{A\cup\{\ell\}}X_{ij}
\in\widehat{\mathfrak g}_C(c;p,q)\).
Induction proves the claim for every nonempty allowed \(A\).
\end{proof}

\begin{theorem}[Single \(C\) generates every perfect-pairing coordinate at
positive spin]
\label{thm:C-app-to-X}
Let \(M\ge3\), let \((N,S)\) be physically allowed with \(S>0\), and let
\(c,p,q\in[M]\) be pairwise distinct.  Then, for every
\(i,j\in[M]\) with \(i\ne j\),
\begin{equation}
\boxed{
X_{ij}^{(M,N,S)}\in\mathfrak g_C^{(M,N,S)}(c;p,q).
}
\label{eq:C-app-X-in-sector-algebra}
\end{equation}
No interior assumption on \(b\) and \(h\) is required.
\end{theorem}

\begin{proof}
The zero-seniority projector is
\begin{equation}
P_{\mathrm{zsen}}=\prod_{a=1}^M(1-q_a).
\label{eq:C-app-zsen-projector}
\end{equation}
Its range is spanned by determinants having only empty and doubly occupied
spatial orbitals, hence lies entirely in total spin zero.  Therefore
\begin{equation}
\Pi_{M,N,S}P_{\mathrm{zsen}}=0
\qquad(S>0).
\label{eq:C-app-zsen-vanishes}
\end{equation}
The perfect-pairing coordinate obeys the full-space identities
\begin{equation}
q_iX_{ij}=q_jX_{ij}=0,
\qquad
[q_k,X_{ij}]=0\quad(k\notin\{i,j\}).
\label{eq:C-app-q-X-identities}
\end{equation}
Consequently, after restriction to the target positive-spin sector,
\begin{align}
X_{ij}^{(M,N,S)}
&=\left[1-\prod_{k\ne i,j}(1-q_k)\right]
X_{ij}^{(M,N,S)}\notag\\
&=\sum_{\varnothing\ne A\subseteq[M]\setminus\{i,j\}}
(-1)^{\lvert A\rvert+1}(q_AX_{ij})^{(M,N,S)}.
\label{eq:C-app-inclusion-exclusion}
\end{align}
Every term in the last line belongs to the restricted generated algebra by
\cref{cor:C-app-multiconditioned}.  This proves
\cref{eq:C-app-X-in-sector-algebra}.
\end{proof}

\begin{corollary}[Positive-spin interior universality of one conditioned single]
\label{cor:C-app-open-shell-universal}
Let \(M\ge3\), let \((N,S)\) be physically allowed with \(b,h\ge1\)
and \(S>0\), and let \(c,p,q\in[M]\) be pairwise distinct.  Then
\begin{equation}
\boxed{
\mathfrak g_C^{(M,N,S)}(c;p,q)=\so(d_{M,N,S}).
}
\label{eq:C-app-open-shell-so}
\end{equation}
Together with the boundary part of \maincref{thm:main} and the one-particle and
one-hole cases where orbital rotations are already complete, this covers
every physically allowed positive-spin sector, including the trivial
\(d=1\) cases.
\end{corollary}

\begin{proof}
By \cref{thm:C-app-to-X}, the algebra generated by the orbital backbone and
one elementary \(C\) contains any chosen \(X_{ij}^{(M,N,S)}\).
Corollary~\mainref{thm:interior-single-seed} then gives
\begin{align*}
\so(d_{M,N,S})
&=\Lie_{\R}\bigl(\rho_{M,N,S}(\so(M)),
X_{ij}^{(M,N,S)}\bigr)\\
&\subseteq\mathfrak g_C^{(M,N,S)}(c;p,q)\\
&\subseteq\so(d_{M,N,S}),
\end{align*}
which proves equality.
\end{proof}

\begin{proposition}[Common fixed vector in every nontrivial singlet sector]
\label{prop:C-app-singlet-obstruction}
Let \(M\ge3\), let \(S=0\), \(N=2b\), and \(1\le b\le M-1\), and fix
pairwise distinct \(c,p,q\in[M]\).  Define
\begin{equation}
\Delta^\dagger=\sum_{a=1}^MP_a^\dagger,
\qquad
\ket{\Omega_b}=(\Delta^\dagger)^b\ket0.
\label{eq:C-app-omega}
\end{equation}
Then \(\ket{\Omega_b}\ne0\), it lies in \(\Hsp_{M,2b,0}\), and every orbital
rotation and every elementary conditioned single annihilates it.  Hence
\begin{equation}
\boxed{\begin{aligned}
\mathfrak g_C^{(M,2b,0)}(c;p,q)
&\subseteq\so(\Omega_b^\perp)\\
&\cong\so(d_{M,2b,0}-1)
\subsetneq\so(d_{M,2b,0}).
\end{aligned}}
\label{eq:C-app-singlet-stabilizer}
\end{equation}
\end{proposition}

\begin{proof}
Even pair-creation operators on different orbitals commute and square to
zero, so
\begin{align}
\ket{\Omega_b}
&=b!\sum_{\substack{I\subseteq[M]\\\lvert I\rvert=b}}
\left(\prod_{a\in I}P_a^\dagger\right)\ket0,\notag\\
\|\Omega_b\|^2&=(b!)^2\binom Mb>0.
\label{eq:C-app-omega-norm}
\end{align}
Each \(P_a^\dagger\) creates an on-site spin singlet, proving that
\(\ket{\Omega_b}\in\Hsp_{M,2b,0}\).  For an orbital rotation
\(Q\in SO(M)\), second quantization gives
\begin{align*}
\mathcal U(Q)\Delta^\dagger\mathcal U(Q)^{-1}
&=\sum_{a,j,k}Q_{ja}Q_{ka}
a_{j\alpha}^\dagger a_{k\beta}^\dagger\\
&=\sum_{j,k}(QQ^{\mathsf T})_{jk}
a_{j\alpha}^\dagger a_{k\beta}^\dagger
=\Delta^\dagger.
\end{align*}
Differentiating the invariance of \(\ket{\Omega_b}\) gives
\(L_{pq}\ket{\Omega_b}=0\) for every \(p\ne q\), and therefore
\begin{equation}
C_{c;pq}\ket{\Omega_b}
=D_cL_{pq}\ket{\Omega_b}=0.
\end{equation}
All Lie words in the generators annihilate the same vector.  Finally,
\begin{align*}
d_{M,2b,0}
&=\binom Mb^2-\binom M{b+1}\binom M{b-1}\\
&=\binom Mb^2\frac{M+1}{(b+1)(M-b+1)}>1.
\end{align*}
The defining representation of \(\so(d)\) has no nonzero common fixed
vector when \(d>1\), proving the strict obstruction.
\end{proof}

%% file: supplementary/06_numerics.tex
\section{Finite-dimensional and molecular cross-checks}
\label{app:dla-checks}

\subsection{Restricted DLA cross-checks}

As a check on signs, spin projection, boundary conditions, and exceptional
small sectors, we formed restricted generator matrices and computed their
commutator closures.  Except for the explicitly labelled \(d=70\)
sector-native calculation, the calculations began with full-Fock-space
fermionic operators, restricted first to fixed
\((N_\alpha,N_\beta)\), equivalently fixed \((N,M_S)\), and then to a fixed
\(S^2\) eigenspace.  No
Hamiltonian eigenvector or exact-root projector entered this construction.
Each rank below is reported as
\(\dim\mathfrak g/\dim\so(d)=\dim\mathfrak g/\binom d2\).

The one-seed checks use a relative rank tolerance of $10^{-10}$.
Pooled calculations vary this tolerance from $10^{-8}$ to $10^{-12}$ and
use three independent log-uniform generator rescalings, three generator
permutations and three random real orthogonal sector-basis rotations.
The $F$-boundary cases are also compared across two closure algorithms and
a direct span-closure test.

\begingroup
\footnotesize
\renewcommand{\arraystretch}{1.18}
\setlength{\tabcolsep}{5pt}
\begin{longtable}{@{}llll@{}}
\caption{Representative fixed-sector DLA checks.  ``One seed'' means that
each listed nonzero coordinate was tested separately with the connected
singlet-single backbone; parentheses give successful over tested seeds.
Boundary \(Y\) counts refer to disjoint coordinates unless shared coordinates
are stated explicitly.  For each listed sector,
\(\mathcal F:=\sum_{p<q}(\mathcal F_{pq}|_{\Hsp})\) is the real span of
the restricted local five-axis \(F\)-block tangent spaces defined in
Section~\ref{app:F-boundary}. Here \(\Ypooldisj\) is the disjoint subpool
and \(\Ypool\) includes both disjoint and shared-orbital coordinates.}
\label{tab:dla-checks}\\
\toprule
\tblcell{0.19\textwidth}{Sector}
& \tblcell{0.09\textwidth}{\(d\) / \(\binom d2\)}
& \tblcell{0.19\textwidth}{Matrix provenance}
& \tblcell{0.35\textwidth}{Selected closure ranks} \\
\midrule
\endfirsthead
\multicolumn{4}{@{}l}{\textbf{Table \thetable} \textit{(continued)}}\\
\toprule
\tblcell{0.19\textwidth}{Sector}
& \tblcell{0.09\textwidth}{\(d\) / \(\binom d2\)}
& \tblcell{0.19\textwidth}{Matrix provenance}
& \tblcell{0.35\textwidth}{Selected closure ranks} \\
\midrule
\endhead
\midrule
\multicolumn{4}{r@{}}{\textit{Continued on the next page}}\\
\endfoot
\bottomrule
\endlastfoot
\tblcell{0.19\textwidth}{\(M=4,N=2,S=1\), electron boundary}
& \tblcell{0.09\textwidth}{\(6/15\)}
& \tblcell{0.19\textwidth}{Full-space fermionic projection}
& \tblcell{0.35\textwidth}{\(\Lpool:6/15\); \(\Lpool+\Xpool:6/15\); \(\mathcal F:6/15\); \(\Lie(\Lpool,Y):15/15\) for one disjoint seed \((3/3)\).} \\
\tblcell{0.19\textwidth}{\(M=4,N=6,S=1\), hole boundary}
& \tblcell{0.09\textwidth}{\(6/15\)}
& \tblcell{0.19\textwidth}{Full-space fermionic and five-axis projection}
& \tblcell{0.35\textwidth}{\(\Lpool:6/15\); \(\mathcal F:6/15\); \(\Lie(\Lpool,Y):15/15\) for one disjoint seed \((3/3)\) and one shared-orbital seed \((12/12)\).} \\
\tblcell{0.19\textwidth}{\(M=6,N=3,S=3/2\), electron boundary}
& \tblcell{0.09\textwidth}{\(20/190\)}
& \tblcell{0.19\textwidth}{Full-space fermionic projection}
& \tblcell{0.35\textwidth}{\(\Lpool:15/190\); \(\mathcal F:15/190\); \(\Lie(\Lpool,Y):190/190\) for one disjoint seed \((45/45)\).} \\
\tblcell{0.19\textwidth}{\(M=3,N=3,S=1/2\), interior, \(D=2\)}
& \tblcell{0.09\textwidth}{\(8/28\)}
& \tblcell{0.19\textwidth}{Full-space fermionic projection}
& \tblcell{0.35\textwidth}{\(\Lpool:3/28\); \(\Lie(\Lpool,X):28/28\) for one seed \((3/3)\).} \\
\tblcell{0.19\textwidth}{\(M=4,N=4,S=1\), interior}
& \tblcell{0.09\textwidth}{\(15/105\)}
& \tblcell{0.19\textwidth}{Full-space fermionic projection}
& \tblcell{0.35\textwidth}{\(\Lpool:6/105\); \(\Lie(\Lpool,X):105/105\) for one seed \((6/6)\); \(\Lpool+\Ypooldisj:15/105\); \(\Lie(\Lpool,Y_{\rm shared}):105/105\) \((12/12)\); and \(\Lpool+\Ypool:105/105\).} \\
\tblcell{0.19\textwidth}{\(M=5,N=2,S=0\), interior, \(D=5\)}
& \tblcell{0.09\textwidth}{\(15/105\)}
& \tblcell{0.19\textwidth}{Full-space fermionic projection}
& \tblcell{0.35\textwidth}{\(\Lpool:10/105\); \(\Lie(\Lpool,X):105/105\) for one seed \((10/10)\).} \\
\tblcell{0.19\textwidth}{\(M=8,N=4,S=2\), electron boundary}
& \tblcell{0.09\textwidth}{\(70/2415\)}
& \tblcell{0.19\textwidth}{Sector-native restricted matrices}
& \tblcell{0.35\textwidth}{\(\Lpool+\Xpool:28/2415\); the full intermediate-triplet coordinate pool gives \(\Lpool+\Ypool:2415/2415\).} \\
\end{longtable}
\endgroup
\FloatBarrier

The ranks were stable under the tolerance, rescaling and basis changes
described above.  The $F$-boundary calculations also agreed across
the independent closure tests.  The counts \(3\) and \(45\) exhaust the disjoint
pair-pair coordinates in the listed boundary sectors.  The \(F\)-family rows
correspond, after translating the spin-label convention of
Ref.~\maincite{Anselmetti2021}, to representative edge entries in its
Tables~III--IV.  The \(M=3\) and \(M=5\) interior rows exercise the
\(D=2\) and \(D\ge3\) slice dimensions,
respectively.

\subsection{Fixed-coordinate boundary mechanism validation}
\label{app:fixed-coordinate-vqd}

\subsubsection{Systems and protocol}
\label{si:fixed-coordinate-vqd-methods}

This separate ten-start suite tests fixed completing coordinates on electron
and hole boundaries; it does not supply the main-text resource comparisons.
The H$_2$, H$_3$, and H$_4$ calculations use all 6-31G spatial orbitals,
interleaved $(p\alpha,p\beta)$ Jordan--Wigner ordering, and the systems in
Table~\ref{tab:vqd-systems}. Their chemically specified references are
\begin{align*}
T(0,1)&=\frac{a^\dagger_{0\alpha}a^\dagger_{1\beta}
                 +a^\dagger_{0\beta}a^\dagger_{1\alpha}}{\sqrt2}\ket{\mathrm{vac}},\\
Q(0,1,2)&=\frac{\ket{0\beta,1\alpha,2\alpha}
+\ket{0\alpha,1\beta,2\alpha}
+\ket{0\alpha,1\alpha,2\beta}}{\sqrt3},\\
D(0,1,2,3)&=\ket{0\alpha,1\alpha,2\alpha,3\alpha}.
\end{align*}
Creation operators in the occupation kets are ordered by increasing spatial
index; the explicit $T(0,1)$ definition fixes its relative fermionic sign.
The O$_2$ Hamiltonian and reference are those specified in
Sec.~\ref{si:benchmark-systems-details}: canonical ROHF/cc-pVDZ at
1.21~\AA{}, with five doubly occupied orbitals frozen and a
$\pi$-CAS(6e,4o) triplet reference.

Each repeat applies all $L_{pq}$ in lexicographic order, followed by either
no completing coordinate, all $X_{pq}$, one fixed
$Y_\star=Y_{\{0,1\},\{0,2\}}$, or one fixed $C_\star=C_{0;1,2}$.
These coordinates are chosen before optimization. Projected-zero $X$
matrices are removed, so $L$ and $L+X$ have identical effective products,
initial points, and parameter counts on these boundaries.
Sequential VQD uses the objective in Eq.~\mainref{eq:vqd-objective} with
$\beta=1$~Ha. Each lower-root deflator is its own lowest-objective
variational output, and each root's restarts are selected by that objective.
Exact eigenvectors and overlaps enter only post-hoc error and fidelity
checks, never reference, coordinate, restart, or root selection.

Each point uses ten starts: zero, a linear ramp from $-0.05$ to $0.05$,
and eight Gaussian draws with standard deviation $0.05$ and base seed
20260707. L-BFGS-B uses bounds $[-\pi,\pi]$, at most 1000 iterations,
$\mathtt{ftol}=10^{-10}$, and $\mathtt{gtol}=10^{-8}$.
Repetition grids are $1,2,3,4,8$ for H$_2$, $1,2,4$ for H$_3$, and
$1,2,3,4,6$ for H$_4$ and O$_2$.
Success requires $|\Delta E_n|\le0.1$~mHa and target fidelity at least
$0.99$. For O$_2$, degenerate roots 1 and 2 use their joint eigenspace
projector; individual projector fidelities do not establish collective
coverage of that subspace.

\subsubsection{Fixed-coordinate results}

Figure~\ref{fig:vqd-depth} summarizes the error curves for all four systems,
and Table~\ref{tab:vqd-systems} records the largest tested repetition
counts. The H$_2$ and H$_3$ tests use non-highest magnetic components;
O$_2$ tests the hole boundary. The fixed $Y$ and $C$ seeds remove the
observed orbital-only plateaus. The finite-product success counts remain
separate from DLA completeness and from the twenty-chain benchmark
statistics below.

\begin{figure}[!htbp]
\centering
\includegraphics[width=\textwidth]{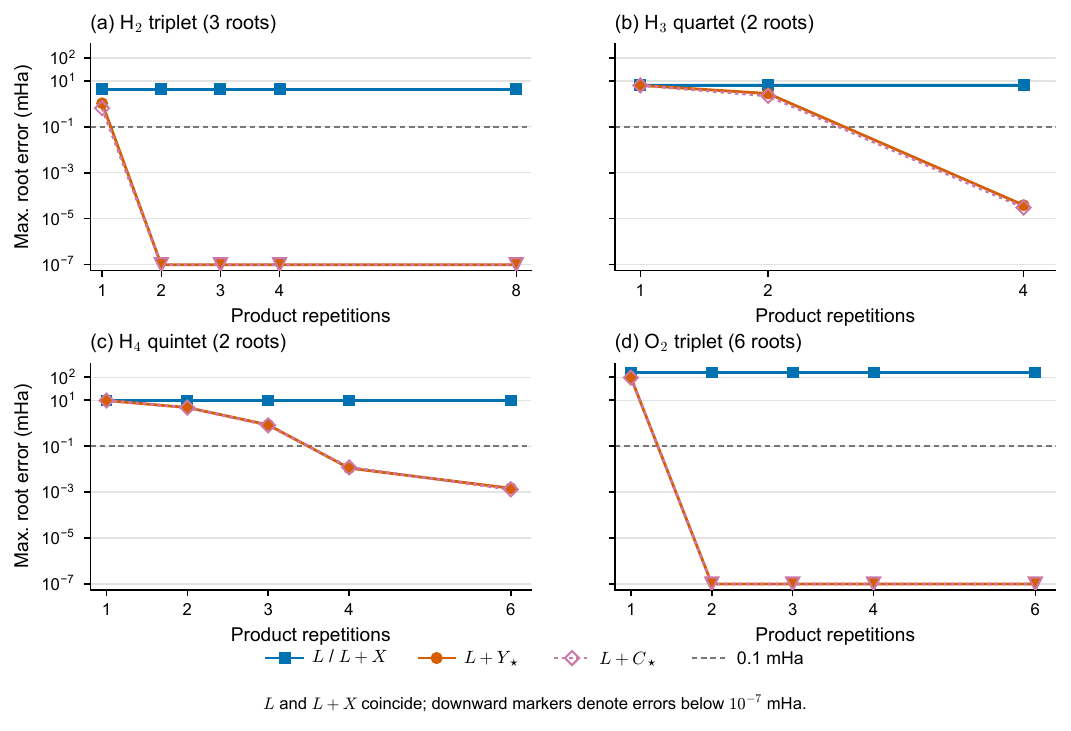}
\caption{\textbf{Fixed completing coordinates remove electron- and
hole-boundary plateaus.} Each panel shows the maximum error over
objective-selected roots versus product repetitions in the separate
fixed-coordinate suite. $L$ and $L+X$ coincide because $X$ projects to zero.
The dotted line marks 0.1~mHa; values below $10^{-7}$~mHa are displayed at
that floor. Roots are selected individually by their VQD objectives,
not by exact-state diagnostics.}
\label{fig:vqd-depth}
\end{figure}

\begin{table}[!htbp]
\centering\footnotesize
\caption{Fixed-coordinate mechanism-validation systems and endpoints.
The lower panel reports maximum objective-selected root energy errors
(in mHa) at the largest tested repetition count.
$L+X$ has the same errors as $L$; errors below $10^{-7}$~mHa are reported
as $<10^{-7}$.
Success counts for $L/Y/C$ sum root--restart pairs, with denominators
30 for H$_2$ and 20 for H$_3$ and H$_4$; they are not whole-chain success
counts. A dash indicates that this statistic is not reported for the
O$_2$ suite.}
\label{tab:vqd-systems}
\begin{tabular*}{\textwidth}{@{\extracolsep{\fill}}lll@{}}
\toprule
System and geometry & $(M,N,S,M_S);d$ & Roots; reference\\
\midrule
H$_2$, $R=0.80$~\AA{} & $(4,2,1,0);6$ & $0$--$2$; $T(0,1)$\\
H$_3$, equilateral $R=1.50$~\AA{} & $(6,3,3/2,1/2);20$ & $0$--$1$; $Q(0,1,2)$\\
H$_4$, linear spacing $R=1.50$~\AA{} & $(8,4,2,2);70$ & $0$--$1$; $D(0,1,2,3)$\\
O$_2$, $R=1.21$~\AA{} & $(4,6,1,1);6$ & $0$--$5$; $\pi$-shell determinant\\
\bottomrule
\end{tabular*}
\smallskip
\input{benchmark/si_fixed_coordinate_endpoints.tex}
\end{table}
\FloatBarrier

\subsection{Main-benchmark initialization keys and coordinate ordering}
\label{si:benchmark-keys}

Each uniquely keyed parameter is initialized by a deterministic PCG64 stream,
uniform on \([-0.01,0.01]\).  For a canonical parameter key \(q\), the
unsigned 32-bit seed is the big-endian integer represented by the first four
bytes of the SHA-256 digest of the compact, key-sorted JSON serialization
of the following object, shown schematically with run-specific fields:
\begin{center}
\begin{minipage}{0.70\textwidth}
\small
\begin{verbatim}
{
  "namespace": "sector-ansatz-benchmark-v1",
  "parts": ["paired_initial_parameter",
            protocol, system, reference, r, q]
}
\end{verbatim}
\end{minipage}
\end{center}
Here \(r\) is the restart index, the reference identifier includes the VQD
root when applicable, and
the key \(q\) has the literal format \texttt{layer=LLL/operator}, including
the layer and generator label.
A fresh NumPy PCG64 generator seeded in this way supplies one uniform draw.
Thus parameters with the same canonical key receive bitwise-identical initial
values whenever the protocol, system, reference and restart also agree;
an equal restart number alone does not imply equal perturbations across
different ansatz families or references.  The construction above fixes these
initial values deterministically.  The public numerical archive retains the
adopted endpoint parameter arrays and their SHA-256 digests.

The boundary redundancy coordinates are fixed without Hamiltonian or target
information.  Every H$_4$ backbone contains the path edges
\((0,1),(1,2),\ldots,(6,7)\).  The added chords, in selection order, are
\[
 (0,2),(1,3),(2,4),(3,5),(4,6),(5,7),(0,3),
 (1,4),(2,5),(3,6),(4,7),(0,4),(1,5),(2,6),
\]
with the remaining chords included at the dense endpoint.  This order is
obtained from the integer tuple
\((N_{\rm CX}^{\rm LNN},N_{\rm CX}^{\rm A2A},q-p,p,q)\); after a prefix is
selected, its union with the path is executed in lexicographic \((p,q)\)
order.  The nested H$_4$ backbones contain 7, 10, 14, 21 and 28 \(L\)
directions, respectively, and every repeat ends with \(C_{1;0,2}\).

For O$_2$, the cost-ranked conditioned-coordinate prefix is
\[
 C_{1;0,2},\quad C_{2;1,3},\quad C_{0;1,2},\quad C_{1;2,3},
\]
for the one-, two- and four-coordinate variants; the twelve-coordinate
endpoint contains every \(C_{c;pq}\) with pairwise distinct indices and
\(p<q\).  Selection uses the integer tuple
\((N_{\rm CX}^{\rm LNN},\allowbreak N_{\rm CX}^{\rm A2A},\allowbreak V_{\rm line},
\allowbreak g_{\max},\allowbreak w_{\rm tot},\allowbreak c,\allowbreak p,\allowbreak q)\), where \(V_{\rm line}\) is the summed
line-edge traversal count, \(g_{\max}\) is the largest direct Pauli-support
gap and \(w_{\rm tot}\) is the total Pauli weight of the fixed eight-string
factorization.  Each selected set is executed in
lexicographic \((c,p,q)\) order following the path backbone.  Circuit
repetition conventions are specified in
Methods, Section~\mainref{sec:benchmark-systems}; the optimization budgets,
completion rules, and record-adoption policy are specified below.
These twenty-start ensembles are distinct from the five-chain repair scan.

\subsection{Benchmark optimization and compilation protocol}
\label{si:main-benchmark-methods}

\subsubsection{Benchmark systems and references}
\label{si:benchmark-systems-details}

The boundary tasks use linear H$_4$/6-31G at 1.50~\AA{} adjacent spacing
with all eight spatial orbitals, \(N=4\), \(S=M_S=2\), and \(d=70\).
The reference occupies orbitals \(0,1,2,3\) with four \(\alpha\) electrons;
roots 0--1 are targeted. O$_2$ uses a spin-free canonical-ROHF/cc-pVDZ
Hamiltonian at 1.21~\AA{}. Freezing five doubly occupied core and
\(\sigma\)-type orbitals leaves the four \(\pi/\pi^\ast\) orbitals as
CAS(6e,4o), with \(S=M_S=1\), \(d=6\), and roots 0--5.
Its reference doubly occupies both bonding-\(\pi\) orbitals and singly
occupies both \(\pi^\ast\) orbitals with \(\alpha\) electrons.

The main benchmark uses complete fixed-spin multiplicity spaces and no
additional point-group projector.  Linear H\(_6\) has adjacent H--H spacing
1.50~\AA{}, an STO-3G full CAS(6e,6o), \(S=M_S=0\), \(d=175\), and the
canonical closed-shell RHF reference.  Planar tetramethyleneethane uses the
geometry in Sec.~I.F of the Supplemental Material to Ref.~\cite{Burton2024TUPS},
with the STO-3G carbon-\(\pi\) CAS(6e,6o) specified in that article;
state-specific VQE is performed independently in the
singlet \((S=M_S=0,d=175)\) and triplet \((S=1,M_S=0,d=189)\)
sectors.  The singlet uses the
canonical \(\pi\)-RHF reference and the triplet a spin-adapted HOMO--LUMO
\(T_0\) reference declared before active-space diagonalization.  N\(_2\) uses
STO-3G with four tracked occupied orbitals frozen and a CAS(6e,6o) singlet
sector of dimension 175.  Eight bond lengths,
\(R=1.10,1.20,1.30,1.50,1.80,2.25,3.00,4.00\)~\AA{}, are optimized
independently without parameter continuation.  Occupied/active orbital
identities are propagated by nearest-neighbour cross-AO overlap with separate
occupied and virtual assignments; no exact active-space energy enters the
tracking rule.

Interior one-seed products use \(X_{23}\) for the six-orbital active spaces;
boundary one-seed products use \(C_{1;0,2}\). QNP-Q, QNP-F, tUPS, and tUPS
with a perfect-pairing reference are implemented with the same Hamiltonian,
orbital ordering, and complete fixed-spin scope. The common reference is used
except for the explicitly labelled tUPS (PP ref.) singlet row, which uses its
target-independent alternating perfect-pairing occupation pattern; its triplet
component uses the common spin-adapted reference. All methods in the main
interior comparison use fixed orbitals and the local optimization procedure
specified below.

In the five-chain boundary-repair scan, every root uses the common reference
specified above. Boundary-effective tUPS retains the ordered adjacent-pair
orbital rotations of the alternating tiles, and tUPS+$C$ appends
\(C_{1;0,2}\) per repeat.

\subsubsection{Optimization and statistics}

The main benchmark calculations use double precision.  Each fixed-ansatz
checkpoint in the twenty-start ensemble has 20 independent initializations,
with parameters drawn uniformly from $[-0.01,0.01]$.  L-BFGS-B uses analytic
gradients, bounds $[-\pi,\pi]$, $\mathtt{ftol}=10^{-12}$,
$\mathtt{gtol}=10^{-8}$, and $\mathtt{maxls}=40$.  H$_4$ and O$_2$
respectively target roots 0--1 and 0--5 with $\beta=2$~Ha.  Each
initialization supplies its own VQD chain, and root $n$ uses only the
lower-root states from that chain.

Strict success requires \(|\Delta E|\le0.1\)~mHa and fidelity
(or degenerate-eigenspace fidelity) at least 0.99 for every state in the
task. TME joint success additionally requires a singlet--triplet-gap error
at most 0.1~mHa. The supplementary N\(_2\) initialization statistic
requires at least 16 of 20 strict successes independently at every
geometry.

The per-root or per-state iteration/evaluation limits are collected in
Supplementary Table~\ref{tab:si-optimization-budgets}.  The extended budget
was used only for declared reruns of records that had exhausted an original
iteration limit; it was not applied as an additional set of random starts.

\begin{table}[!htbp]
\centering\footnotesize
\caption{Optimization limits used in the reported numerical datasets.  A dash
means that no extended-budget records were adopted for that dataset.  Limits
are stated as maximum optimizer iterations/objective--gradient evaluations per
root or state.}
\label{tab:si-optimization-budgets}
\begin{tabular*}{\textwidth}{@{\extracolsep{\fill}}lll@{}}
\toprule
Dataset & Original limit & Extended limit\\
\midrule
H$_4$, 20-start boundary ensemble & 20,000/40,000 & 200,000/400,000\\
H$_6$, 20-start interior ensemble & 50,000/100,000 & 200,000/400,000\\
O$_2$, 20-start boundary ensemble & 2,000/5,000 & --\\
TME and N$_2$, 20-start interior ensembles & 50,000/100,000 & --\\
H$_4$, five-chain repair scan & 100,000/200,000 & --\\
O$_2$, five-chain repair scan & 50,000/100,000 & --\\
\bottomrule
\end{tabular*}
\end{table}

Exactly 27 H$_4$ chains and four H$_6$ single-state trajectories that were
right-censored by their original iteration limits were rerun independently
from the same deterministic initial parameter vectors under the extended
limits.  These were fresh optimizations, not continuations from the censored
parameters.  Only those 31 declared coordinates replace their original
records in the adopted extended-budget summaries.  The public archive retains
the aggregate original-versus-extended comparison, the adoption policy, and
all adopted plotting endpoints, but not optimizer trajectories or superseded
per-run payloads.  Two additional H$_6$ reruns whose originals were not
right-censored are excluded.  No records are replaced for O$_2$, TME, or
N$_2$.

All 20 starts remain in every success-count denominator.  In an
original-budget summary, an incomplete chain is counted as a nonsuccess and
omitted only from the conditional median and interquartile range.  After the
declared replacements, all 20 chains are complete in every adopted H$_4$ and
O$_2$ boundary configuration, so the final error bands and success counts use
the same starts.

The five-chain boundary scan uses L-BFGS-B bounds, gradient tolerances, and
line-search settings as above, with iteration/objective--gradient limits of
100,000/200,000 for H$_4$ and 50,000/100,000 for O$_2$.  The common
repetition grids are $1,2,4,6,8,12$ for H$_4$ and $1,2,3,4,5,6$ for
O$_2$.  The tUPS scans also include repetitions $16,22,28$ for H$_4$ and
$12,16,24,32$ for O$_2$.  At the reported endpoints, the selected
optimizations satisfied
an L-BFGS-B stopping condition before reaching the iteration and evaluation
limits.

The deterministic initialization keys and target-independent coordinate
ordering used by the main benchmark are recorded in
Sec.~\ref{si:benchmark-keys}.

We count CNOT gates after binding the optimized parameters and compiling
each complete ansatz with Q\textsuperscript{2}Chemistry 0.3.0.  All methods use
the same all-to-all connectivity and interleaved spin-orbital wire order.
Adjacent orbital rotations and pair transfers use the exact QNP orbital-
rotation and pair-exchange templates.  Nonadjacent $L$ and $X$ generators,
and every conditioned rotation $C$, are synthesized from their respective
commuting Jordan--Wigner Pauli factors.  In particular, the eight-factor
commutativity statement applies to one $C$ block.

QNP-F is compiled separately from the published one-parameter circuits for
its five-axis local block \cite{Anselmetti2021}.  The one-particle and one-hole
coordinates use the corresponding source circuits, with the source angle
convention mapping each bound coordinate \(\theta\) to \(2\theta\).  For the
remaining three coordinates \(v_X,v_{\rm BU},v_{\rm BL}\), the simultaneous
rotation is first evaluated on the local two-electron singlet space in the
ordered basis \((\ket{P_p},\ket{P_q},\ket{S_{pq}})\).  We then solve the exact
\(SO(3)\) factorization
\[
\exp(v_XK_{\rm PX}+v_{\rm BU}K_{\rm PBU}+v_{\rm BL}K_{\rm PBL})
=\exp(zK_{\rm PBL})\exp(yK_{\rm PBU})\exp(xK_{\rm PX}).
\]
The published \(P_X\), \(P_{\rm BU}\), and \(P_{\rm BL}\) source circuits are
applied in that order with circuit angles \(2x\), \(2y\), and \(-2z\),
respectively.  The reconstructed \(3\times3\) rotation is checked against the
simultaneous target before circuit decomposition.  This block-specific
factorization does not assume that the Pauli terms of a QNP-F generator
commute.

Each circuit is then decomposed into $H$, $R_z$, $X$, and CNOT gates and
simplified by one default optimization pass that cancels inverse gates and
merges adjacent rotations.
The same compilation convention is used for the supplementary interior
ablation and direct-repair resource tables below.

Counts cover the ansatz only and exclude reference preparation,
measurements, and VQD overlaps.  Independent
simulation checks the circuit states before and after compilation after basis
conversion and global-phase alignment.  For the main compiled circuits, the
largest amplitude and normalized-energy differences are
$2.23\times10^{-8}$ and $1.14\times10^{-13}$~Ha, respectively.  The
separate supplementary frozen-parameter batch also passes the stated
tolerances; its extrema are reported below.  Additional replay checks,
including reference-spin, degenerate-subspace, and N$_2$ embedding
diagnostics, are given in
Sec.~\ref{si:endpoint-replay}.

\subsection{Interior accuracy--resource diagnostics}

The interior figure uses variational-energy-selected outputs.
Initialization-success statistics are reported separately from the
energy--CX endpoints.

Supplementary Table~\ref{tab:si-interior-cx-depth} reports the CNOT counts
and two-qubit depths of the same selected circuits at the energy endpoints
and the TME strict-success checkpoints.

\begin{table}[!htbp]
\centering\small
\caption{Compiled ansatz resources under the common parameter-bound,
all-to-all protocol above. CNOT counts follow the main comparison: one
circuit for H$_6$, the sum for the two TME states, and the maximum over
the eight N$_2$ geometries. Two-qubit depth is the maximum over the task's
separately executed circuits.}
\label{tab:si-interior-cx-depth}
\begin{tabular*}{\textwidth}{@{\extracolsep{\fill}}llrrr@{}}
\toprule
Task / checkpoint & Method & Repeats & CNOTs & Max. 2Q depth\\
\midrule
H$_6$, energy endpoint & Tree+one $X$ & 24 & 816 & 411\\
 & tUPS & 8 & 880 & 256\\
N$_2$, energy endpoint & Tree+one $X$ & 30 & 1020 & 513\\
 & tUPS & 12 & 1320 & 384\\
TME, energy endpoint & Tree+one $X$ & 16 & 1088 & 275\\
 & tUPS & 12 & 2640 & 384\\
TME, strict-success checkpoint & Tree+one $X$ & 24 & 1632 & 411\\
 & tUPS & 12 & 2640 & 384\\
\bottomrule
\end{tabular*}
\end{table}

For fixed-orbital tUPS on N\(_2\), the maximum error over eight geometries
is 0.003800~mHa at 12 layers (1320 A2A CX), 10.507277~mHa at
16 layers (1760 CX), and 0.00004443~mHa at 24 layers (2640 CX).
The first maximum occurs at 1.80~\AA{}, while the latter two occur at
1.50~\AA{}.  At 1.50~\AA{}, all 20 starts of the 16-layer circuit return
errors between 10.507277 and 10.507434~mHa.  The selected restart, index 9,
stops after 742 of 50,000 allowed iterations and 800 of 100,000 allowed
function evaluations, with the L-BFGS-B relative-function-reduction criterion
satisfied.  Thus the selected output is not budget-censored.

The 12-, 16-, and 24-layer records use the same frozen physical problem
and reference.  A diagnostic Qiskit replay of the saved parameters agrees
with their stored sector states to approximately \(1.5\times10^{-13}\)
in Euclidean norm.  Padding the 180-parameter 12-layer solution with 60
zeros embeds it in the 16-layer circuit: the full-state norm difference
is \(2.5\times10^{-14}\), and the original approximately 0.000223~mHa
error at 1.50~\AA{} is preserved.  The rebound therefore reflects the
independent local searches, not reduced expressivity at greater repetition
count; the plotted outputs are retained without smoothing.

\input{benchmark/si_benchmark_evidence.tex}

%% file: benchmark/si_fixed_coordinate_endpoints.tex
\begin{tabular*}{\textwidth}{@{\extracolsep{\fill}}lrrrrl@{}}
\toprule
System & Repeats & $L$ error & $L+Y_\star$ error & $L+C_\star$ error & Successes $L/Y/C$\\
\midrule
H$_2$ & 8 & 4.348 & $<10^{-7}$ & $<10^{-7}$ & $0/30;\ 28/30;\ 28/30$\\
H$_3$ & 4 & 6.436 & $3.91\times10^{-5}$ & $3.08\times10^{-5}$ & $10/20;\ 20/20;\ 20/20$\\
H$_4$ & 6 & 9.669 & $1.46\times10^{-3}$ & $1.32\times10^{-3}$ & $10/20;\ 19/20;\ 20/20$\\
O$_2$ & 6 & 160.2 & $<10^{-7}$ & $<10^{-7}$ & --\\
\bottomrule
\end{tabular*}

%% file: benchmark/si_benchmark_evidence.tex
\subsection{Boundary operation choices, repetitions, and compiled CNOT counts}
\label{si:boundary-depth-cx}

We vary the number of orbital connections and the number of conditioned rotations separately, within families whose completeness is established in the article. For H$_4$, each repeated block contains one fixed $C$ and either the seven-edge path, a denser orbital graph, or all 28 orbital connections. For O$_2$, the orbital graph is a path and the repeated block contains one, two, four, or all twelve $C$ coordinates. Additional connections and coordinates are ordered by circuit cost and locality before optimization. All configurations use 20 initializations; the original and extended optimization budgets and the record-adoption policy are specified in Sec.~\ref{si:main-benchmark-methods} and compared in Supplementary Table~\ref{tab:si-boundary-budget}.

With the extended-budget reruns included, the path, ten-edge and dense
backbones for the $\mathrm{H}_4$/6-31G quintet attain 20/20 strict successes
at 10, 8 and 4 repetitions, respectively (\cref{fig:benchmark-boundary}).
The corresponding circuits use 80, 88 and 116 parameters per root, and
maximum per-root CNOT counts of 600, 864 and 2704. Additional
edges thus reduce the number of repetitions required at these endpoints, at
the cost of more parameters and CNOTs. At the same path and ten-edge
checkpoints, the extended-budget reruns increase strict-success counts from
14/20 to 20/20 and from 17/20 to 20/20, respectively; dense already reaches
20/20 under the original budget. Supplementary
Table~\ref{tab:si-boundary-budget} records these repetition counts and
retains the original-budget diagnostics.

The $\mathrm{O}_2$ hole boundary shows a complementary dependence on seed count under the original budget. One $C$ coordinate attains 19/20 strict successes at the largest tested circuit size, whereas two coordinates attain 20/20 with 10 parameters and 152 CNOTs per root. Four or all twelve coordinates give no further improvement in the success count and require more resources.

\begin{figure}[!htbp]
\centering
\includegraphics[width=\textwidth]{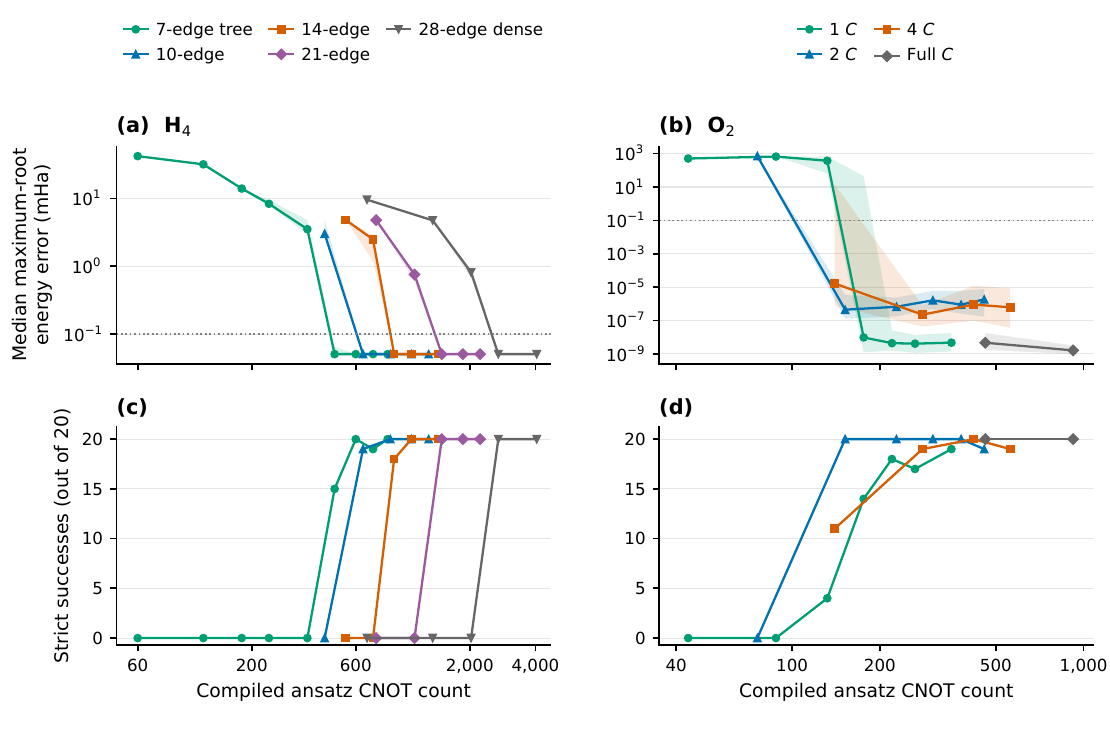}
\caption{Boundary redundancy versus compiled ansatz CNOT count. (a,b) Median maximum-root energy error, with interquartile shading over 20 completed VQD chains; dotted lines mark 0.1 mHa. (c,d) Strict successes among the same starts: every targeted root has energy error at most 0.1 mHa and fidelity at least 0.99. The horizontal coordinate is the maximum per-root count of the representative chain selected by minimum summed variational root energies. H$_4$ varies orbital connectivity with one $C$; O$_2$ varies the number of $C$ coordinates on a path. All 48 configurations are shown on logarithmic horizontal axes. Counts describe representative circuits, whereas the vertical statistics summarize all 20 starts. The extended-budget reruns are included.}
\label{fig:benchmark-boundary}
\end{figure}

The distinction between algebraic completeness and efficient convergence also motivates enlarging complete Pauli pools in earlier adaptive studies \cite{Sapova2022,Viswanathan2026}.

The stored parameters and adopted-record membership were verified for all 960 chains across these 48 configurations. The 172 root circuits of the representative chains all pass the independent state and energy checks in Section~\ref{si:endpoint-replay}. Compilation leaves the selected parameters and ensemble statistics unchanged.

The separate five-chain repair scan in main-text Fig.~\mainref{fig:benchmark-boundary-repair} follows Methods, Section~\mainref{sec:additional-controls}. All 222 selected root circuits across its 55 configurations pass the same compilation checks. The main figure shows all 50 checkpoints within its displayed windows; five H$_4$ Dense $L$ checkpoints lie beyond the window and extend the same energy plateau to 7728 CNOTs. Nonmonotonic checkpoints are retained, and errors below 0.01 mHa are placed at the stated display floor. This energy comparison remains separate from the twenty-start strict-success statistic and the common-deflator control below.

\FloatBarrier
\subsection{Initialization sensitivity and budget dependence}
\label{si:benchmark-protocol}

The main-benchmark initialization, optimization, chain-completion criteria,
original/extended budgets, and record-adoption policy are specified in
Sec.~\ref{si:main-benchmark-methods}.  The PP-reference comparisons use
fixed orbitals and local optimization, not the full literature pp-tUPS
search procedure.

The original and extended-budget boundary endpoints are
listed in Table~\ref{tab:si-boundary-budget}.  At the H$_6$ Tree+one $X$
24-repeat checkpoint, original-budget restart 5 already has error
0.005424~mHa.  The main figure instead uses the adopted extended-budget restart 16
and its own 816-CNOT compiled circuit, giving 0.004571~mHa.  The original result
establishes energy attainment at this circuit size; its parameters are not
assigned the other restart's compiled resource count.

\begin{table}[!htbp]
\centering\footnotesize
\caption{Boundary initialization sensitivity at unchanged checkpoints.
Counts are out of 20; extended-budget results are not matched-budget
success probabilities.  The CNOT count is the maximum ansatz cost per root.}
\label{tab:si-boundary-budget}
\begin{tabular*}{\textwidth}{@{\extracolsep{\fill}}llrrrrr@{}}
\toprule
Task & Circuit & Repeats & Params. & Original & Extended & CNOTs\\
\midrule
H$_4$ & Tree+one $C$ & 10 & 80 & 14 & 20 & 600\\
 & 10 $L$+one $C$ & 8 & 88 & 17 & 20 & 864\\
 & Dense+one $C$ & 4 & 116 & 20 & -- & 2704\\
O$_2$ & Tree+two $C$ & 2 & 10 & 20 & -- & 152\\
\bottomrule
\end{tabular*}
\end{table}

\subsection{Independent circuit validation}
\label{si:endpoint-replay}

The source data for the two main resource figures and Supplementary
Fig.~\ref{fig:benchmark-boundary} comprise 187 checkpoints, represented by
650 root or state circuits.  The supplementary interior dataset and
direct-repair table add a separate 511-circuit frozen-parameter compilation
batch; 256 of its 506 interior circuits repeat inputs from the main interior
comparison, while 250 supply the additional ablation methods.  Every source
and compiled circuit in these saved batches is replayed against its saved
production state, including every selected N$_2$ geometry/checkpoint circuit
in the figures. The comparison includes sector embedding,
fermionic ordering and phase alignment.  No variational optimization is
repeated.  The largest source-state amplitude errors are
$4.29\times10^{-13}$ in the main batch and $3.96\times10^{-12}$ in the
supplementary batch; Table~\ref{tab:si-replay-summary} gives the
compiled-state checks.

\begin{table}[!htbp]
\centering\footnotesize
\caption{Independent validation of the compiled circuits used in the resource figures. The acceptance thresholds are $2\times10^{-6}$ for the maximum phase-aligned amplitude error and $5\times10^{-10}$ Ha for the normalized energy difference from the saved state. Every listed circuit passes both checks.}
\label{tab:si-replay-summary}
\begin{tabular*}{\textwidth}{@{\extracolsep{\fill}}lrrrr@{}}
\toprule
Comparison & Points & Circuits & Max. amplitude error & Max. $|\delta E|$ (Ha)\\
\midrule
Boundary repair & 55 & 222 & $1.65\times10^{-8}$ & $8.53\times10^{-14}$\\
Boundary redundancy & 48 & 172 & $2.23\times10^{-8}$ & $8.53\times10^{-14}$\\
Interior comparison & 84 & 256 & $2.15\times10^{-8}$ & $1.14\times10^{-13}$\\
SI interior dataset & 164 & 506 & $3.15\times10^{-7}$ & $3.78\times10^{-12}$\\
Direct repair & 5 & 5 & $1.92\times10^{-13}$ & $4.44\times10^{-16}$\\
\bottomrule
\end{tabular*}
\end{table}

For reproducibility, accepted parameters, source hashes and individual
residuals are retained with the numerical data.  The public data package omits
compiled gate-sequence payloads; the independent Qiskit workflow reconstructs
the source circuits from the frozen parameters and reports its compilation and
validation results separately. Energy is normalized using the full state norm,
without separately renormalizing its projection into the target sector.
Validation covers all circuits in these comparisons, but does not imply an
ordering between nearly exact energy tails.

\FloatBarrier
\subsection{Internal compression ablations and task-level accuracy}

The four Dense/Tree $\times$ one/full $X$ constructions and QNP-F are
retained in Fig.~\ref{fig:si-internal-ablation} and
Table~\ref{tab:si-internal-ablation}. Dense uses all $\binom{M}{2}$ orbital rotations, whereas Tree uses the $M-1$ rotations along a path. The one/full $X$ choices use one fixed pair-transfer generator or all pair-transfer generators, respectively. They use the same objective-selected
energy criterion as the interior main figure.  All existing checkpoints remain visible, including search
rebounds.  The complete 164-point dataset accompanies the numerical results.  These additional results separate coordinate redundancy
from backbone density; they do not imply that adding generators guarantees
a better local-search result.

The ablation figure and common-deflator table use the same post-bind
Q\textsuperscript{2}Chemistry all-to-all protocol as the main resource
figures.  Their 511 frozen circuits all pass the common state and energy
tolerances.  The supplementary compilation changes only the resource axis;
the selected parameters, energies, fidelities and checkpoint set are
unchanged.

\begin{figure}[!htbp]
\centering
\includegraphics[width=\textwidth]{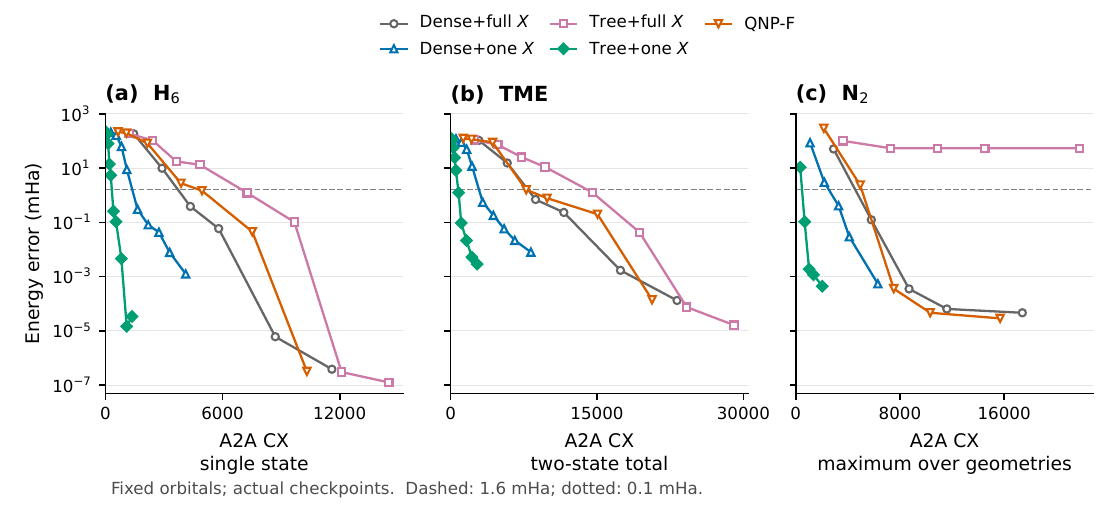}
\caption{\textbf{Separate backbone and seed compression in interior sectors.}
The four proposed products and QNP-F use existing recorded outputs.  Errors
are the H$_6$ ground-state energy error, maximum TME two-state energy error, and maximum
N$_2$ eight-geometry error; costs are respectively one-state, two-state sum,
and maximum geometry CNOT counts under the common
Q\textsuperscript{2}Chemistry protocol. Threshold crossings are not
interpolated, and no monotonic envelope is applied. The complete 164-point
dataset is validated in the separate 506-circuit supplementary batch.}
\label{fig:si-internal-ablation}
\end{figure}
\begin{table}[!htbp]
\centering\footnotesize
\caption{Internal ablation endpoints at 0.1~mHa, selected by smallest
observed compiled CNOT count under the common
Q\textsuperscript{2}Chemistry protocol among qualifying checkpoints. A star instead reports
the largest tested checkpoint when the threshold is not reached.}
\label{tab:si-internal-ablation}
\input{benchmark/si_internal_ablation.tex}
\end{table}

For the selected TME Tree+one $X$ pair at 16 repeats per state, the singlet
and triplet errors are respectively 0.065706 and 0.094966~mHa
(restarts 0 and 3).  The variational and exact gaps are 3.321437 and
3.292177~mHa, giving 0.029260~mHa gap error.  With normalized states in the
same frozen active-space problem, the sector variational principle gives
$\delta_S,\delta_T\ge0$ and hence
\[
 |\Delta_{ST}^{\rm var}-\Delta_{ST}^{0}|
   =|\delta_T-\delta_S|\le\max(\delta_T,\delta_S).
\]
This independently selected energy pair is not itself a joint-restart
success statistic.  The separate twenty-paired-start audit is reported in
Table~\ref{tab:si-tme-joint}: joint strict success requires both state errors
to be at most 0.1~mHa, both fidelities to be at least 0.99, and the gap error
to be at most 0.1~mHa.  At the 16-repeat Tree+one \(X\) energy endpoint,
none of the 20 paired starts meets this joint criterion; its first tested
checkpoint with at least 16 joint successes is 24 repeats.  This paired-start
count describes the specified matched-initialization audit, while the plotted
energy endpoint independently selects the two spin sectors.

\begin{table}[!htbp]
\centering\footnotesize
\caption{TME joint-initialization robustness.  Each row is the first tested
checkpoint at which at least 16 of 20 paired starts satisfy the joint strict
criterion.  Parameter counts are given per state and in total for the two
independently optimized spin sectors.}
\label{tab:si-tme-joint}
\begin{tabular*}{\textwidth}{@{\extracolsep{\fill}}lrrrr@{}}
\toprule
Method & Repeats/state & Params./state & Total params. & Joint strict \\
\midrule
Tree+one $X$ & 24 & 144 & 288 & 20/20 \\
Dense+one $X$ & 12 & 192 & 384 & 20/20 \\
Dense+full $X$ & 6 & 180 & 360 & 20/20 \\
Tree+full $X$ & 10 & 200 & 400 & 20/20 \\
QNP-Q & 35 & 176 & 352 & 18/20 \\
QNP-F & 19 & 240 & 480 & 20/20 \\
tUPS & 12 & 180 & 360 & 19/20 \\
tUPS (PP ref.) & 12 & 180 & 360 & 19/20 \\
\bottomrule
\end{tabular*}
\end{table}

\begin{samepage}
For N$_2$, Table~\ref{tab:si-n2-success} gives the per-geometry auxiliary
strict-success counts.  The companion full-curve table also retains the mean
absolute error and
$\mathrm{NPE}=\max_R\delta E(R)-\min_R\delta E(R)$.
NPE describes curve shape and is not used without absolute-error information.
\par
\end{samepage}

\begin{table}[!htbp]
\centering\scriptsize
\caption{N$_2$ strict successes out of 20 at each bond length (in \AA{}). For methods
that attain the initialization-robustness criterion, the row is the first tested
checkpoint with at least 16 successes at every geometry.  A dagger denotes
the largest tested checkpoint for a method that does not attain that
criterion.  These 20-start success counts are separate from the
variational-energy-selected endpoints in the main figure.}
\label{tab:si-n2-success}
\begin{tabular*}{\textwidth}{@{\extracolsep{\fill}}lrrrrrrrrrr@{}}
\toprule
Method & Rep. & Params. & 1.10 & 1.20 & 1.30 & 1.50 & 1.80 & 2.25 & 3.00 & 4.00 \\
\midrule
Tree+one $X$ & 40 & 240 & 20 & 20 & 20 & 20 & 17 & 20 & 18 & 16 \\
Dense+one $X$ & 15 & 240 & 20 & 20 & 20 & 18 & 20 & 20 & 19 & 17 \\
Dense+full $X$ & 6 & 180 & 20 & 20 & 20 & 20 & 20 & 20 & 20 & 20 \\
Tree+full $X^{\dagger}$ & 18 & 360 & 0 & 0 & 0 & 0 & 0 & 0 & 0 & 0 \\
QNP-F & 29 & 365 & 20 & 20 & 20 & 20 & 18 & 20 & 20 & 20 \\
QNP-Q$^{\dagger}$ & 72 & 360 & 16 & 15 & 6 & 1 & 0 & 9 & 8 & 13 \\
tUPS$^{\dagger}$ & 24 & 360 & 20 & 20 & 17 & 4 & 6 & 7 & 19 & 19 \\
tUPS (PP ref.)$^{\dagger}$ & 24 & 360 & 20 & 20 & 12 & 9 & 14 & 6 & 13 & 10 \\
\bottomrule
\end{tabular*}
\end{table}

\begin{table}[!htbp]
\centering\footnotesize
\caption{N$_2$ full-curve diagnostics at the main-table endpoints, in mHa.
All eight geometries use a common repetition count for each method.  The star retains QNP-Q's
largest tested, nonattaining checkpoint rather than omitting that method.}
\input{benchmark/si_n2_pes_metrics.tex}
\end{table}

\FloatBarrier

\subsection{Direct boundary repair}

\label{si:common-deflator-protocol}
This is a separate common-deflator control, not the five-chain sequential
scan of main-text Fig.~\mainref{fig:benchmark-boundary-repair}.
The H$_4$ repair comparison targets root 1 with a common variational
root-0 state $\widetilde\psi_0$, selected by minimum variational energy
from an earlier full-saGSD calculation.  All circuits minimize the
expectation value of
$H_{\rm def}=H+4|\widetilde\psi_0\rangle\langle\widetilde\psi_0|$
from the four-$\alpha$ reference.  The boundary-effective tUPS fabric
contains seven independent orbital rotations per repeat, since its
$L$--$X$--$L$ tiles restrict exactly to
$e^{\theta_1L}e^{\theta_2X}e^{\theta_3L}=e^{(\theta_1+\theta_3)L}$:
the $X$ factor is the identity and the adjacent rotations merge.
The repaired fabric appends $C_{1;0,2}$ to
each repeat; the tree control applies the path-ordered orbital rotations
followed by the same $C$.  The five comparisons are tUPS and tUPS+$C$ at
four and eight repeats, and Tree+$C$ at eight repeats.  Each uses three
starts with the deterministic uniform \([-0.01,0.01]\) initialization
and L-BFGS-B bounds \([-\pi,\pi]\), \(\mathtt{ftol}=10^{-12}\),
\(\mathtt{gtol}=10^{-8}\), \(\mathtt{maxls}=40\), and limits of
50,000 iterations and 100,000 joint objective--gradient calls.  Outputs
are selected by the VQD objective.  Deflator preparation is outside this
controlled single-root resource comparison.

The selected variational deflator has energy approximately
$-1.86386215833$~Ha.  Table~\ref{tab:si-direct-repair} reports the
objective-selected outputs for the five circuit settings.

Despite algebraic completion, four repaired repetitions remain near the
orbital-only plateau, whereas eight reach the energy threshold
(Table~\ref{tab:si-direct-repair}).

\begin{table}[!htbp]
\centering\footnotesize
\caption{H$_4$ root-1 direct repair with a common variational deflator.
The deflator's construction and preparation costs are outside this controlled
single-root ansatz comparison. Each row is selected from three starts by the
VQD objective and uses the common Q\textsuperscript{2}Chemistry CNOT count.}
\label{tab:si-direct-repair}
\input{benchmark/si_h4_direct_repair.tex}
\end{table}

All five objective-selected parameter/state replays pass the source and
compiled checks in the direct-repair comparison; the largest source and
compiled component residuals are $1.75\times10^{-13}$ and
$1.92\times10^{-13}$, respectively, and the largest compiled energy
difference is $4.44\times10^{-16}$~Ha.  Under the common compilation
protocol, the eight-repeat repaired fabric and Tree+$C$ each use 480 CNOTs,
while the tree gives the lower observed error. These are separate outputs
from the five-chain main comparison.

\subsection{Reference spin and degenerate-subspace diagnostics}

A direct occupation-basis CAR check of the explicit two-orbital $T(0,1)$
definition in Section~\ref{si:fixed-coordinate-vqd-methods} gives $\langle S^2\rangle=2$ and
$\langle S_z\rangle=0$.  Applying the same full-space test to the frozen
TME triplet reference gives the same values, normalization error below
$3\times10^{-16}$, and zero computed $S^2-2$ and $S_z$ action residuals.
These checks establish the spin and normalization of the reference states.

To distinguish individual eigenspace fidelity from joint coverage, all
20 saved six-root O$_2$ Tree+two-$C$, depth-two chains were inspected
without reoptimization.  The largest entrywise deviation of their Gram
matrices from identity is $1.147\times10^{-4}$; the largest overlap
between roots 1 and 2 is $1.112\times10^{-4}$.  The minimum singular
value of the two outputs projected onto the exact two-dimensional target
eigenspace is at least 0.9992229.  After orthonormalizing the two-output
span, the minimum principal-angle cosine is at least 0.9992269.
Thus these saved outputs cover two independent directions, rather than
duplicating one vector with high individual projector fidelity.
This validation concerns the main-benchmark seed-sweep point; it is
not assigned to the separate fixed-coordinate O$_2$ tests.

\FloatBarrier

%% file: benchmark/si_internal_ablation.tex
\begin{tabular*}{\textwidth}{@{\extracolsep{\fill}}llrrrr@{}}
\toprule
Task & Method & Repeats & Params. & A2A CX & Error (mHa) \\
\midrule
H$_6$ & Dense+full $X$ & 4 & 120 & 5800 & 0.059014 \\
 & Dense+one $X$ & 8 & 128 & 2192 & 0.081321 \\
 & Tree+full $X$ & 10 & 200 & 12100 & 3.03e-07 \\
 & Tree+one $X$ & 24 & 144 & 816 & 0.004571 \\
 & QNP-F & 14 & 175 & 7525 & 0.045178 \\
TME & Dense+full $X$ & 6 & 360 & 17400 & 0.001705 \\
 & Dense+one $X$ & 10 & 320 & 5480 & 0.057940 \\
 & Tree+full $X$ & 8 & 320 & 19360 & 0.042889 \\
 & Tree+one $X$ & 16 & 192 & 1088 & 0.094966 \\
 & QNP-F & 19 & 480 & 20640 & 0.000141 \\
N$_2$ & Dense+full $X$ & 6 & 180 & 8700 & 0.000352 \\
 & Dense+one $X$ & 15 & 240 & 4110 & 0.029495 \\
 & Tree+full $X$$^{*}$ & 18 & 360 & 21780 & 54.127851 \\
 & Tree+one $X$ & 30 & 180 & 1020 & 0.001888 \\
 & QNP-F & 14 & 175 & 7525 & 0.000359 \\
\bottomrule
\end{tabular*}

%% file: benchmark/si_n2_pes_metrics.tex
\begin{tabular*}{\textwidth}{@{\extracolsep{\fill}}lrrrr@{}}
\toprule
Method & Repeats & MAE & Maximum error & NPE \\
\midrule
Tree+one $X$ & 30 & 0.000948 & 0.001888 & 0.001732 \\
QNP-Q$^{*}$ & 72 & 0.783245 & 6.265837 & 6.265835 \\
tUPS & 12 & 0.000528 & 0.003800 & 0.003796 \\
tUPS (PP ref.) & 12 & 0.000570 & 0.001565 & 0.001538 \\
\bottomrule
\end{tabular*}

%% file: benchmark/si_h4_direct_repair.tex
\begin{tabular*}{\textwidth}{@{\extracolsep{\fill}}lrrrrr@{}}
\toprule
Circuit & Repeats & Params. & Error (mHa) & Fidelity & A2A CX \\
\midrule
Tree+$C$ & 8 & 64 & 0.004229 & 0.999997 & 480 \\
Effective tUPS & 4 & 28 & 9.671224 & 0.327324 & 112 \\
Effective tUPS & 8 & 56 & 9.671224 & 0.327323 & 224 \\
Effective tUPS+$C$ & 4 & 32 & 9.643108 & 0.327299 & 240 \\
Effective tUPS+$C$ & 8 & 64 & 0.024221 & 0.999983 & 480 \\
\bottomrule
\end{tabular*}

%% file: main_and_si.bbl
\begin{thebibliography}{30}%
\makeatletter
\providecommand \@ifxundefined [1]{%
 \@ifx{#1\undefined}
}%
\providecommand \@ifnum [1]{%
 \ifnum #1\expandafter \@firstoftwo
 \else \expandafter \@secondoftwo
 \fi
}%
\providecommand \@ifx [1]{%
 \ifx #1\expandafter \@firstoftwo
 \else \expandafter \@secondoftwo
 \fi
}%
\providecommand \natexlab [1]{#1}%
\providecommand \enquote  [1]{``#1''}%
\providecommand \bibnamefont  [1]{#1}%
\providecommand \bibfnamefont [1]{#1}%
\providecommand \citenamefont [1]{#1}%
\providecommand \href@noop [0]{\@secondoftwo}%
\providecommand \href [0]{\begingroup \@sanitize@url \@href}%
\providecommand \@href[1]{\@@startlink{#1}\@@href}%
\providecommand \@@href[1]{\endgroup#1\@@endlink}%
\providecommand \@sanitize@url [0]{\catcode `\\12\catcode `\$12\catcode
  `\&12\catcode `\#12\catcode `\^12\catcode `\_12\catcode `\%12\relax}%
\providecommand \@@startlink[1]{}%
\providecommand \@@endlink[0]{}%
\providecommand \url  [0]{\begingroup\@sanitize@url \@url }%
\providecommand \@url [1]{\endgroup\@href {#1}{\urlprefix }}%
\providecommand \urlprefix  [0]{URL }%
\providecommand \Eprint [0]{\href }%
\providecommand \doibase [0]{https://doi.org/}%
\providecommand \selectlanguage [0]{\@gobble}%
\providecommand \bibinfo  [0]{\@secondoftwo}%
\providecommand \bibfield  [0]{\@secondoftwo}%
\providecommand \translation [1]{[#1]}%
\providecommand \BibitemOpen [0]{}%
\providecommand \bibitemStop [0]{}%
\providecommand \bibitemNoStop [0]{.\EOS\space}%
\providecommand \EOS [0]{\spacefactor3000\relax}%
\providecommand \BibitemShut  [1]{\csname bibitem#1\endcsname}%
\let\auto@bib@innerbib\@empty
\bibitem [{\citenamefont {Gard}\ \emph {et~al.}(2020)\citenamefont {Gard},
  \citenamefont {Zhu}, \citenamefont {Barron}, \citenamefont {Mayhall},
  \citenamefont {Economou},\ and\ \citenamefont {Barnes}}]{Gard2020}%
  \BibitemOpen
  \bibfield  {author} {\bibinfo {author} {\bibfnamefont {B.~T.}\ \bibnamefont
  {Gard}}, \bibinfo {author} {\bibfnamefont {L.}~\bibnamefont {Zhu}}, \bibinfo
  {author} {\bibfnamefont {G.~S.}\ \bibnamefont {Barron}}, \bibinfo {author}
  {\bibfnamefont {N.~J.}\ \bibnamefont {Mayhall}}, \bibinfo {author}
  {\bibfnamefont {S.~E.}\ \bibnamefont {Economou}},\ and\ \bibinfo {author}
  {\bibfnamefont {E.}~\bibnamefont {Barnes}},\ }\bibfield  {title} {\bibinfo
  {title} {Efficient symmetry-preserving state preparation circuits for the
  variational quantum eigensolver algorithm},\ }\href
  {https://doi.org/10.1038/s41534-019-0240-1} {\bibfield  {journal} {\bibinfo
  {journal} {npj Quantum Information}\ }\textbf {\bibinfo {volume} {6}},\
  \bibinfo {pages} {10} (\bibinfo {year} {2020})}\BibitemShut {NoStop}%
\bibitem [{\citenamefont {Higgott}\ \emph {et~al.}(2019)\citenamefont
  {Higgott}, \citenamefont {Wang},\ and\ \citenamefont
  {Brierley}}]{Higgott2019VQD}%
  \BibitemOpen
  \bibfield  {author} {\bibinfo {author} {\bibfnamefont {O.}~\bibnamefont
  {Higgott}}, \bibinfo {author} {\bibfnamefont {D.}~\bibnamefont {Wang}},\ and\
  \bibinfo {author} {\bibfnamefont {S.}~\bibnamefont {Brierley}},\ }\bibfield
  {title} {\bibinfo {title} {Variational quantum computation of excited
  states},\ }\href {https://doi.org/10.22331/q-2019-07-01-156} {\bibfield
  {journal} {\bibinfo  {journal} {Quantum}\ }\textbf {\bibinfo {volume} {3}},\
  \bibinfo {pages} {156} (\bibinfo {year} {2019})}\BibitemShut {NoStop}%
\bibitem [{\citenamefont {Helgaker}\ \emph {et~al.}(2000)\citenamefont
  {Helgaker}, \citenamefont {J{\o}rgensen},\ and\ \citenamefont
  {Olsen}}]{Helgaker2000}%
  \BibitemOpen
  \bibfield  {author} {\bibinfo {author} {\bibfnamefont {T.}~\bibnamefont
  {Helgaker}}, \bibinfo {author} {\bibfnamefont {P.}~\bibnamefont
  {J{\o}rgensen}},\ and\ \bibinfo {author} {\bibfnamefont {J.}~\bibnamefont
  {Olsen}},\ }\href {https://doi.org/10.1002/9781119019572} {\emph {\bibinfo
  {title} {Molecular Electronic-Structure Theory}}}\ (\bibinfo  {publisher}
  {Wiley},\ \bibinfo {address} {Chichester},\ \bibinfo {year}
  {2000})\BibitemShut {NoStop}%
\bibitem [{\citenamefont {Evangelista}\ \emph {et~al.}(2019)\citenamefont
  {Evangelista}, \citenamefont {Chan},\ and\ \citenamefont
  {Scuseria}}]{Evangelista2019UCC}%
  \BibitemOpen
  \bibfield  {author} {\bibinfo {author} {\bibfnamefont {F.~A.}\ \bibnamefont
  {Evangelista}}, \bibinfo {author} {\bibfnamefont {G.~K.-L.}\ \bibnamefont
  {Chan}},\ and\ \bibinfo {author} {\bibfnamefont {G.~E.}\ \bibnamefont
  {Scuseria}},\ }\bibfield  {title} {\bibinfo {title} {Exact parameterization
  of fermionic wave functions via unitary coupled cluster theory},\ }\href
  {https://doi.org/10.1063/1.5133059} {\bibfield  {journal} {\bibinfo
  {journal} {The Journal of Chemical Physics}\ }\textbf {\bibinfo {volume}
  {151}},\ \bibinfo {pages} {244112} (\bibinfo {year} {2019})}\BibitemShut
  {NoStop}%
\bibitem [{\citenamefont {Burton}\ \emph {et~al.}(2023)\citenamefont {Burton},
  \citenamefont {Marti-Dafcik}, \citenamefont {Tew},\ and\ \citenamefont
  {Wales}}]{Burton2023Exact}%
  \BibitemOpen
  \bibfield  {author} {\bibinfo {author} {\bibfnamefont {H.~G.~A.}\
  \bibnamefont {Burton}}, \bibinfo {author} {\bibfnamefont {D.}~\bibnamefont
  {Marti-Dafcik}}, \bibinfo {author} {\bibfnamefont {D.~P.}\ \bibnamefont
  {Tew}},\ and\ \bibinfo {author} {\bibfnamefont {D.~J.}\ \bibnamefont
  {Wales}},\ }\bibfield  {title} {\bibinfo {title} {Exact electronic states
  with shallow quantum circuits from global optimisation},\ }\href
  {https://doi.org/10.1038/s41534-023-00744-2} {\bibfield  {journal} {\bibinfo
  {journal} {npj Quantum Information}\ }\textbf {\bibinfo {volume} {9}},\
  \bibinfo {pages} {75} (\bibinfo {year} {2023})}\BibitemShut {NoStop}%
\bibitem [{\citenamefont {Magoulas}\ and\ \citenamefont
  {Evangelista}(2025)}]{Magoulas2025SpinAdapted}%
  \BibitemOpen
  \bibfield  {author} {\bibinfo {author} {\bibfnamefont {I.}~\bibnamefont
  {Magoulas}}\ and\ \bibinfo {author} {\bibfnamefont {F.~A.}\ \bibnamefont
  {Evangelista}},\ }\href {https://doi.org/10.48550/arXiv.2511.13485} {\bibinfo
  {title} {Spin-adapted fermionic unitaries: From {Lie} algebras to compact
  quantum circuits}} (\bibinfo {year} {2025}),\ \Eprint
  {https://arxiv.org/abs/2511.13485v2} {arXiv:2511.13485v2 [quant-ph]}
  \BibitemShut {NoStop}%
\bibitem [{\citenamefont {D'Alessandro}(2007)}]{DAlessandro2007}%
  \BibitemOpen
  \bibfield  {author} {\bibinfo {author} {\bibfnamefont {D.}~\bibnamefont
  {D'Alessandro}},\ }\href {https://doi.org/10.1201/9781584888833} {\emph
  {\bibinfo {title} {Introduction to Quantum Control and Dynamics}}}\ (\bibinfo
   {publisher} {Chapman and Hall/CRC},\ \bibinfo {address} {Boca Raton},\
  \bibinfo {year} {2007})\BibitemShut {NoStop}%
\bibitem [{\citenamefont {Anselmetti}\ \emph {et~al.}(2021)\citenamefont
  {Anselmetti}, \citenamefont {Wierichs}, \citenamefont {Gogolin},\ and\
  \citenamefont {Parrish}}]{Anselmetti2021}%
  \BibitemOpen
  \bibfield  {author} {\bibinfo {author} {\bibfnamefont {G.-L.~R.}\
  \bibnamefont {Anselmetti}}, \bibinfo {author} {\bibfnamefont
  {D.}~\bibnamefont {Wierichs}}, \bibinfo {author} {\bibfnamefont
  {C.}~\bibnamefont {Gogolin}},\ and\ \bibinfo {author} {\bibfnamefont {R.~M.}\
  \bibnamefont {Parrish}},\ }\bibfield  {title} {\bibinfo {title} {Local,
  expressive, quantum-number-preserving {VQE} ans{\"a}tze for fermionic
  systems},\ }\href {https://doi.org/10.1088/1367-2630/ac2cb3} {\bibfield
  {journal} {\bibinfo  {journal} {New Journal of Physics}\ }\textbf {\bibinfo
  {volume} {23}},\ \bibinfo {pages} {113010} (\bibinfo {year}
  {2021})}\BibitemShut {NoStop}%
\bibitem [{Sup()}]{SupplementalMaterial}%
  \BibitemOpen
  \href@noop {} {}\bibinfo {note} {See the Supplemental Material appended to
  this manuscript for detailed operator comparisons, additional mathematical
  results, extended numerical checks, and circuit-replay
  diagnostics}\BibitemShut {NoStop}%
\bibitem [{\citenamefont {Lee}\ \emph {et~al.}(2019)\citenamefont {Lee},
  \citenamefont {Huggins}, \citenamefont {Head-Gordon},\ and\ \citenamefont
  {Whaley}}]{Lee2019UpCCGSD}%
  \BibitemOpen
  \bibfield  {author} {\bibinfo {author} {\bibfnamefont {J.}~\bibnamefont
  {Lee}}, \bibinfo {author} {\bibfnamefont {W.~J.}\ \bibnamefont {Huggins}},
  \bibinfo {author} {\bibfnamefont {M.}~\bibnamefont {Head-Gordon}},\ and\
  \bibinfo {author} {\bibfnamefont {K.~B.}\ \bibnamefont {Whaley}},\ }\bibfield
   {title} {\bibinfo {title} {Generalized unitary coupled cluster wave
  functions for quantum computation},\ }\href
  {https://doi.org/10.1021/acs.jctc.8b01004} {\bibfield  {journal} {\bibinfo
  {journal} {Journal of Chemical Theory and Computation}\ }\textbf {\bibinfo
  {volume} {15}},\ \bibinfo {pages} {311} (\bibinfo {year} {2019})}\BibitemShut
  {NoStop}%
\bibitem [{\citenamefont {Burton}(2024)}]{Burton2024TUPS}%
  \BibitemOpen
  \bibfield  {author} {\bibinfo {author} {\bibfnamefont {H.~G.~A.}\
  \bibnamefont {Burton}},\ }\bibfield  {title} {\bibinfo {title} {Accurate and
  gate-efficient quantum ans{\"a}tze for electronic states without adaptive
  optimization},\ }\href {https://doi.org/10.1103/PhysRevResearch.6.023300}
  {\bibfield  {journal} {\bibinfo  {journal} {Physical Review Research}\
  }\textbf {\bibinfo {volume} {6}},\ \bibinfo {pages} {023300} (\bibinfo {year}
  {2024})}\BibitemShut {NoStop}%
\bibitem [{\citenamefont {Magoulas}\ \emph {et~al.}(2026)\citenamefont
  {Magoulas}, \citenamefont {Zhang},\ and\ \citenamefont
  {Evangelista}}]{Magoulas2026Dilemmas}%
  \BibitemOpen
  \bibfield  {author} {\bibinfo {author} {\bibfnamefont {I.}~\bibnamefont
  {Magoulas}}, \bibinfo {author} {\bibfnamefont {M.}~\bibnamefont {Zhang}},\
  and\ \bibinfo {author} {\bibfnamefont {F.~A.}\ \bibnamefont {Evangelista}},\
  }\bibfield  {title} {\bibinfo {title} {Symmetry dilemmas in quantum computing
  for chemistry: A comprehensive analysis},\ }\href
  {https://doi.org/10.1063/5.0316482} {\bibfield  {journal} {\bibinfo
  {journal} {The Journal of Chemical Physics}\ }\textbf {\bibinfo {volume}
  {164}},\ \bibinfo {pages} {144113} (\bibinfo {year} {2026})}\BibitemShut
  {NoStop}%
\bibitem [{\citenamefont {Stergiou}\ and\ \citenamefont
  {Sawaya}(2026)}]{Stergiou2026}%
  \BibitemOpen
  \bibfield  {author} {\bibinfo {author} {\bibfnamefont {A.}~\bibnamefont
  {Stergiou}}\ and\ \bibinfo {author} {\bibfnamefont {N.~P.~D.}\ \bibnamefont
  {Sawaya}},\ }\href {https://doi.org/10.48550/arXiv.2605.00979} {\bibinfo
  {title} {Universality of quantum gates in particle and symmetry constrained
  subspaces}} (\bibinfo {year} {2026}),\ \Eprint
  {https://arxiv.org/abs/2605.00979v1} {arXiv:2605.00979v1 [quant-ph]}
  \BibitemShut {NoStop}%
\bibitem [{\citenamefont {Li}\ and\ \citenamefont
  {Paldus}(1993)}]{LiPaldus1993}%
  \BibitemOpen
  \bibfield  {author} {\bibinfo {author} {\bibfnamefont {X.}~\bibnamefont
  {Li}}\ and\ \bibinfo {author} {\bibfnamefont {J.}~\bibnamefont {Paldus}},\
  }\bibfield  {title} {\bibinfo {title} {Unitary group tensor operator algebras
  for many-electron systems: {II}. one- and two-body matrix elements},\ }\href
  {https://doi.org/10.1007/BF01165571} {\bibfield  {journal} {\bibinfo
  {journal} {Journal of Mathematical Chemistry}\ }\textbf {\bibinfo {volume}
  {13}},\ \bibinfo {pages} {273} (\bibinfo {year} {1993})}\BibitemShut
  {NoStop}%
\bibitem [{\citenamefont {Jain}\ \emph {et~al.}(2026)\citenamefont {Jain},
  \citenamefont {Izmaylov},\ and\ \citenamefont {Kjellgren}}]{Jain2026Exact}%
  \BibitemOpen
  \bibfield  {author} {\bibinfo {author} {\bibfnamefont {P.}~\bibnamefont
  {Jain}}, \bibinfo {author} {\bibfnamefont {A.~F.}\ \bibnamefont {Izmaylov}},\
  and\ \bibinfo {author} {\bibfnamefont {E.~R.}\ \bibnamefont {Kjellgren}},\
  }\bibfield  {title} {\bibinfo {title} {Exact factorization of unitary
  transformations with spin-adapted generators},\ }\href
  {https://doi.org/10.1063/5.0326865} {\bibfield  {journal} {\bibinfo
  {journal} {The Journal of Chemical Physics}\ }\textbf {\bibinfo {volume}
  {164}},\ \bibinfo {pages} {194107} (\bibinfo {year} {2026})}\BibitemShut
  {NoStop}%
\bibitem [{\citenamefont {Avetisyan}\ and\ \citenamefont
  {Mkrtchyan}(2020)}]{Avetisyan2020X2}%
  \BibitemOpen
  \bibfield  {author} {\bibinfo {author} {\bibfnamefont {M.~Y.}\ \bibnamefont
  {Avetisyan}}\ and\ \bibinfo {author} {\bibfnamefont {R.~L.}\ \bibnamefont
  {Mkrtchyan}},\ }\bibfield  {title} {\bibinfo {title} {\({X}_2\) series of
  universal quantum dimensions},\ }\href
  {https://doi.org/10.1088/1751-8121/ab5f4d} {\bibfield  {journal} {\bibinfo
  {journal} {Journal of Physics A: Mathematical and Theoretical}\ }\textbf
  {\bibinfo {volume} {53}},\ \bibinfo {pages} {045202} (\bibinfo {year}
  {2020})}\BibitemShut {NoStop}%
\bibitem [{\citenamefont {Arrazola}\ \emph {et~al.}(2022)\citenamefont
  {Arrazola}, \citenamefont {{Di Matteo}}, \citenamefont {Quesada},
  \citenamefont {Jahangiri}, \citenamefont {Delgado},\ and\ \citenamefont
  {Killoran}}]{Arrazola2022}%
  \BibitemOpen
  \bibfield  {author} {\bibinfo {author} {\bibfnamefont {J.~M.}\ \bibnamefont
  {Arrazola}}, \bibinfo {author} {\bibfnamefont {O.}~\bibnamefont {{Di
  Matteo}}}, \bibinfo {author} {\bibfnamefont {N.}~\bibnamefont {Quesada}},
  \bibinfo {author} {\bibfnamefont {S.}~\bibnamefont {Jahangiri}}, \bibinfo
  {author} {\bibfnamefont {A.}~\bibnamefont {Delgado}},\ and\ \bibinfo {author}
  {\bibfnamefont {N.}~\bibnamefont {Killoran}},\ }\bibfield  {title} {\bibinfo
  {title} {Universal quantum circuits for quantum chemistry},\ }\href
  {https://doi.org/10.22331/q-2022-06-20-742} {\bibfield  {journal} {\bibinfo
  {journal} {Quantum}\ }\textbf {\bibinfo {volume} {6}},\ \bibinfo {pages}
  {742} (\bibinfo {year} {2022})}\BibitemShut {NoStop}%
\bibitem [{\citenamefont {Jordan}\ and\ \citenamefont
  {Wigner}(1928)}]{JordanWigner1928}%
  \BibitemOpen
  \bibfield  {author} {\bibinfo {author} {\bibfnamefont {P.}~\bibnamefont
  {Jordan}}\ and\ \bibinfo {author} {\bibfnamefont {E.}~\bibnamefont
  {Wigner}},\ }\bibfield  {title} {\bibinfo {title} {{{\"U}ber das Paulische
  {\"A}quivalenzverbot}},\ }\href {https://doi.org/10.1007/BF01331938}
  {\bibfield  {journal} {\bibinfo  {journal} {Zeitschrift f{\"u}r Physik}\
  }\textbf {\bibinfo {volume} {47}},\ \bibinfo {pages} {631} (\bibinfo {year}
  {1928})}\BibitemShut {NoStop}%
\bibitem [{\citenamefont {Liu}(2026)}]{Liu2026Zenodo}%
  \BibitemOpen
  \bibfield  {author} {\bibinfo {author} {\bibfnamefont {M.}~\bibnamefont
  {Liu}},\ }\href {https://doi.org/10.5281/zenodo.22230014} {\bibinfo {title}
  {Formal verification, numerical data, and reproducibility code for ``minimal
  building blocks for molecular quantum circuits with exact spin symmetry''}}
  (\bibinfo {year} {2026})\BibitemShut {NoStop}%
\bibitem [{\citenamefont {Javadi-Abhari}\ \emph {et~al.}(2024)\citenamefont
  {Javadi-Abhari}, \citenamefont {Treinish}, \citenamefont {Krsulich},
  \citenamefont {Wood}, \citenamefont {Lishman}, \citenamefont {Gacon},
  \citenamefont {Martiel}, \citenamefont {Nation}, \citenamefont {Bishop},
  \citenamefont {Cross}, \citenamefont {Johnson},\ and\ \citenamefont
  {Gambetta}}]{Qiskit2024}%
  \BibitemOpen
  \bibfield  {author} {\bibinfo {author} {\bibfnamefont {A.}~\bibnamefont
  {Javadi-Abhari}}, \bibinfo {author} {\bibfnamefont {M.}~\bibnamefont
  {Treinish}}, \bibinfo {author} {\bibfnamefont {K.}~\bibnamefont {Krsulich}},
  \bibinfo {author} {\bibfnamefont {C.~J.}\ \bibnamefont {Wood}}, \bibinfo
  {author} {\bibfnamefont {J.}~\bibnamefont {Lishman}}, \bibinfo {author}
  {\bibfnamefont {J.}~\bibnamefont {Gacon}}, \bibinfo {author} {\bibfnamefont
  {S.}~\bibnamefont {Martiel}}, \bibinfo {author} {\bibfnamefont {P.~D.}\
  \bibnamefont {Nation}}, \bibinfo {author} {\bibfnamefont {L.~S.}\
  \bibnamefont {Bishop}}, \bibinfo {author} {\bibfnamefont {A.~W.}\
  \bibnamefont {Cross}}, \bibinfo {author} {\bibfnamefont {B.~R.}\ \bibnamefont
  {Johnson}},\ and\ \bibinfo {author} {\bibfnamefont {J.~M.}\ \bibnamefont
  {Gambetta}},\ }\href {https://doi.org/10.48550/arXiv.2405.08810} {\bibinfo
  {title} {Quantum computing with {Qiskit}}} (\bibinfo {year} {2024}),\ \Eprint
  {https://arxiv.org/abs/2405.08810} {arXiv:2405.08810 [quant-ph]} \BibitemShut
  {NoStop}%
\bibitem [{\citenamefont {Fan}\ \emph {et~al.}(2022)\citenamefont {Fan},
  \citenamefont {Liu}, \citenamefont {Zeng}, \citenamefont {Xu}, \citenamefont
  {Shang}, \citenamefont {Li},\ and\ \citenamefont
  {Yang}}]{Fan2022Q2Chemistry}%
  \BibitemOpen
  \bibfield  {author} {\bibinfo {author} {\bibfnamefont {Y.}~\bibnamefont
  {Fan}}, \bibinfo {author} {\bibfnamefont {J.}~\bibnamefont {Liu}}, \bibinfo
  {author} {\bibfnamefont {X.}~\bibnamefont {Zeng}}, \bibinfo {author}
  {\bibfnamefont {Z.}~\bibnamefont {Xu}}, \bibinfo {author} {\bibfnamefont
  {H.}~\bibnamefont {Shang}}, \bibinfo {author} {\bibfnamefont
  {Z.}~\bibnamefont {Li}},\ and\ \bibinfo {author} {\bibfnamefont
  {J.}~\bibnamefont {Yang}},\ }\bibfield  {title} {\bibinfo {title}
  {{Q\textsuperscript{2}Chemistry}: A quantum computation platform for quantum
  chemistry},\ }\href {https://doi.org/10.52396/JUSTC-2022-0118} {\bibfield
  {journal} {\bibinfo  {journal} {JUSTC}\ }\textbf {\bibinfo {volume} {52}},\
  \bibinfo {pages} {2} (\bibinfo {year} {2022})}\BibitemShut {NoStop}%
\bibitem [{\citenamefont {Oszmaniec}\ and\ \citenamefont
  {Zimbor{\'a}s}(2017)}]{Oszmaniec2017}%
  \BibitemOpen
  \bibfield  {author} {\bibinfo {author} {\bibfnamefont {M.}~\bibnamefont
  {Oszmaniec}}\ and\ \bibinfo {author} {\bibfnamefont {Z.}~\bibnamefont
  {Zimbor{\'a}s}},\ }\bibfield  {title} {\bibinfo {title} {Universal extensions
  of restricted classes of quantum operations},\ }\href
  {https://doi.org/10.1103/PhysRevLett.119.220502} {\bibfield  {journal}
  {\bibinfo  {journal} {Physical Review Letters}\ }\textbf {\bibinfo {volume}
  {119}},\ \bibinfo {pages} {220502} (\bibinfo {year} {2017})}\BibitemShut
  {NoStop}%
\bibitem [{\citenamefont {Fulton}\ and\ \citenamefont
  {Harris}(1991)}]{FultonHarris1991}%
  \BibitemOpen
  \bibfield  {author} {\bibinfo {author} {\bibfnamefont {W.}~\bibnamefont
  {Fulton}}\ and\ \bibinfo {author} {\bibfnamefont {J.}~\bibnamefont
  {Harris}},\ }\href {https://doi.org/10.1007/978-1-4612-0979-9} {\emph
  {\bibinfo {title} {Representation Theory: A First Course}}},\ \bibinfo
  {series} {Graduate Texts in Mathematics}, Vol.\ \bibinfo {volume} {129}\
  (\bibinfo  {publisher} {Springer},\ \bibinfo {address} {New York},\ \bibinfo
  {year} {1991})\BibitemShut {NoStop}%
\bibitem [{\citenamefont {{de Moura}}\ and\ \citenamefont
  {Ullrich}(2021)}]{deMoura2021Lean4}%
  \BibitemOpen
  \bibfield  {author} {\bibinfo {author} {\bibfnamefont {L.}~\bibnamefont {{de
  Moura}}}\ and\ \bibinfo {author} {\bibfnamefont {S.}~\bibnamefont
  {Ullrich}},\ }\bibfield  {title} {\bibinfo {title} {The {Lean 4} theorem
  prover and programming language},\ }in\ \href
  {https://doi.org/10.1007/978-3-030-79876-5_37} {\emph {\bibinfo {booktitle}
  {Automated Deduction---CADE 28}}},\ \bibinfo {series} {Lecture Notes in
  Computer Science}, Vol.\ \bibinfo {volume} {12699}\ (\bibinfo  {publisher}
  {Springer},\ \bibinfo {year} {2021})\ pp.\ \bibinfo {pages}
  {625--635}\BibitemShut {NoStop}%
\bibitem [{\citenamefont {{The mathlib Community}}(2020)}]{Mathlib2020}%
  \BibitemOpen
  \bibfield  {author} {\bibinfo {author} {\bibnamefont {{The mathlib
  Community}}},\ }\bibfield  {title} {\bibinfo {title} {The {Lean} mathematical
  library},\ }in\ \href {https://doi.org/10.1145/3372885.3373824} {\emph
  {\bibinfo {booktitle} {Proceedings of the 9th ACM SIGPLAN International
  Conference on Certified Programs and Proofs}}}\ (\bibinfo  {publisher}
  {ACM},\ \bibinfo {year} {2020})\ pp.\ \bibinfo {pages} {367--381}\BibitemShut
  {NoStop}%
\bibitem [{\citenamefont {Tang}\ \emph {et~al.}(2021)\citenamefont {Tang},
  \citenamefont {Shkolnikov}, \citenamefont {Barron}, \citenamefont {Grimsley},
  \citenamefont {Mayhall}, \citenamefont {Barnes},\ and\ \citenamefont
  {Economou}}]{Tang2021}%
  \BibitemOpen
  \bibfield  {author} {\bibinfo {author} {\bibfnamefont {H.~L.}\ \bibnamefont
  {Tang}}, \bibinfo {author} {\bibfnamefont {V.~O.}\ \bibnamefont
  {Shkolnikov}}, \bibinfo {author} {\bibfnamefont {G.~S.}\ \bibnamefont
  {Barron}}, \bibinfo {author} {\bibfnamefont {H.~R.}\ \bibnamefont
  {Grimsley}}, \bibinfo {author} {\bibfnamefont {N.~J.}\ \bibnamefont
  {Mayhall}}, \bibinfo {author} {\bibfnamefont {E.}~\bibnamefont {Barnes}},\
  and\ \bibinfo {author} {\bibfnamefont {S.~E.}\ \bibnamefont {Economou}},\
  }\bibfield  {title} {\bibinfo {title} {{Qubit-ADAPT-VQE}: An adaptive
  algorithm for constructing hardware-efficient ans{\"a}tze on a quantum
  processor},\ }\href {https://doi.org/10.1103/PRXQuantum.2.020310} {\bibfield
  {journal} {\bibinfo  {journal} {PRX Quantum}\ }\textbf {\bibinfo {volume}
  {2}},\ \bibinfo {pages} {020310} (\bibinfo {year} {2021})}\BibitemShut
  {NoStop}%
\bibitem [{\citenamefont {Shkolnikov}\ \emph {et~al.}(2023)\citenamefont
  {Shkolnikov}, \citenamefont {Mayhall}, \citenamefont {Economou},\ and\
  \citenamefont {Barnes}}]{Shkolnikov2023}%
  \BibitemOpen
  \bibfield  {author} {\bibinfo {author} {\bibfnamefont {V.~O.}\ \bibnamefont
  {Shkolnikov}}, \bibinfo {author} {\bibfnamefont {N.~J.}\ \bibnamefont
  {Mayhall}}, \bibinfo {author} {\bibfnamefont {S.~E.}\ \bibnamefont
  {Economou}},\ and\ \bibinfo {author} {\bibfnamefont {E.}~\bibnamefont
  {Barnes}},\ }\bibfield  {title} {\bibinfo {title} {Avoiding symmetry
  roadblocks and minimizing the measurement overhead of adaptive variational
  quantum eigensolvers},\ }\href {https://doi.org/10.22331/q-2023-06-12-1040}
  {\bibfield  {journal} {\bibinfo  {journal} {Quantum}\ }\textbf {\bibinfo
  {volume} {7}},\ \bibinfo {pages} {1040} (\bibinfo {year} {2023})}\BibitemShut
  {NoStop}%
\bibitem [{\citenamefont {Haidar}\ \emph {et~al.}(2025)\citenamefont {Haidar},
  \citenamefont {Adjoua}, \citenamefont {Badreddine}, \citenamefont {Peruzzo},\
  and\ \citenamefont {Piquemal}}]{Haidar2025}%
  \BibitemOpen
  \bibfield  {author} {\bibinfo {author} {\bibfnamefont {M.}~\bibnamefont
  {Haidar}}, \bibinfo {author} {\bibfnamefont {O.}~\bibnamefont {Adjoua}},
  \bibinfo {author} {\bibfnamefont {S.}~\bibnamefont {Badreddine}}, \bibinfo
  {author} {\bibfnamefont {A.}~\bibnamefont {Peruzzo}},\ and\ \bibinfo {author}
  {\bibfnamefont {J.-P.}\ \bibnamefont {Piquemal}},\ }\bibfield  {title}
  {\bibinfo {title} {Non-iterative disentangled unitary coupled-cluster based
  on {Lie}-algebraic structure},\ }\href
  {https://doi.org/10.1088/2058-9565/adb3c5} {\bibfield  {journal} {\bibinfo
  {journal} {Quantum Science and Technology}\ }\textbf {\bibinfo {volume}
  {10}},\ \bibinfo {pages} {025031} (\bibinfo {year} {2025})}\BibitemShut
  {NoStop}%
\bibitem [{\citenamefont {Viswanathan}\ \emph {et~al.}(2026)\citenamefont
  {Viswanathan}, \citenamefont {Adjoua}, \citenamefont {Feniou}, \citenamefont
  {Badreddine},\ and\ \citenamefont {Piquemal}}]{Viswanathan2026}%
  \BibitemOpen
  \bibfield  {author} {\bibinfo {author} {\bibfnamefont {Y.}~\bibnamefont
  {Viswanathan}}, \bibinfo {author} {\bibfnamefont {O.}~\bibnamefont {Adjoua}},
  \bibinfo {author} {\bibfnamefont {C.}~\bibnamefont {Feniou}}, \bibinfo
  {author} {\bibfnamefont {S.}~\bibnamefont {Badreddine}},\ and\ \bibinfo
  {author} {\bibfnamefont {J.-P.}\ \bibnamefont {Piquemal}},\ }\href
  {https://doi.org/10.48550/arXiv.2511.22593} {\bibinfo {title} {An optimized
  construction of {Lie} algebra generator pools for variational quantum
  eigensolvers in chemistry}} (\bibinfo {year} {2026}),\ \Eprint
  {https://arxiv.org/abs/2511.22593v3} {arXiv:2511.22593v3 [quant-ph]}
  \BibitemShut {NoStop}%
\bibitem [{\citenamefont {Sapova}\ and\ \citenamefont
  {Fedorov}(2022)}]{Sapova2022}%
  \BibitemOpen
  \bibfield  {author} {\bibinfo {author} {\bibfnamefont {M.~D.}\ \bibnamefont
  {Sapova}}\ and\ \bibinfo {author} {\bibfnamefont {A.~K.}\ \bibnamefont
  {Fedorov}},\ }\bibfield  {title} {\bibinfo {title} {Variational quantum
  eigensolver techniques for simulating carbon monoxide oxidation},\ }\href
  {https://doi.org/10.1038/s42005-022-00982-4} {\bibfield  {journal} {\bibinfo
  {journal} {Communications Physics}\ }\textbf {\bibinfo {volume} {5}},\
  \bibinfo {pages} {199} (\bibinfo {year} {2022})}\BibitemShut {NoStop}%
\end{thebibliography}
